\RequirePackage{ifpdf,xcolor,tikz}

\documentclass[12pt,letterpaper]{article}
\pdfoutput=1
\usepackage{ifpdf}
\usepackage[a4paper,top=3.2cm,width=17cm, bottom=3.2cm, left=2cm]{geometry}
\usepackage{graphicx}
\usepackage{microtype}
\usepackage[utf8x]{inputenc}
\usepackage[T1]{fontenc}
\usepackage{ae}
\usepackage{aecompl}
\usepackage{ytableau}
\usepackage{hyperref}
\usepackage{amsthm}
\usepackage{amssymb}
\usepackage{fullpage}
\usepackage{amsmath}
\usepackage{nccmath}
\usepackage{cleveref}
\usepackage{authblk}
\usepackage{multicol}
\usepackage{multirow}
\usepackage{amsmath}
\usepackage{ dsfont }
\usepackage{frcursive}
\usepackage{arydshln}
\usepackage{wrapfig}

\usepackage{chngcntr}
\counterwithin{equation}{section}

\usepackage{hyperref}
\hypersetup{
    colorlinks=true,
    linkcolor=blue,
    filecolor=magenta,      
    urlcolor=cyan,
    citecolor=blue
}
\usepackage{tikz}
\usetikzlibrary{braids}
\usepackage{quiver}

\usepackage{comment}

\usepackage{tikz-cd}
\usepackage{tkz-graph}
\tikzset{
  LabelStyle/.style = { rectangle, rounded corners, draw,
                        minimum width = 1em, fill = yellow!50,
                        text = red, font = \bfseries },
  VertexStyle/.append style = { inner sep=1pt,
                                font = \small},
  EdgeStyle/.append style = {->, bend left} }

\usepackage{arydshln}
\usepackage{fdsymbol}
\usepackage{faktor}
\usepackage{lipsum}
\usepackage{adjustbox}
\usepackage{dynkin-diagrams}

\usepackage{mathrsfs}

\usepackage{appendix}

\makeatletter
\newcommand*\dashline{\rotatebox[origin=c]{90}{$\dabar@\dabar@\dabar@$}}
\makeatother

\newtheorem{thm}{Theorem}[section]
\newtheorem{prp}[thm]{Proposition}

\newtheorem{lma}[thm]{Lemma}
\newtheorem{rmk}[thm]{Remark}
\newtheorem{dfn}[thm]{Definition}

\theoremstyle{definition}
\newtheorem{exmp}{Example}[section]

\definecolor{amber}{rgb}{1.0, 0.49, 0.0}
\definecolor{Green}{rgb}{0.0, 0.5, 0.0}
\definecolor{purple}{rgb}{0.7,0,0.7}

\DeclareUnicodeCharacter{2212}{\ensuremath{-}}

\usepackage{authblk} 
\usepackage{orcidlink}

\title{Unlinking symmetric quivers}

\author[1]{Piotr Kucharski \orcidlink{0000-0002-9599-5658}\thanks{piotr.kucharski@mimuw.edu.pl}}

\author[2]{H\'{e}lder Larragu\'{i}vel \orcidlink{0000-0002-2894-5052}\thanks{helder.larraguivel@uj.edu.pl}}

\author[3]{\break Dmitry Noshchenko \orcidlink{0000-0002-9639-5603}\thanks{dsnoshchenko@stp.dias.ie}}

\author[4]{Piotr Su{\l}kowski \orcidlink{0000-0002-6176-6240}\thanks{psulkows@fuw.edu.pl}}

\affil[1]{Institute of Mathematics, University of Warsaw, ul. Banacha 2, 02-097 Warsaw, Poland}
\affil[2]{Institute of Theoretical Physics and Mark Kac Center
for Complex Systems Research, Jagiellonian University, ul. {\L}ojasiewicza 11, 30-348 Krak\'{o}w, Poland}
\affil[3]{School of Theoretical Physics, Dublin Institute for Advanced Studies, 10 Burlington Road, Dublin 4, D04 C932, Ireland}
\affil[4]{Faculty of Physics, University of Warsaw, ul. Pasteura 5, 02-093 Warsaw, Poland}

\date{\today}

\begin{document}

\maketitle

\begin{abstract}
We analyse the structure of equivalence classes of symmetric quivers whose generating series are equal. We consider such classes constructed using the basic operation of unlinking, which increases a~size of a~quiver. The existence and features of such classes do not depend on a~particular quiver but follow from the properties of unlinking. We show that such classes include sets of quivers assembled into permutohedra, and all quivers in a~given class are determined by one quiver of the largest size, which we call a universal quiver. These findings generalise the previous ones for permutohedra graphs for knots. We illustrate our results with generic examples, as well as specialisations related to the knots-quivers correspondence.
\end{abstract}

\newpage

\tableofcontents

\newpage

\section{Introduction and summary}\label{sec:Intro}

Symmetric quivers play an important role in recent developments in mathematics and mathematical physics. (A quiver consists of nodes, which we label by $i=1,\ldots,m$, and arrows; it is symmetric if $\forall _{i,j} \, C_{ij}=C_{ji}$, where $C_{ij}$ is the number of arrows from node $i$ to node $j$.) They provide a playground for quiver representation theory and homological algebra -- in particular, one can effectively determine numerical and motivic Donaldson-Thomas invariants of symmetric quivers, which characterize their moduli spaces of representations, and also encode homological dimensions of various associated algebra structures, including cohomological Hall algebras, vertex Lie algebras, and their Koszul duals
\cite{KS0811,KS1006,Efi12,Rei12,FR1512,MR1411,dotsenko2021dt,DotsenkoFeiginReineke}. 

Apart from understanding structures mentioned above, our work is motivated by recent advances in quantum topology, in particular the categorification program for quantum invariants of knots and three-manifolds. This area of research includes a~fascinating connection between symmetric quivers and knot (and link) homologies  \cite{KRSS1707long,KRSS1707short,PSS1802,EKL1811}, as well as the $\widehat{Z}$ and $F_K$ invariants of three-manifolds and knot complements \cite{Kuch2005,Ekholm:2020lqy,Ekholm:2021irc,Cheng:2022rqr},
whose categorification and connection to vertex algebras is one of the mysteries to be revealed. These developments are often referred to as the knots-quivers correspondence. 
In parallel, symmetric quivers turn out to play somewhat analogous role in topological string theory \cite{PS1811,Kimura:2020qns}. Another interesting connection arises between the representation theory of some symmetric quivers and characters (or linear combinations thereof) of vertex operator algebras associated to rational \cite{Nahm,Kedem:1992jv,Kedem:1993ze,Dasmahapatra:1993pw,Berkovich:1997ht} and logarithmic \cite{Flohr:2006id,Feigin:2007sp}
two-dimensional conformal field theories. The structures that we find in this paper should have interesting interpretations in all these contexts.

A~crucial role in our work is played by the~movitic generating series of a symmetric quiver, which encodes Donaldson-Thomas invariants and serves as a bridge to quantum topology in \cite{KRSS1707long,KRSS1707short}. Such a~series depends on $C_{ij}$, generating parameters $\boldsymbol{x}=[x_1,\dots,x_m]$ (with each $x_i$ assigned to the node $i$, for $i=1,\ldots,m$) and the motivic parameter $q$, and arises from the summation over all dimension vectors $\boldsymbol{d}=[d_1,\dots,d_m]$:
\begin{equation}
P_Q(x_1,\dots,x_m,q)=\sum_{d_1,\ldots,d_m\geq 0} \frac{(-q)^{\sum_{i,j} C_{ij} d_i d_j}}{(q^2;q^2)_{d_1} \ldots (q^2;q^2)_{d_m}} x_1^{d_1}\ldots x_m^{d_m}\,,  \label{Z}
\end{equation}
where  $(\xi;q^{2})_{n} = \prod_{k=0}^{n-1}(1-\xi q^{2k})$ is a $q$-Pochhammer symbol. Expressions of this form are also referred to as Nahm sums \cite{Nahm} or fermionic $q$-series \cite{Kedem:1992jv,Kedem:1993ze}.

In what follows we assemble $C_{ij}$ into a~symmetric matrix (for readability we omit its under-diagonal part):
\begin{equation}
    C=\begin{bmatrix}
        C_{11} & \cdots  & C_{1i}\\
        & \ddots  & \vdots \\
        &  & C_{mm}
    \end{bmatrix} \,.
\end{equation}
We often use the notation $Q=(C,\boldsymbol{x})$ to represent a~quiver together with the corresponding generating parameters. The interpretation of $C_{ij}$ as the numbers of arrows implies that $C_{ij}\in \mathbb{N}$, but the results presented in this paper hold more generally for $C_{ij}\in \mathbb{Q}$ -- in that case we simply treat them as parameters in (\ref{Z}).

In this work we study an equivalence class of symmetric quivers which, upon appropriate (monomial) specialisations of the generating parameters, have the same generating series:
\begin{equation}
    Q\sim Q' \Leftrightarrow \left. P_Q(x_1,\dots,x_m,q) = P_{Q'}(x'_1,\dots,x'_n,q) \right|_{x'_i=q^{\alpha_{i,0}}x_1^{\alpha_{i,1}}\cdots x_m^{\alpha_{i,m}}}\, .
\end{equation}
The first example of equivalent quivers was found in \cite{KRSS1707long}. In \cite{JKLNS2105} equivalent quivers of the same size were studied systematically and it was shown that they form families labelled by permutohedra. In \cite{EKL1910} it was shown that equivalent quivers can be generated by two basic operations of linking and unlinking, which change the number of nodes in a~quiver, while keeping the associated generating series invariant. One special example of a~quiver of a larger size and the same generating series, obtained by applying unlinking multiple times, is an infinite diagonal quiver found in \cite{Jankowski:2022qdp}.
In this work we conduct a thorough analysis of equivalent quivers that are generated by a~basic operation of unlinking. In fact, all results from this work apply to the linking as well; we leave the study of combining operations of linking and unlinking for future work.


\bigskip

Our first result in this work is revealing that unlinking operations form an operator monoid with the left action on pairs $(C,\boldsymbol{x})$. Consequently, we prove that all identities in this monoid can be derived from the finite set of basic relations between unlinkings. We also formulate the Connector Theorem, which states that any two sequences of unlinking can be extended in a~way that ultimately produces the same quiver. 
These results provide very efficient tools to determine the structure of equivalent quivers that are produced by multiple application of unlinking. An analogous structure arises from the operation of linking.

As the second step, we introduce permutohedra structures from unlinking operation. Equivalent quivers obtained from unlinking can be assembled into permutohedra: vertices of a~given permutohedron $\Pi_n$ represent quivers $\{Q_{\sigma}\}_{\sigma \in S_n}$ that are related by operations in an associated symmetric group $S_n$, which interchange the numbers of arrows between various nodes of a~quiver and preserve the size of the quiver (such operation corresponds to the fact that two different quivers can be unlinked to the same one). For a~given quiver there may exist several such symmetry groups, so that its $m$ equivalent quivers $\{Q_{\sigma}\}_{\sigma=1,\ldots,m}$ form a~structure that we call a~permutohedra graph, which is made of several permutohedra glued together. Note that in \cite{JKLNS2105} a~similar structure of permutohedra graphs was found for quivers that are associated to knots via the knots-quivers correspondence. However, our current results are more general and follow from properties of unlinking operation, rather than specific features of knot generating functions discussed in \cite{JKLNS2105}.

Furthermore, we show that all quivers that form a~permutohedron can be subsequently unlinked to the same quiver, which we call a universal quiver: 

 \begin{thm}[Permutohedron Theorem]\label{thm:Permutohedron Theorem}
For any quivers $\{Q_{\sigma}\}_{\sigma\in S_n}$ that form permutohedron $\Pi_{n}$ there exists a~universal quiver $\hat{Q}$.
\end{thm}

We also prove that this result can be extended to the permutohedra graph that consists of several permutohedra:

\begin{thm}[Permutohedra Graph Theorem]\label{thm:Permutohedra Graph Theorem}
For any quivers $\{Q_{\sigma}\}_{\sigma=1,\dots,m}$ that form a~permutohedra graph there exists a~universal quiver $\hat{Q}$.
\end{thm}

\begin{wrapfigure}{L}{0.3\textwidth}
\centering
\includegraphics[width=0.25\textwidth]{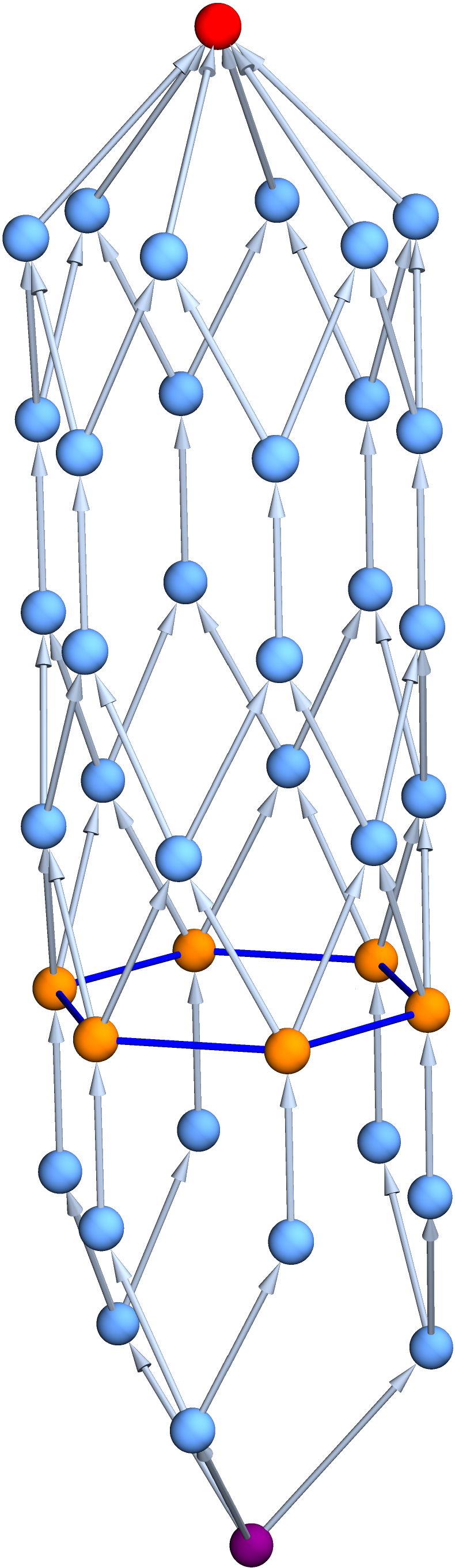}
\caption{\label{fig:Pi3_Rocket}``Cocoon'' of equivalent quivers that includes initial quiver (purple node), permutohedron $\Pi_3$ (blue edges and yellow nodes), and universal quiver (red node).}
\end{wrapfigure}

Apart from the universal quiver of larger size, from which (by the inverse operation to unlinking) every quiver in a~permutohedra graph can be obtained, we conjecture that the whole permutohedra graph can be also obtained from a~single quiver of a~smaller size. An example of such a~structure, in which quivers are represented by nodes and unlinkings by arrows, is shown in Fig.~\ref{fig:Pi3_Rocket}.

To sum up, we find that the operation of unlinking produces a~large web of equivalent quivers of various size, which have the same generating series. This whole web can be reconstructed from a~single universal quiver $\hat Q$. Note that, in general, the quiver generating series depends on a~number of generating parameters, which are in one-to-one correspondence with nodes in a~quiver. Therefore, the equality of generating series that we mentioned above involves certain identification of generating parameters assigned to those nodes in a~larger quiver, which are produced by the unlinking operation. This identification follows directly from the form of unlinking.

As already mentioned above, our results are relevant not only for quiver representation theory, but also in other contexts where quivers or generating series of the form (\ref{Z}) arise. In particular, in this paper we apply the general theorems for unlinking to the knots-quivers correspondence, and relate our results to those in \cite{JKLNS2105}.

\bigskip

The paper is structured as follows. In Sec.~\ref{sec:pre} we introduce operations of unlinking and linking, and briefly review the knots-quivers correspondence as well as permutohedra graphs. In Sec.~\ref{sec:commutation relations} we present the monoidal structure of unlinking and linking, and formulate the Connector Theorem. Sec.~\ref{sec:unlinking} shows how permutohedra graphs are obtained from unlinking. In Sec.~\ref{sec:Universal quivers} we show that the whole permutohedra graph can be determined from a~single universal quiver. In Sec.~\ref{sec:Universal quivers for knots} we specialize to the case of knots and explain how permutohedra graphs for knots arise from the constructions found in this paper. Finally, Sec.~\ref{sec:future} focuses on ideas for future research.


\section{Prerequisites} \label{sec:pre}

In this section we introduce the notation and recall basic facts about unlinking operation and the knots-quivers correspondence. As mentioned in Sec.~\ref{sec:Intro}, in this work we consider symmetric quivers, for which the number of arrows $C_{ij}$ between the nodes $i$ and $j$ is equal to $C_{ji}$. We treat $C_{ij}$ as entries of a~matrix $C$, whose size (equal to the number of nodes in a~given quiver) we typically denote by $m$. Of our primary interest are also generating functions (\ref{Z}) that, in addition to $C_{ij}$, depend also on $q$ and on the vector of generating parameters $\boldsymbol{x}=[x_1,\dots,x_m]$, with each $x_i$ assigned to a~node $i$. In what follows we use the notation $Q=(C,\boldsymbol{x})$, and sometimes, with some abuse of terminology, refer to $Q$ as a~quiver too. 

\subsection{Unlinking and linking}
\label{sec:Multi-cover skein relations}

A crucial role in our consideration is played by the operations of unlinking and linking. They were introduced in \cite{EKL1910}, together with a~fascinating interpretation in terms of multi-cover skein relations and 3d $\mathcal{N}=2$ theories (see App. \ref{sec:Statements from the multi-cover skein paper} for a short summary). 

Linking and unlinking are operations that respectively add or remove one pair of arrows between two particular nodes $i$ and $j$, as well as increase the size of a~quiver by one node, which gets connected by a~particular pattern of arrows to other nodes. Importantly, the quiver generating series for the enlarged quiver is equal to the generating series of the original quiver (\ref{Z}), once the generating parameter $x_{m+1}$  assigned to a~new node is appropriately identified in terms of $x_i$ and $x_j$. 

More precisely, for $Q=(C,\boldsymbol{x})$ given by
\begin{equation}
    \begin{split}
    C & =\left[\begin{array}{ccccccc}
 C_{11}  &  \cdots  &  C_{1i}  &  \cdots  &  C_{1j}  &  \cdots  &  C_{1m}\\
   &  \ddots\  &  \vdots &   &  \vdots  &   &  \vdots\\
 &  &  C_{ii}  &  \cdots &  C_{ij}  &  \cdots &  C_{im}\\
 &  &  &  \ddots\  &  \vdots  &   &  \vdots\\
 &  &  &  &  C_{jj}  &  \cdots  &  C_{jm}\\
 &  &  &  &  &  \ddots &  \vdots\\
 &  &  &  &  &  & C_{mm}
\end{array}\right]\,,\\
\boldsymbol{x} & =\left[x_{1},\dots,x_{i},\dots,x_{j},\dots,x_{m}\right]\,,
\end{split}
\end{equation}
we define unlinking of distinct nodes $i$ and $j$ as $U(ij)Q = (U(ij)C,U(ij)\boldsymbol{x})$, where
\begin{equation}\label{eq:unlinking definition}
    \begin{split}U(ij)C & =\left[\begin{array}{cccccccc}
 C_{11}  &  \cdots &  C_{1i}  &  \cdots &  C_{1j}  &  \cdots &  C_{1m}  &  C_{1i}+C_{1j}\\
 &  \ddots\  &  \vdots &   &  \vdots &   &  \vdots &  \vdots\\
 &  &  C_{ii}  &  \cdots &  C_{ij}-1  &  \cdots &  C_{im}  &  C_{ii}+C_{ij}-1\\
 &  &  &  \ddots\  &  \vdots &   &  \vdots &  \vdots\\
 &  &  &  &  C_{jj}  &  \cdots &  C_{jm}  &  C_{ij}-1+C_{jj}\\
 &  &  &  &  &  \ddots\  &  \vdots &  \vdots\\
 &  &  &  &  &  &  C_{mm}  &  C_{im}+C_{jm}\\
 &  &  &  &  &  &    &  C_{ii}+C_{jj}+2C_{ij}-1 
\end{array}\right]\,,\\
U(ij)\boldsymbol{x} & =\left[x_{1},\dots,x_{i},\dots,x_{j},\dots,x_{m},q^{-1}x_{i}x_{j}\right]\,.
\end{split}
\end{equation}

Let us stress that unlinking is symmetric in indices $i$ and $j$: $U(ij)=U(ji)$. Moreover, we can perform unlinking many times which means that $i$ and $j$ might refer to the new nodes created in the previous unlinking. In order to prevent errors coming from the attempt of unlinking of nodes that have not been created yet, we define $U(ij)Q$ to be an empty pair~$(\,,)$ if either $i$ or $j$ is not a~node in $Q$.

Furthermore, we will express all generating parameters in terms of initial $x_{1},\dots,x_{m}$, which will also allow us to keep track of all unlinkings. In order to do it, we will denote the new node coming from  $U(ij)$ by $(ij)$ and the corresponding generating parameter $q^{-1}x_{i}x_{j}$ by $(x_{i}x_{j})$. For example, if we first unlink nodes $i$ and $j$, and then the new node coming from that unlinking together with node $k$, we will write
\begin{equation*}
    U((ij)k)  U(ij) \boldsymbol{x} = \left[x_{1},\dots,x_{m},(x_{i}x_{j}),((x_{i}x_{j})x_k)\right] = \left[x_{1},\dots,x_{m},q^{-1}x_{i}x_{j},q^{-2}x_{i}x_{j}x_k\right]
     \,.
\end{equation*}
Note that the ordering of unlinkings imposes the ordering of new generating parameters and we will always assume that this ordering is followed. It imposes a~one-to-one correspondence between the ordering of operators and the ordering in the vector of generating parameters\footnote{This applies also to vectors without brackets, but we use them to make everything more explicit.}, which will be useful in proving our results.

We also define linking of distinct nodes $i$ and $j$ as $L(ij)Q = (L(ij)C,L(ij)\boldsymbol{x})$, where
\begin{equation}\label{eq:linking definition}
    \begin{split}L(ij)C & =\left[\begin{array}{cccccccc}
 C_{11}  &  \cdots &  C_{1i}  &  \cdots &  C_{1j}  &  \cdots &  C_{1m}  &  C_{1i}+C_{1j}\\
 &  \ddots\  &  \vdots &   &  \vdots &   &  \vdots &  \vdots\\
 &  &  C_{ii}  &  \cdots &  C_{ij}+1  &  \cdots &  C_{im}  &  C_{ii}+C_{ij}\\
 &  &  &  \ddots\  &  \vdots &   &  \vdots &  \vdots\\
 &  &  &  &  C_{jj}  &  \cdots &  C_{jm}  &  C_{ij}+C_{jj}\\
 &  &  &  &  &  \ddots\  &  \vdots &  \vdots\\
 &  &  &  &  &  &  C_{mm}  &  C_{im}+C_{jm}\\
 &  &  &  &  &  &    &  C_{ii}+C_{jj}+2C_{ij} 
\end{array}\right]\,,\\
L(ij)\boldsymbol{x} & =\left[x_{1},\dots,x_{i},\dots,x_{j},\dots,x_{m},x_{i}x_{j}\right]\,.
\end{split}
\end{equation}

Our main results, such as permutohedra structure of equivalent quivers and the~existence of universal quivers, follow directly from properties of unlinking and linking and are exactly analogous for both operations treated separately. Because of that, in the majority of the paper (with an exception of one specific section) we will focus on unlinking and leave analysis of the structures that arise from combining both operations for future work.

\subsection{Knots-quivers correspondence and permutohedra graphs}\label{sec:KQ correspondence and permutohedra}

The knots-quivers correspondence, introduced in \cite{KRSS1707long,KRSS1707short}, is the statement that for a~given knot one can find a~symmetric quiver, whose various characteristics reproduce invariants of the original knot. In particular, the nodes of a~quiver correspond to generators of HOMFLY-PT homology of the corresponding knot, the quiver generating series (\ref{Z})  is a~generating series of colored HOMFLY-PT polynomials upon appropriate identification of generating parameters $x_i$, and motivic Donaldson-Thomas invariants of a~quiver encode LMOV invariants \cite{OV9912,LM01,LMV00} of the corresponding knot. This correspondence has been proven to hold for various families of knots, such as some torus and twist knots, rational knots, and arborescent knots \cite{KRSS1707long,SW1711,SW2004}, and it is conjectured to hold (with some modifications for knots that do not satisfy the exponential growth property \cite{EKL1811}) for all knots.

In fact, there are typically many quivers that correspond to the same knot. This phenomenon was found already in \cite{KRSS1707long}. It was then systematically analyzed in \cite{JKLNS2105}, where it was shown that equivalent quivers arise from various symmetries that are generated by transpositions of the entries of a~quiver matrix $C$. Such transpositions generate permutation groups, and in consequence equivalent quivers are in one-to-one correspondence with vertices of permutohedra, while edges of such permutohedra represent the above mentioned transpositions. A given quiver may encode several such permutation groups, and in consequence its equivalent quivers correspond to vertices of the permutohedra graph, which is obtained from combining appropriate permutohedra together. 
For more details see App.~\ref{sec:Statements from the permutohedra paper}, which will serve as a~reference for Sec.~\ref{sec:Universal quivers for knots}, where we will show that permutohedra graphs constructed in \cite{JKLNS2105} are special cases of permutohedra graphs analysed in this work, which will lead us to analogues of Theorems~\ref{thm:Permutohedron Theorem} and~\ref{thm:Permutohedra Graph Theorem}.

\section{Monoidal structure of unlinking}\label{sec:commutation relations}

In this section we introduce the operator monoid of unlinking operations for  quivers and study its properties in detail.
We list basic relations and show that all other relations can be expressed by them. 
We also propose a~generalisation to linking. Furthermore, we prove that any two sequences of unlinking can be amended by other sequences in a~way that they form a~relation and illustrate it with an~example.

\subsection{Unlinking relations}\label{sec:generators and relations}

We can treat unlinking $U(ij)$ (and, similarly, linking) defined in Sec.~\ref{sec:Multi-cover skein relations} as operator acting on the set of pairs $(C,\boldsymbol{x})$, where $U(ij)(C,\boldsymbol{x})$ is given by (\ref{eq:unlinking definition}) if both $i$ and $j$ are nodes of the corresponding quiver, otherwise it is an empty pair $(\,,)$.
The annihilation condition may seem rather arbitrary, but it provides consistency for the definition to follow.
\begin{dfn}
Consider the set $\{U(ij)\}$ where $i$ and $j$ run through a~universal alphabet $\mathcal{A}$ (we assume that the set of nodes of any quiver is labelled by this alphabet).
If we complete this set by an element $U(\emptyset)$ which acts on any quiver as an identity (by definition), then its closure under composition operation
\begin{equation}
\mathcal{U}:=(\{U(ij)\}_{i,j\in \mathcal{A}} \cup \{U(\emptyset)\}, \circ)
\end{equation}
is an operator monoid, which we call the unlinking monoid.
\end{dfn}
Elements of 
$\mathcal{U}$ are sequences of unlinking of the form
\begin{equation}
\boldsymbol{U}(\boldsymbol{i}\boldsymbol{j}) = U(i_{n}j_{n})\dots U(i_{2}j_{2})U(i_{1}j_{1}),
\end{equation}
where the indices may be repeated, and application of operations goes from right to left.
The operator identity between two elements in this monoid will be called a~relation, and it is basic if it cannot be derived from other relations.\footnote{This is conceptually different from relations between quivers -- operator identity is assumed to hold for an arbitrary symmetric quiver, i.e. is independent from either choice of adjacency matrix or a~vector of variables.}  Using this notation, a~generic relation in $\mathcal{U}$ can be compactly written as
\begin{equation}\label{eq:relation}
    \boldsymbol{U}(\boldsymbol{i}\boldsymbol{j}) = \boldsymbol{U}(\boldsymbol{k}\boldsymbol{l}).
\end{equation}
When we write this equation, we assume that $m=n$ because if two sequences of unlinking had different lengths,  their application to the same quiver would yield two quivers of different sizes. Note that such equalities put severe constraints on indices $(\boldsymbol{ij})$ and $(\boldsymbol{jk})$.
It turns out that in $\mathcal{U}$ identifications of the form (\ref{eq:relation}) have an intricate structure, and the rest of the paper is devoted to their thorough analysis. In turn, when applied to a~particular quiver, these yield non-trivial $q$-identities between their generating series~(\ref{Z}).

In the rest of this section we assume that $Q$ is an arbitrary quiver and $(i,j,k,l)$ is a~quadruple of \emph{distinct} nodes of $Q$. We proceed to our first important result.
\begin{prp}\label{prp:all basic relations}
The following basic relations for unlinking hold:
\begin{subequations}
\begin{align}
        U(kl) U(ij) &= U(ij) U(kl)  \label{eq:square} \\
        U((ij)k)  U(jk)  U(ij) &= U(i(jk))  U(ij)  U(jk) \label{eq:first hexagon}\\
        U(jk)  U((ij)k)  U(ij) &= U(ij)  U(i(jk))  U(jk)  \label{eq:second hexagon}
    \end{align}
\end{subequations}
\end{prp}
The proof goes by a~straightforward computation, which is presented in App.~\ref{app:hexagon}.
%
%
Remarkably, the two unlinkings do not commute if there is a~repeated index:
\begin{equation}
    U(jk)  U(ij) \neq U(jk)  U(ij)\,. 
\end{equation}
The resulting quivers
\begin{equation}
    U(jk)[U(ij)Q],\quad U(ij)[U(jk)Q]
\end{equation}
differ by one pair of arrows between nodes $(ij)$ and $k$ on one side and between nodes $i$ and $(jk)$ on the other. However, we can unlink this extra pair of arrows using $U((ij)k)$ and $U(i(jk))$ respectively and obtain a~non-trivial relation which cannot be reduced to a~simpler one (again, we refer the reader to App.~\ref{app:hexagon} for the details).

\begin{rmk}\label{rmk:square and hexagon}
    We call the relation \eqref{eq:square} the square and relations (\ref{eq:first hexagon}-\ref{eq:second hexagon}) the hexagons, which refer to the shapes they form in the diagrams, where an arrow labelled by $(ab)$ corresponds to application of $U(ab)$ in which the new node $(ab)$ is created:
    \[\begin{tikzcd}[column sep={1cm},row sep={1cm}]
	   &&&& \bullet &&& \bullet \\
	   & \bullet && \bullet && \bullet & \bullet && \bullet \\
	   \bullet && \bullet & \bullet && \bullet & \bullet && \bullet \\
	   & \bullet &&& \bullet &&& \bullet
	   \arrow["{(ij)}"{description}, from=4-2, to=3-1]
	   \arrow["{(kl)}"{description}, from=3-1, to=2-2]
	   \arrow["{(kl)}"{description}, from=4-2, to=3-3]
	   \arrow["{(ij)}"{description}, from=3-3, to=2-2]
	   \arrow["{(ij)}"{description}, from=4-5, to=3-4]
	   \arrow["{(jk)}"{description}, from=3-4, to=2-4]
	   \arrow["{((ij)k)}"{description}, from=2-4, to=1-5]
	   \arrow["{(jk)}"{description}, from=4-5, to=3-6]
	   \arrow["{(ij)}"{description}, from=3-6, to=2-6]
	   \arrow["{(i(jk))}"{description}, from=2-6, to=1-5]
	   \arrow["{(ij)}"{description}, from=4-8, to=3-7]
	   \arrow["{((ij)k)}"{description}, from=3-7, to=2-7]
	   \arrow["{(jk)}"{description}, from=2-7, to=1-8]
	   \arrow["{(jk)}"{description}, from=4-8, to=3-9]
	   \arrow["{(i(jk))}"{description}, from=3-9, to=2-9]
	   \arrow["{(ij)}"{description}, from=2-9, to=1-8]
    \end{tikzcd}\]
    The operators  $U((ij)k)$ and $U(i(jk))$ from the hexagon relation will be called associativity unlinkings. Note that they differ by a~position of inner bracket (which corresponds to the unlinking that acts first) and nodes $((ij)k)$ and $(i(jk))$ can be identified.

\end{rmk}


\subsection{Completeness Theorem}\label{sec:completeness}

After showing that relations (\ref{eq:square}-\ref{eq:second hexagon}) hold, we will prove that this list is complete and there are no other basic relations formed by unlinkings.

\begin{thm}[Completeness Theorem]
   For every relation $\boldsymbol{U(ij)}=\boldsymbol{U(kl)}$ there exists a~sequence of square and hexagon moves which transforms $\boldsymbol{U(ij)}$ into $\boldsymbol{U(kl)}$, and vice versa.
\end{thm}

\begin{proof}
    
Denote  $Q_{(\boldsymbol{ij})}=(C_{(\boldsymbol{ij})},\boldsymbol{x}_{(\boldsymbol{ij})})=\boldsymbol{U}(\boldsymbol{i}\boldsymbol{j})Q$ and $Q_{(\boldsymbol{kl})}=(C_{(\boldsymbol{kl})},\boldsymbol{x}_{(\boldsymbol{kl})})=\boldsymbol{U}(\boldsymbol{k}\boldsymbol{l})Q$.  We assume that the ordering in $\boldsymbol{x}_{(\boldsymbol{ij})}$ and $\boldsymbol{x}_{(\boldsymbol{kl})}$ follows the order of unlinkings in $\boldsymbol{U}(\boldsymbol{i}\boldsymbol{j})$ and $\boldsymbol{U}(\boldsymbol{k}\boldsymbol{l})$ respectively 
    
        The relation $\boldsymbol{U}(\boldsymbol{i}\boldsymbol{j})=\boldsymbol{U}(\boldsymbol{k}\boldsymbol{l})$ implies that $C_{(\boldsymbol{kl})}=\pi C_{(\boldsymbol{ij})}$ and $\boldsymbol{x}_{(\boldsymbol{kl})}=\pi\boldsymbol{x}_{(\boldsymbol{ij})}$ , where $\pi$~is a~permutation. 
Every permutation can be expressed in terms of transpositions of the neighbouring elements:
\begin{equation}
    \pi = \tau_1 \tau_2 \dots \tau_t\,.
\end{equation}
If all indices that are mixed by $\tau_i$ are different
\begin{equation}
    \tau_i[\dots, (x_{i_1}x_{i_2}), (x_{i_3}x_{i_4}), \dots] = [\dots, (x_{i_3}x_{i_4}), (x_{i_1}x_{i_2}), \dots] \ ,
\end{equation}
then it corresponds to the square move and its action is neutral for all other entries of the bracketed vector.
If two indices that are mixed by $\tau_i$ are the same, then in the bracketed vector there must be a~corresponding associativity unlinking (in other case we would violate $C_{(\boldsymbol{kl})}=\pi C_{(\boldsymbol{ij})}$)\footnote{One may ask about the position of the associativity unlinking -- in principle it does not have to be immediately after or between the corresponding pair of unlinkings, it can be moved by other hexagon or square moves. However, in that case there exist a~sequence of transpositions such that the associativity unlinking is immediately after or between the the corresponding pair of unlinkings and all other cases would violate $C_{(\boldsymbol{kl})}=\pi C_{(\boldsymbol{ij})}$.} and $\tau_i$ represents one of the hexagon moves (that includes also the associativity move which is invisible for the unbracketed vectors):
\begin{equation*}
\begin{split}
    \tau_i[\dots, (x_{i_1}x_{i_2}), (x_{i_2}x_{i_3}), ((x_{i_1}x_{i_2})x_{i_3}), \dots] &= [\dots, (x_{i_2}x_{i_3}), (x_{i_1}x_{i_2}), (x_{i_1}(x_{i_2}x_{i_3})), \dots] \ ,\\
    \tau_i[\dots, (x_{i_1}x_{i_2}), ((x_{i_1}x_{i_2})x_{i_3}), (x_{i_2}x_{i_3}), \dots] &= [\dots, (x_{i_2}x_{i_3}), (x_{i_1}(x_{i_2}x_{i_3})), (x_{i_1}x_{i_2}), \dots] \,.
\end{split}
\end{equation*}
The other options are excluded by requirement that we unlink different nodes.

We know that 
\begin{equation}
    \boldsymbol{x}_{(\boldsymbol{kl})}=\pi\boldsymbol{x}_{(\boldsymbol{ij})}=\tau_1 \tau_2 \dots \tau_t \boldsymbol{x}_{(\boldsymbol{ij})}\,,
\end{equation}
and since operators and vectors of generating parameters and in one-to-one correspondence, we have
\begin{equation}
    \boldsymbol{U}(\boldsymbol{k}\boldsymbol{l})=\tau_1 \tau_2 \dots \tau_t\ \boldsymbol{U}(\boldsymbol{i}\boldsymbol{j}) \ ,
\end{equation}
where each $\tau_i$ is understood as a~square or hexagon move. This confirms that list of basic relations is complete. 
\end{proof}

\begin{exmp}
    Below we illustrate the above theorem with an example.
    Consider the following relation between two sequences of unlinking
    \begin{equation}\label{eq:non-basic relation example}
        U((ij)(jk))U((ij)k)U(jk)U(ij) = U(j(i(jk)))U(ij)U(i(jk))U(jk)\,.
    \end{equation}
    We can express the permutation $\pi$ explicitly in terms of two transpositions:
    \begin{equation}
    \pi = \tau_1 \tau_2\,,
    \end{equation}
    reducing the proof of (\ref{eq:non-basic relation example}) to applying a basic relation twice. The first transposition
    exchanges $U(ij)$ and $U(jk)$, and requires a hexagon move:
    \begin{gather}
        U((ij)(jk))U((ij)k)U(jk)U(ij) \nonumber\\
        {\big\downarrow }\,{\tau_1}  \\
        U((ij)(jk))U(i(jk))U(ij)U(jk)\,.\nonumber
    \end{gather}
    Note that $\tau_1$ does not simply exchange $U(ij)$ and $U(jk)$, but also swaps the associativity unlinking, which is required by (\ref{eq:first hexagon}). The second transposition exchanges $U(ij)$ with $U(i(jk))$:
       \begin{gather}
        U((ij)(jk))U(i(jk))U(ij)U(jk) \nonumber\\
        {\big\downarrow }\,\tau_2  \\
        U(j(i(jk)))U(ij)U(i(jk))U(jk)\,,\nonumber
    \end{gather}
    which is an application of the relation (\ref{eq:first hexagon}). The associativity unlinking changes into $U(j(i(jk)))$, since 
    $(i(jk))$ is the node which has been created before. This shows how (\ref{eq:non-basic relation example}) follows from the basic relations (\ref{eq:square}-\ref{eq:second hexagon}).
\end{exmp}

\subsection{Generalisation to linking}\label{sec:linking}
\begin{prp}
The set of unlinking relations (\ref{eq:square}-\ref{eq:second hexagon}) can be extended to the following three groups of relations:
\begin{equation}
    \label{eq:commutators}
\begin{split}
       U(kl)  U(ij) &= U(ij)  U(kl) \\[4pt]
        L(kl)  L(ij) &= L(ij)  L(kl) \\[4pt]
        L(kl)  U(ij) &= U(ij)  L(kl) \\[4pt]
\end{split}        
\end{equation}
\begin{equation}
\label{eq:commutators2}
\begin{split}
        U((ij)k)  U(jk)  U(ij) &= U(i(jk))  U(ij)  U(jk)\\[4pt]
        U(jk)  U((ij)k)  U(ij) &= U(ij)  U(i(jk))  U(jk)\\[4pt]
        L((ij)k)  L(jk)  L(ij) &= L(i(jk))  L(ij)  L(jk)\\[4pt]
        L(jk)  L((ij)k)  L(ij) &= L(ij)  L(i(jk))  L(jk)\\[4pt]
\end{split}        
\end{equation}
\begin{equation}
\label{eq:commutators3}
\begin{split}
        U(j(k(ij)))  L((ij)k) L(jk)  U(ij)&= 
        U((ij)(jk))  U(i(jk))  U(ij) L(jk) \\[4pt]
        L((ij)(jk))  L((ij)k) L(jk)  U(ij)&= 
        L(j(i(jk)))  U(i(jk))  U(ij) L(jk) \\[4pt]
        U((ij)(jk)) L(jk)  L((ij)k)  U(ij)&= 
        U(j(i(jk)))  U(ij)  U(i(jk)) L(jk) \\[4pt]
        L(j(k(ij))) L(jk)  L((ij)k)  U(ij)&= 
        L((ij)(jk))  U(ij)  U(i(jk)) L(jk) \\[4pt]
        L(j(k(ij)))  U((ij)k) U(jk)  L(ij)&= 
        L((ij)(jk))  L(i(jk))  L(ij) U(jk) \\[4pt]
        U((ij)(jk))  U((ij)k) U(jk)  L(ij)&= 
        U(j(i(jk)))  L(i(jk))  L(ij) U(jk) \\[4pt]
        L((ij)(jk)) U(jk)  U((ij)k)  L(ij)&= 
        L(j(i(jk)))  L(ij)  L(i(jk)) U(jk) \\[4pt]
        U(j(k(ij))) U(jk)  U((ij)k)  L(ij)&= 
        U((ij)(jk))  L(ij)  L(i(jk)) U(jk) \\[4pt]
    \end{split}
\end{equation}
\end{prp}
We will omit the proof as it is a~straightforward computation (it is done in App.~\ref{app:hexagon} for the case of pure unlinking).
It is important to stress that relations (\ref{eq:commutators}-\ref{eq:commutators3}) are mutually independent and cannot be reduced to simpler relations. It follows from the fact that relations (\ref{eq:commutators}-\ref{eq:commutators3}) provide the minimal numbers (2, 3 and 4, respectively) of operations needed to commute one unlinking or linking with another, as can be checked directly by applying these operations step-wise to a~quiver whose set of nodes includes $\{i,j,k,l\}$ as a~subset.

In the rest of the paper we will focus on unlinking (the case of linking is completely analogous) and leave the analysis of combined unlinking and linking to the future research. We also note that it would be very interesting to recognise some familiar algebraic object behind the relations (\ref{eq:commutators}-\ref{eq:commutators3}).

\subsection{Connector Theorem}\label{sec:connector}

In Sec.~\ref{sec:generators and relations} and~\ref{sec:completeness} we presented the list of basic relations and proved that it is complete. Now we will show that any two sequences of unlinking can be amended by other sequences in a~way that they form a~relation.

\begin{dfn}
Given any two sequences of unlinking $\boldsymbol{U}(\boldsymbol{i}\boldsymbol{j})= U(i_{n}j_{n})\dots U(i_{2}j_{2})U(i_{1}j_{1})$
and $\boldsymbol{U}(\boldsymbol{k}\boldsymbol{l})=U(k_{m}l_{m})\dots U(k_{2}l_{2})U(k_{1}l_{1})$ (we will call them initial unlinkings) of length $n$ and $m$, a~pair of sequences of unlinking $\boldsymbol{U}(\boldsymbol{i'}\boldsymbol{j'})=U(i'_{n'}j'_{n'})\dots U(i'_{2}j'_{2})U(i'_{1}j'_{1})$ and $\boldsymbol{U}(\boldsymbol{k'}\boldsymbol{l'})=U(k'_{m'}l'{}_{m'})\dots U(k'_{2}l'_{2})U(k'_{1}l'_{1})$ is called an $(n,m)$-connector if
\begin{equation}
\boldsymbol{U}(\boldsymbol{k'}\boldsymbol{l'})\boldsymbol{U}(\boldsymbol{i}\boldsymbol{j})=\boldsymbol{U}(\boldsymbol{i'}\boldsymbol{j'})\boldsymbol{U}(\boldsymbol{k}\boldsymbol{l}).
\end{equation}
\end{dfn}

In the following we want to show the existence of a~finite $(n,m)$-connector
for any $n,m\in\mathbb{Z_{+}}$. In order to do it, we will start
from proving three lemmas.

\begin{lma}
\label{lma:(1,1)-connector}
For any two unlinkings there exists a~finite $(1,1)$-connector. 
\end{lma}
\begin{proof}
For two unlinkings that do not share an index (i.e. unlinkings of the form $U(ij)$, $U(kl)$ ), the $(1,1)$-connector is given by the square relation:
\begin{equation}
    U(kl)  U(ij) = U(ij)  U(kl) \ ,
\end{equation}
hence $\boldsymbol{U}(\boldsymbol{k'}\boldsymbol{l'})=U(kl)$, $\boldsymbol{U}(\boldsymbol{i}\boldsymbol{j})=U(ij)$, $\boldsymbol{U}(\boldsymbol{i'}\boldsymbol{j'})=U(ij)$, $\boldsymbol{U}(\boldsymbol{k}\boldsymbol{l})=U(kl)$.

For two unlinkings that do share an index (i.e. unlinkings of the form $U(ij)$, $U(jk)$ ), the $(1,1)$-connector is given by the hexagon relation:
\begin{equation}
    U((ij)k)  U(jk)  U(ij) = U(i(jk))  U(ij)  U(jk) \ ,
\end{equation}
so $\boldsymbol{U}(\boldsymbol{k'}\boldsymbol{l'})=U((ij)k)  U(jk)$, $\boldsymbol{U}(\boldsymbol{i}\boldsymbol{j})=U(ij)$, $\boldsymbol{U}(\boldsymbol{i'}\boldsymbol{j'})=U(i(jk))  U(ij)$, $\boldsymbol{U}(\boldsymbol{k}\boldsymbol{l})=U(jk)$.
\end{proof}

Basing on Lemma~\ref{lma:(1,1)-connector}, we will write an algorithm for constructing any
$(n,m)$-connector by sequential creation of $(1,1)$-connectors.

\begin{dfn}

We define the Connector Algorithm as follows:
\begin{enumerate}
\item Start from $U(i_{1}j_{1})$ and $U(k_{1}l_{1})$ and apply $(1,1)$-connector. Create the set of identities which contains the resulting equality. 
\item For $p>1$, do the following loop:
\begin{enumerate}
\item Take all sequences of unlinking from the set of identities and cut them at length~$p$ (take first $p$ unlinkings from each of them). 
\item Create the (ordered) list of sequences of unlinking of length $p$ as follows:
\begin{enumerate}
\item If $p > n$, start from the sequence containing $\boldsymbol{U}(\boldsymbol{i}\boldsymbol{j})$; if $p\leq n$, use the first~$p$ unlinkings: $U(i_{p}j_{p})\dots U(i_{2}j_{2})U(i_{1}j_{1})$. 
\item The next sequence on the list is the one that differs from the previous
one by an application of $(1,1)$-connector or by the last unlinking in the sequence. 
\item If $p > n$, the last sequence on the list is the one containing $\boldsymbol{U}(\boldsymbol{k}\boldsymbol{l})$; if 
$p\leq m$, it is $U(k_{p}l{}_{p})\dots U(k_{2}l_{2})U(k_{1}l_{1})$. 
\end{enumerate}
\item From the list of sequences of length $p$ take all pairs which share
the first $p-1$ unlinkings and apply $(1,1)$-connector to each pair.
Add the resulting equalities to the set of identities. Include also equalities coming from combining different identities and the identification of nodes that differ only
by bracketing (e.g. $(1(23))$ and $((12)3)$ are the same node).
\item If the set of identities contains $\boldsymbol{U}(\boldsymbol{k'}\boldsymbol{l'})\boldsymbol{U}(\boldsymbol{i}\boldsymbol{j})=\boldsymbol{U}(\boldsymbol{i'}\boldsymbol{j'})\boldsymbol{U}(\boldsymbol{k}\boldsymbol{l})$ for
some sequences of unlinking $\boldsymbol{U}(\boldsymbol{k'}\boldsymbol{l'})$ and $\boldsymbol{U}(\boldsymbol{i'}\boldsymbol{j'})$, stop.
\item Increase $p$ by $1$.
\end{enumerate}
\end{enumerate}
\end{dfn}

\begin{dfn}
    A descendant of a~given sequence of unlinkings $U(i_{n}j_{n})\dots U(i_{2}j_{2})U(i_{1}j_{1})$ is any sequence of unlinkings that is created in the Connector Algorithm and starts from $U(i_{n}j_{n})\dots U(i_{2}j_{2})U(i_{1}j_{1})$.
\end{dfn}

Lemma~\ref{lma:(1,1)-connector} guarantees that the Connector Algorithm cannot stop for any reason other than constructing the $(n,m)$-connector. The next two lemmas will show that the Connector Algorithm cannot lead to an infinite process, which in turn will show the existence of a~finite $(n,m)$-connector.

\begin{lma}
\label{lma:uniqueness of (1,1)-connectors}
In the Connector Algorithm every application of the $(1,1)$-connector
is unique. In other words, no $(1,1)$-connector can appear in this
process twice if we differentiate between the initial unlinkings that
unlink the same nodes.\footnote{For example, in $U(i_{2}j_{2})U(i_{1}j_{1})=U(12)U(12)$ we treat
$U(i_{2}j_{2})$ and $U(i_{1}j_{1})$ as different.}
\end{lma}

\begin{proof}
Without loss of generality, we assume that a~given $(1,1)$-connector
is applied to unlinkings $U(ab)$ and $U(cd)$ at the end of two sequences
of length $p$ (we do not assume that $a,b,c,d$ are all different
indices). 
Then we know that all their descendants have either $U(ab)$ or $U(cd)$ as $p$-th unlinking, and the other one further in the
sequence (i.e. it is $r$-th unlinking in the sequence, with $r>p$). This means that in further steps of the algorithm $U(ab)$ and $U(cd)$ cannot appear in the same  $(1,1)$-connector.
On the other hand, all other sequences of unlinking are either before
or after them on the lists of sequences. Those before have $U(ab)$
and $U(cd)$ in one order, and those after have them in the reversed
order. In consequence, the $(1,1)$-connector between $U(ab)$ and
$U(cd)$ cannot appear again.
\end{proof}

\begin{lma}
\label{lma:initial indices}
Only indices from the initial unlinkings can be repeated indices in
$(1,1)$-connectors that appear in the Connector Algorithm.    
\end{lma}

\begin{proof}
We prove the lemma  indirectly. Without loss of generality, we assume
that the index $(ab)$ (referring to the node that was produced in
$U(ab)$) is repeated in the $(1,1)$-connector between $U((ab)c)$
and $U((ab)d)$. Moreover, we assume that $U(bc)$ is before $U(bd)$
in the initial sequence of unlinkings. 

The $(1,1)$-connector between $U((ab)c)$ and $U((ab)d)$ can appear only for descendants of the
sequences  appearing in the $(1,1)$-connector between $U(ab)$ and $U(bc)$ with $U(ab)$ being
before $U(bc)$ in the sequence of unlinkings (see Remark~\ref{rmk:square and hexagon}). However, this means that $U(ab)$ does not appear further in the sequence and cannot meet
$U(bd)$ to create $U((ab)d)$ in the $(1,1)$-connector realised by the hexagon relation. 
\end{proof}

\begin{thm}[Connector Theorem]
For any two sequences of unlinking of length $n$ and $m$ there
exist a~finite $(n,m)$-connector.    
\end{thm}

\begin{proof}
We prove the theorem indirectly. Let us assume that for given two
sequences of unlinking of length $n$ and $m$ the Connector Algorithm
never stops. From Lemma~\ref{lma:uniqueness of (1,1)-connectors} we know that every $(1,1)$-connector is
unique, therefore there must be an infinite number of hexagons that
create new nodes. However, from Lemma~\ref{lma:initial indices} we know that only indices
from the initial unlinkings can be repeated indices in $(1,1)$-connectors
that appear in the Connector Algorithm, so only a~finite number of
different hexagons can be created in this way. We arrive at a~contradiction,
which means that the the Connector Algorithm has to stop. Lemma~\ref{lma:(1,1)-connector} guarantees that the Connector Algorithm cannot stop for any reason other than constructing the $(n,m)$-connector, which proves the theorem.    
\end{proof}

\subsection{Example of \texorpdfstring{$(2,1)$}{(2,1)}-connector}

Let us illustrate the action of the Connector Algorithm with the example of $(2,1)$-connector between $\boldsymbol{U}(\boldsymbol{i}\boldsymbol{j})=U(45)U(12)$ and $\boldsymbol{U}(\boldsymbol{k}\boldsymbol{l})=U(23)$.  Using the notation from Sec.~\ref{sec:connector}, we have:
\begin{itemize}
\item $n=2$, $U(i_{1}j_{1})=U(12)$, $U(i_{2}j_{2})=U(45)$ ($i_{1}=1$,$j_{1}=2$, $i_{2}=4$, $j_{2}=5$),
\item $m=1$, $U(k_{1}l_{1})=U(23)$, ($k_{1}=2$, $l_{1}=3$).
\end{itemize}
The Connector Algorithm runs as follows:
\begin{enumerate}
\item We start from $U(i_{1}j_{1})=U(12)$ and $U(k_{1}l_{1})=U(23)$ and
apply $(1,1)$ connector. Since $2$ is a~repeated index, we obtain
\begin{equation}
U((12)3)U(23)U(12)=U(1(23))U(12)U(23),
\end{equation}
which we add to the newly created set of identities. Graphically, the outcome can be represented as
\begin{equation}
\begin{tikzcd}[column sep={1cm},row sep={1cm}]
	&& \bullet \\
	\bullet & \bullet && \bullet \\
	& \bullet && \bullet \\
	&& \bullet
	\arrow["{(12)}"{description}, from=4-3, to=3-2]
	\arrow["{(23)}"{description}, from=4-3, to=3-4]
	\arrow["{(45)}"{description}, from=3-2, to=2-1]
	\arrow["{(23)}"{description}, color={rgb,255:red,92;green,92;blue,214}, from=3-2, to=2-2]
	\arrow["{((12)3)}"{description}, color={rgb,255:red,92;green,92;blue,214}, from=2-2, to=1-3]
	\arrow["{(12)}"{description}, color={rgb,255:red,92;green,92;blue,214}, from=3-4, to=2-4]
	\arrow["{(1(23))}"{description}, color={rgb,255:red,92;green,92;blue,214}, from=2-4, to=1-3]
\end{tikzcd}
\end{equation}
\item For $p=2$, we do the following steps:
\begin{enumerate}
\item Take sequences $U((12)3)U(23)U(12)$ and $U(1(23))U(12)U(23)$, and
cut them at length $2$, which results in $U(23)U(12)$ and $U(12)U(23)$.
\item We build our ordered list of sequences:
\begin{enumerate}
\item We start from $\boldsymbol{U}(\boldsymbol{i}\boldsymbol{j})=U(45)U(12)$.
\item The next one is $U(23)U(12)$, because it differs from $U(45)U(12)$
by the last unlinking.
\item The last one is $U(12)U(23)$; it differs from $U(23)U(12)$ by an application of $(1,1)$-connector and it contains $\boldsymbol{U}(\boldsymbol{k}\boldsymbol{l})=U(23)$.
\end{enumerate}
\item From our list of sequences $\left(U(45)U(12),\,U(23)U(12),\,U(12)U(23)\right)$
we take the pair $(U(45)U(12),\,U(23)U(12))$ and apply the $(1,1)$-connector.
Since there is no repeated index, we obtain
\begin{equation}
U(23)U(45)U(12)=U(45)U(23)U(12),
\end{equation}
which we add to the set of identities. Graphically, the outcome can be shown as (with the newly added part shown in blue)
\begin{equation}
\begin{tikzcd}[column sep={1cm},row sep={1cm}]
	\bullet && \bullet \\
	\bullet & \bullet && \bullet \\
	& \bullet && \bullet \\
	&& \bullet
	\arrow["{(12)}"{description}, from=4-3, to=3-2]
	\arrow["{(23)}"{description}, from=4-3, to=3-4]
	\arrow["{(45)}"{description}, from=3-2, to=2-1]
	\arrow["{(23)}"{description}, from=3-2, to=2-2]
	\arrow["{((12)3)}"{description}, from=2-2, to=1-3]
	\arrow["{(12)}"{description}, from=3-4, to=2-4]
	\arrow["{(1(23))}"{description}, from=2-4, to=1-3]
	\arrow["{(45)}"{description}, color={rgb,255:red,92;green,92;blue,214}, from=2-2, to=1-1]
	\arrow["{(23)}"{description}, color={rgb,255:red,92;green,92;blue,214}, from=2-1, to=1-1]
\end{tikzcd}
\end{equation}
\end{enumerate}
\item For $p=3$, we do the following steps:
\begin{enumerate}
\item Our set of identities 
\begin{equation}
\begin{aligned}
&\left\{\, U((12)3)U(23)U(12)=U(1(23))U(12)U(23),\right. \\
&\left. U(23)U(45)U(12)=U(45)U(23)U(12)\, \right\} \label{eq:identities for p=00003D3}
\end{aligned}
\end{equation}
 is given by sequences of length $3$, so we do not need to cut them.
\item We build our list of sequences in the same way as before and get 
\begin{multline*}
(U(23)U(45)U(12),\;U(45)U(23)U(12),\; \\ U((12)3)U(23)U(12),\;U(1(23))U(12)U(23))\,.
\end{multline*}
\item From our list of sequences we take the pair 
\begin{equation*}
(U(45)U(23)U(12),\;U((12)3)U(23)U(12))
\end{equation*}
and apply the $(1,1)$-connector. Since there is no repeated index,
we obtain
\begin{equation}
U((12)3)U(45)U(23)U(12)=U(45)U((12)3)U(23)U(12).
\end{equation}
Combining it with (\ref{eq:identities for p=00003D3}), we obtain
the following identities:
\begin{equation}
\begin{split}U((12)3)U(23)U(45)U(12)= & U((12)3)U(45)U(23)U(12)\\
= & U(45)U((12)3)U(23)U(12)\\
= & U(45)U(1(23))U(12)U(23)\,.
\end{split}
\label{eq:identities for p=00003D4}
\end{equation}
\item The set of identities contains 
\[
U((12)3)U(23)U(45)U(12)=U(45)U(1(23))U(12)U(23),
\]
which is $\boldsymbol{U}(\boldsymbol{k'}\boldsymbol{l'})\boldsymbol{U}(\boldsymbol{i}\boldsymbol{j})=\boldsymbol{U}(\boldsymbol{i'}\boldsymbol{j'})\boldsymbol{U}(\boldsymbol{k}\boldsymbol{l})$ for $\boldsymbol{U}(\boldsymbol{k'}\boldsymbol{l'})=U((12)3)U(23)$
and $\boldsymbol{U}(\boldsymbol{i'}\boldsymbol{j'})=U(45)U(1(23))U(12)$, so the Connector Algorithm stops. The final stage is depicted as
\begin{equation}
\begin{tikzcd}[column sep={1cm},row sep={1cm}]
& \bullet \\
	\bullet && \bullet \\
	\bullet & \bullet && \bullet \\
	& \bullet && \bullet \\
	&& \bullet
	\arrow["{(12)}"{description}, from=5-3, to=4-2]
	\arrow["{(23)}"{description}, from=5-3, to=4-4]
	\arrow["{(45)}"{description}, from=4-2, to=3-1]
	\arrow["{(23)}"{description}, from=4-2, to=3-2]
	\arrow["{((12)3)}"{description}, from=3-2, to=2-3]
	\arrow["{(12)}"{description}, from=4-4, to=3-4]
	\arrow["{(1(23))}"{description}, from=3-4, to=2-3]
	\arrow["{(45)}"{description}, from=3-2, to=2-1]
	\arrow["{(23)}"{description}, from=3-1, to=2-1]
	\arrow["{(45)}"{description}, color={rgb,255:red,92;green,92;blue,214}, from=2-3, to=1-2]
	\arrow["{((12)3)}"{description}, color={rgb,255:red,92;green,92;blue,214}, from=2-1, to=1-2]
\end{tikzcd}
\end{equation}
\end{enumerate}
\end{enumerate}

\section{Unlinking and permutohedra graphs} \label{sec:unlinking}

In this section we show how permutohedra graphs, first discovered in the context of the knots-quivers correspondence,
arise from unlinking of generic quivers. We define a~permutohedron constructed by unlinking using the tools from Sec. \ref{sec:commutation relations}, and then
propose a systematic procedure to obtain a large class of permutohedra graphs made of several permutohedra.

\subsection{General permutohedron constructed by unlinking}\label{sec:pi_from_U}

Permutohedra graphs for symmetric quivers were introduced in the context of the knots-quivers correspondence as graphs that
encode equivalent quivers for the same knot (see Sec.~\ref{sec:KQ correspondence and permutohedra}, App.~\ref{sec:Statements from the permutohedra paper}, and \cite{JKLNS2105}). However, in this paper we undertake a~different approach based only on the properties of unlinking. This suggests the following general definition:

\begin{dfn}\label{dfn:Permutohedron}
    For $n\in \mathbb{Z}_+$, let $Q$ be a~quiver with at least $n+1$ vertices. Choose $n+1$ distinct vertices of $Q$, denote them as $\iota$, $j_1$, $j_2$, $\dots$, $j_n$, and consider quivers
\begin{equation}
\label{eq:permutohedron from unlinking}
    Q_{\sigma}=U(\iota j_{\sigma(n)})\dots U(\iota j_{\sigma(2)})U(\iota j_{\sigma(1)})Q
\end{equation}
for all permutations $\sigma\in S_n$. We say that such  $\{Q_{\sigma}\}_{\sigma\in S_n}$ form permutohedron $\Pi_n$.
\end{dfn}
We can see that in this construction permutations are very explicit -- the connection with the associated polytope $\Pi_n$ can be seen in the following statement.

\begin{prp}\label{prp:Permutohedron by unlinking}
    If quivers $\{Q_{\sigma}\}_{\sigma\in S_n}$ form permutohedron $\Pi_n$, then there exists a~bijection $\phi : \{Q_{\sigma}\}_{\sigma\in S_n}\to \Pi_n$ such that: 
\begin{itemize}
    \item Each quiver corresponds to a~different vertex of the permutohedron $\Pi_n$.
    \item Each pair of quivers that differ by a~transposition of neighbouring unlinkings corresponds to an edge of the permutohedron $\Pi_n$. 
\end{itemize}
Moreover, the pair of quivers connected by an edge can be unlinked to the same quiver by a~single unlinking.
\end{prp}

\begin{proof}
The map $\phi$ given by
\[
   \{Q_{\sigma}\}_{\sigma\in S_n} \ni U(\iota j_{\sigma(n)})\dots U(\iota j_{\sigma(2)})U(\iota j_{\sigma(1)})Q = Q_{\sigma} \overset{\phi}{\longmapsto} \sigma \in \Pi_n
\]
is a~ bijection in which each quiver corresponds to a~different vertex of the permutohedron $\Pi_n$ and each pair of quivers that differ by a~transposition of neighbouring unlinkings corresponds to an edge of the permutohedron $\Pi_n$.

Therefore, we need to show that for any $k\in \{1,2,\dots,n-1\}$  a~pair of quivers\footnote{We use shorter notation to denote quivers $U( \sigma(n)\iota)\dots U(\iota j_{\sigma(k+1)})U(\iota j_{\sigma(k)})\dots U( \sigma(1) \iota) Q$ and $U( \sigma(n)\iota) \dots U(\iota j_{\sigma(k)})U(\iota j_{\sigma(k+1)})\dots U( \sigma(1) \iota) Q$ respectively.}
\begin{equation*}
\dots U(\iota j_{\sigma(k+1)})U(\iota j_{\sigma(k)})\dots Q  \qquad\text{and}\qquad  \dots U(\iota j_{\sigma(k)})U(\iota j_{\sigma(k+1)})\dots Q
\end{equation*}
can be unlinked to the same quiver. 
 
 Consider unlinkings $U((j_{\sigma(k)}\iota)j_{\sigma(k+1)})$ and $U(j_{\sigma(k)}(\iota j_{\sigma(k+1)}))$. Since nodes $j_{\sigma(k)}$, $j_{\sigma(k+1)}$, $( j_{\sigma(k)}\iota)$, $(\iota j_{\sigma(k+1)})$ are not unlinked by operators located on the left of transposed pair\footnote{These operators are  $U(\iota j_{\sigma(n)})$, $U(\iota j_{\sigma(n-1)})$,$\dots$, $U(\iota j_{\sigma(k+3)})$, $U(\iota j_{\sigma(k+2)})$ .}, they all commute with unlinkings $U(( j_{\sigma(k)}\iota)j_{\sigma(k+1)})$ and $U(j_{\sigma(k)}(\iota j_{\sigma(k+1)}))$. In consequence the hexagon relation leads to
 \begin{equation}
\begin{split}
    U(( j_{\sigma(k)}\iota)j_{\sigma(k+1)})  \dots & U(\iota j_{\sigma(k+1)})U(\iota j_{\sigma(k)})\dots Q \\
    & = \dots U(( j_{\sigma(k)}\iota)j_{\sigma(k+1)}) U(\iota j_{\sigma(k+1)})U(\iota j_{\sigma(k)})\dots Q \\
    & = \dots U(j_{\sigma(k)}(\iota j_{\sigma(k+1)})) U(\iota j_{\sigma(k)}) U(\iota j_{\sigma(k+1)})\dots Q \\
    & =  U(j_{\sigma(k)}(\iota j_{\sigma(k+1)})) \dots U(\iota j_{\sigma(k)}) U(\iota j_{\sigma(k+1)})\dots Q \,.
\end{split}
\end{equation}
This means that each edge of the permutohedron corresponds to a~pair of quivers that can be unlinked to the same quiver, which we wanted to show.
\end{proof}

\subsection{Example -- permutohedron \texorpdfstring{$\Pi_2$}{Pi2}}\label{sec:Example - Pi2}



Let us illustrate Definition~\ref{dfn:Permutohedron} and Proposition~\ref{prp:Permutohedron by unlinking} with the example of permutohedron~$\Pi_2$. We have $n=2$ and consider arbitrary quiver $Q$ with at least 3 nodes. We pick three nodes of~$Q$, denote them as $\iota, j_1, j_2$, and consider quivers 
\begin{equation}
     Q_{(1,2)}=U(\iota j_{2})U(\iota j_{1})Q\ , \qquad \qquad Q_{(2,1)}=U(\iota j_{1})U(\iota j_{2})Q\ ,
\end{equation}
where $(1,2)$ and $(2,1)$ are permutations of two elements. Comparing with Definition~\ref{dfn:Permutohedron}, we can see that quivers $Q_{(1,2)}$ and $Q_{(2,1)}$ form permutohedron $\Pi_2$.

Moving to Proposition~\ref{prp:Permutohedron by unlinking}, we can see that the map $\phi$ given by
\begin{equation}
\begin{split}
    U(\iota j_{2})U(\iota j_{1})Q=Q_{(1,2)} & \overset{\phi}{\longmapsto} (1,2) \in \Pi_2\ , \\
    U(\iota j_{1})U(\iota j_{2})Q=Q_{(2,1)} & \overset{\phi}{\longmapsto} (2,1) \in \Pi_2
\end{split}
\end{equation}
is a~ bijection in which each quiver corresponds to a~vertex of the permutohedron $\Pi_2$. Moreover, $Q_{(1,2)}$ and $Q_{(2,1)}$ are connected by transposition of $U(\iota j_{1})$ and $U(\iota j_{2})$, which corresponds to an edge of the permutohedron connecting these vertices:

\[\begin{tikzcd}
	{Q_{(1,2)}} && {Q_{(2,1)}} \\
	{U(\iota j_1)Q} && {U(\iota j_2)Q} \\
	& Q
	\arrow["{\iota j_1}", from=3-2, to=2-1]
	\arrow["{\iota j_2}", from=2-1, to=1-1]
	\arrow["{\iota j_2}"', from=3-2, to=2-3]
	\arrow["{\iota j_1}"', from=2-3, to=1-3]
\end{tikzcd}\]

Finally, from the hexagon relation we know that 
\begin{equation}
    U(j_{2}(\iota j_{1}))U(\iota j_{2})U(\iota j_{1}) = U((j_{2} \iota ) j_{1})U(\iota j_{1})U(\iota j_{2})\ ,
\end{equation}
so $Q_{(1,2)}$ and $Q_{(2,1)}$ can be unlinked to the same quiver by one unlinking:
\begin{equation}
\label{eq:Pi_2}
    U(j_{2}(\iota j_{1}))Q_{(1,2)} = U((j_{2} \iota ) j_{1})Q_{(2,1)} \,.
\end{equation}

\subsection{Example -- permutohedron \texorpdfstring{$\Pi_3$}{Pi3}}\label{sec:Creating permutohedron Pi3}

Let us move to permutohedron $\Pi_3$ and show how it is created by unlinkings. Following Definition~\ref{dfn:Permutohedron}, we consider $U(\iota j_{\sigma(3)})U(\iota j_{\sigma(2)}) U(\iota j_{\sigma(1)})Q$ for all permutations $\sigma\in S_3$ and generic quiver $Q$ (we only assume that it contains at least four nodes). If we we relabel $(j_1,j_2,j_3)$ as $(j,k,l)$, the application of three unlinkings in different order can be represented by the top-left diagram in Fig. \ref{fig:pi3_bottom_stages}.

Now we can verify that it indeed satisfies the two properties from Proposition~\ref{prp:Permutohedron by unlinking}. The end nodes of this initial ``skeleton'' correspond to six elements of $\Pi_3$. Each pair of neighbouring operators can be appended once to yield an identity (applying this to quivers, it means that those which form an edge of permutohedron, can be unlinked once to the same quiver).
In order to show it, we can start from computing the connector between $U(\iota l)U(\iota k)U(\iota j)$ and $U(\iota k)U(\iota l)U(\iota j)$, which amounts to applying the hexagon relation
\begin{equation}
    U(l(\iota k))U(\iota l)U(\iota k)U(\iota j) = U(k(\iota l))U(\iota k)U(\iota l)U(\iota j).
\end{equation}
Applying this to every permutation of $(j,k,l)$, we get the top-right diagram in Fig. \ref{fig:pi3_bottom_stages}.
\noindent We still have to connect other transpositions of neighbouring  operators, such as 
\[U(\iota l)U(\iota k)U(\iota j) \leftrightarrow U(\iota l)U(\iota j)U(\iota k)\,. \]
From the proof of Proposition~\ref{prp:Permutohedron by unlinking} we know that it can be done by application of
\begin{equation}
    U(k(\iota j))U(\iota l)U(\iota k)U(\iota j) = U(j(\iota k))U(\iota l)U(\iota j)U(\iota k)\, .
\end{equation}
This gives the filled-in diagram shown in Fig. \ref{fig:pi3_bottom_stages}, bottom.

\begin{figure}[h!]
    \centering
    \includegraphics[scale=0.75]{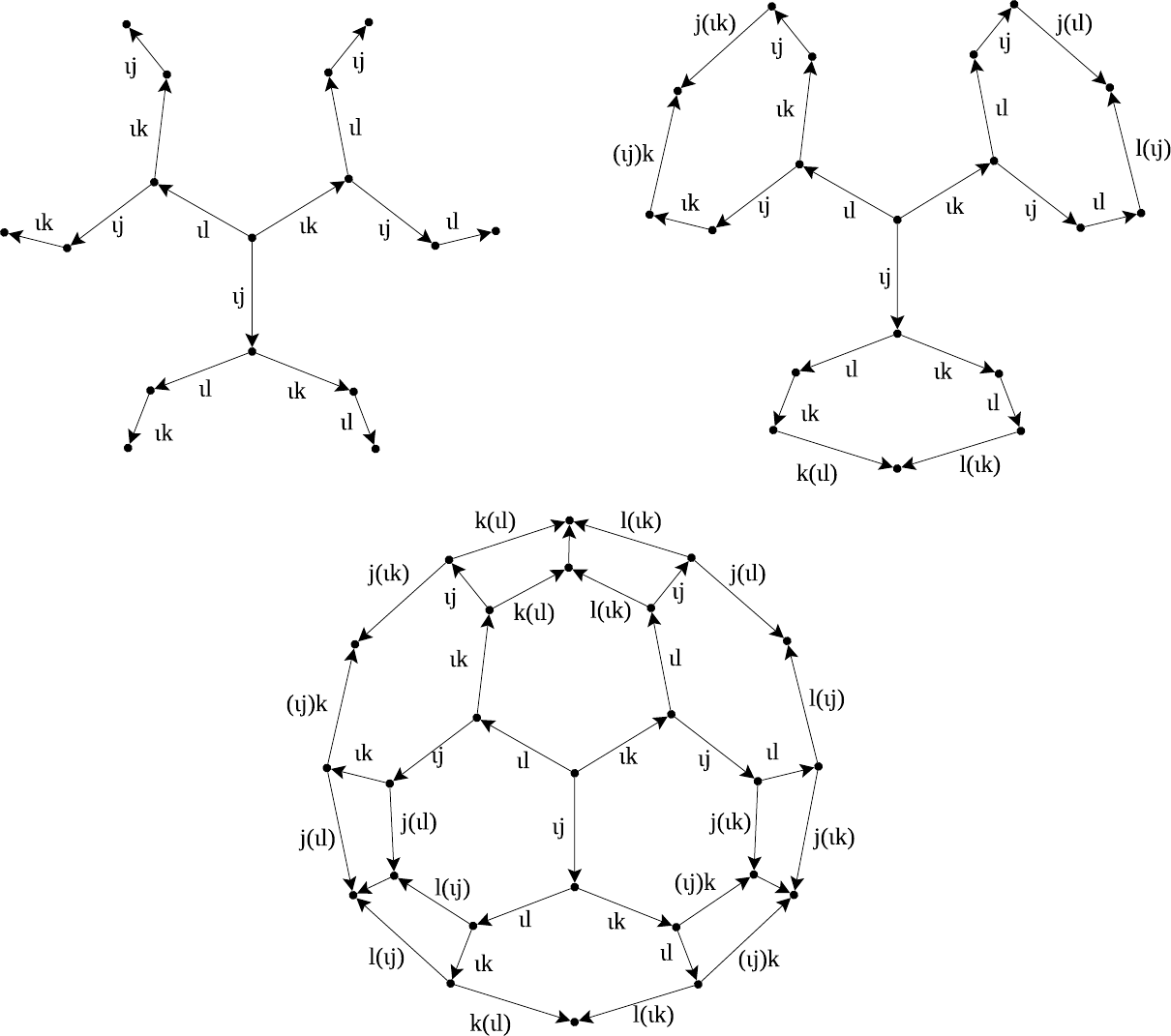}
    \caption{Three stages of the permutohedron $\Pi_3$ from unlinking.}
    \label{fig:pi3_bottom_stages}
\end{figure}

At this point our job is done -- we thus have constructed permutohedron $\Pi_3$, determined how
the nodes are organized, and what unlinkings have to be performed for them to be pairwise connected.

\subsection{Permutohedra graph constructed by unlinking}

In \cite{JKLNS2105} there is an analysis of graphs made of several permutohedra glued together and corresponding to the same knot. We will follow this path, but using only quivers and unlinking.

\begin{dfn}\label{dfn:permutohedra graph}
We say that quivers $ \{Q_{\sigma}\}_{\sigma = 1,\dots, m}$ form a~permutohedra graph if
\begin{itemize}
    \item $\{Q_{\sigma}\}_{\sigma = 1,\dots, m}$ is a~union of sets of quivers $\{Q_{\rho}\}_{\rho\in S_{n_v}}$ that form permutohedra $\Pi_{n_v}$:\footnote{Note that we label permutations by numbers $ 1,\dots, m$.}
        \[
            \{Q_{\sigma}\}_{\sigma = 1,\dots, m} = \bigcup_v \{Q_{\rho}\}_{\rho\in S_{n_v}}
        \]
    \item The graph associated to $ \{Q_{\sigma}\}_{\sigma = 1,\dots, m}$ by bijection $\phi$ is connected.
\end{itemize}
\end{dfn}

Let us see how we can build a~large class of permutohedra graphs by applying unlinking to a~generic quiver.

\begin{prp}\label{prp:Permutohedra graph}
     Consider a~tree (i.e. a~connected acyclic undirected graph) with $p$~vertices and assign a~permutation group $S_{n_v}$, $n_v\in \mathbb{Z}_+$, to each vertex $v$ and a~pair of permutations $\pi_v\in S_{n_v},\pi_w\in S_{n_w}$ to an edge connecting vertices $v$ and $w$. If we fix a~vertex $v$ and consider paths to all other vertices, then this assignment induces a~map 
     \[
     \psi: v \longmapsto (\pi_1, \pi_2, \dots, \pi_{v-1}, \pi_{v+1}, \dots, \pi_p)\ ,
     \]
     where each $\pi_u$ is a~permutation assigned to the other end of the path that connects $v$ and~$u$.\footnote{For example, if we have a~tree $v$---$w$---$u$ with permutations $\pi_v, \pi_w$ assigned to $v$---$w$ and $\pi'_w, \pi_u$ assigned to $w$---$u$, then $\psi(v)=(\pi_w,\pi_u)$.}

     Let $Q$ be a~quiver with at least $p+\sum_{v=1}^p n_v$ nodes. Choose $p+\sum_{v=1}^p n_v$ nodes of $Q$, denote them as $\iota_v$, $j_{v,1}$, $j_{v,2}$, $\dots$, $j_{v,n_v}$, and consider a~sequence of unlinkings
    \begin{equation}
    \boldsymbol{U}(\boldsymbol{\iota_v}\boldsymbol{j_{\rho}})=U(\iota_v j_{v,\rho(n_v)})\dots U(\iota_v j_{v,\rho(2)})U(\iota_v j_{v,\rho(1)})    
    \end{equation}
    for each $v=1,\dots,p$ and permutation $\rho \in S_{n_v}$. 
    If we define
    \begin{equation}
    Q_{v,\rho} = \boldsymbol{U}(\boldsymbol{\iota_v}\boldsymbol{j_{\rho}})      \boldsymbol{U}(\boldsymbol{\iota_p}\boldsymbol{j_{\pi_p}}) \dots \boldsymbol{U}(\boldsymbol{\iota_{v+1}}\boldsymbol{j_{\pi_{v+1}}}) \boldsymbol{U}(\boldsymbol{\iota_{v-1}}\boldsymbol{j_{\pi_{v-1}}}) \dots \boldsymbol{U}(\boldsymbol{\iota_1}\boldsymbol{j_{\pi_1}}) Q \ ,
    \end{equation}
    where $(\pi_1, \pi_2, \dots, \pi_{v-1}, \pi_{v+1}, \dots, \pi_p)=\psi(v)$, then quivers
    \[
    \bigcup_{v=1}^p \{Q_{v,\rho}\}_{\rho\in S_{n_v}}
    \]
    form a~permutohedra graph.
\end{prp}

\begin{proof}
    If we write 
    \begin{equation}\label{eq:permutohedra in the graph}
    \begin{split}
        Q_{v,\rho} & =  U(\iota_v j_{v,\rho(n_v)})\dots U(\iota_v j_{v,\rho(2)})U(\iota_v j_{v,\rho(1)})Q' \ ,\\
        Q' & = \boldsymbol{U}(\boldsymbol{\iota_p}\boldsymbol{j_{\pi_p}}) \dots \boldsymbol{U}(\boldsymbol{\iota_{v+1}}\boldsymbol{j_{\pi_{v+1}}}) \boldsymbol{U}(\boldsymbol{\iota_{v-1}}\boldsymbol{j_{\pi_{v-1}}}) \dots \boldsymbol{U}(\boldsymbol{\iota_1}\boldsymbol{j_{\pi_1}}) Q \ ,
    \end{split}
    \end{equation} 
    then we can see that for each $v=1,\dots,p$ quivers $\{Q_{v,\rho}\}_{\rho\in S_{n_v}}$ form permutohedron~$\Pi_{n_v}$ according to Definition~\ref{dfn:Permutohedron}. This means that vertices of the tree correspond to permutohedra. In order to see that edges correspond to quivers that are gluing points, consider arbitrary edge connecting vertices $v$ and $w$ with assigned permutations $\pi_v, \pi_w$, and the following quivers:
    \begin{equation}\label{eq:gluing point quivers}
    \begin{split}
        Q_{v,\pi_v} & =  \boldsymbol{U}(\boldsymbol{\iota_v}\boldsymbol{j_{\pi_v}}) \boldsymbol{U}(\boldsymbol{\iota_p}\boldsymbol{j_{\pi_p}}) \dots \boldsymbol{U}(\boldsymbol{\iota_{v+1}}\boldsymbol{j_{\pi_{v+1}}}) \boldsymbol{U}(\boldsymbol{\iota_{v-1}}\boldsymbol{j_{\pi_{v-1}}}) \dots \boldsymbol{U}(\boldsymbol{\iota_1}\boldsymbol{j_{\pi_1}}) Q \ ,\\
        Q_{w,\pi_w} & =  \boldsymbol{U}(\boldsymbol{\iota_w}\boldsymbol{j_{\pi_w}}) \boldsymbol{U}(\boldsymbol{\iota_p}\boldsymbol{j_{\pi_p}}) \dots \boldsymbol{U}(\boldsymbol{\iota_{w+1}}\boldsymbol{j_{\pi_{w+1}}}) \boldsymbol{U}(\boldsymbol{\iota_{w-1}}\boldsymbol{j_{\pi_{w-1}}}) \dots \boldsymbol{U}(\boldsymbol{\iota_1}\boldsymbol{j_{\pi_1}}) Q\,.
    \end{split}
    \end{equation}     
    Since $v$ and $w$ are connected by a~single edge, $\psi(v)$ amended by $\pi_v$ and $\psi(w)$ amended by $\pi_w$ are both equal to\footnote{We implicitly assumed $v<w$, but one can exchange $v$ and $w$ and repeat the whole argument for $w>v$.}
    \[
        (\pi_1, \pi_2, \dots, \pi_{v-1}, \pi_v, \pi_{v+1}, \dots, \pi_{w-1}, \pi_w, \pi_{w+1},\dots, \pi_p) \,.
    \]
    On the other hand, we can see that $\boldsymbol{U}(\boldsymbol{\iota_u}\boldsymbol{j_{\pi_u}})$ commute with each other for all $u=1,2,\dots, p$, which comes from the fact that in different sequences all indices are different (note that unlinkings inside each sequence are kept in the same order determined by permutations $\rho$ and $\sigma$).
    Combining these two facts we learn that
    \begin{equation}
         Q_{v,\pi_v} = \boldsymbol{U}(\boldsymbol{\iota_p}\boldsymbol{j_{\pi_p}}) \dots \boldsymbol{U}(\boldsymbol{\iota_{w}}\boldsymbol{j_{\pi_{w}}}) \dots \boldsymbol{U}(\boldsymbol{\iota_{v}}\boldsymbol{j_{\pi_{v}}}) \dots  \boldsymbol{U}(\boldsymbol{\iota_1}\boldsymbol{j_{\pi_1}}) Q = Q_{w,\pi_w} \,.
    \end{equation}
    From Eqs. (\ref{eq:permutohedra in the graph}-\ref{eq:gluing point quivers}) we know that $\phi(Q_{v,\pi_v})\in \Pi_{n_v}$ and $\phi(Q_{w,\pi_w})\in \Pi_{n_w}$, so the quiver $Q_{v,\pi_v} = Q_{w,\pi_w}$ can be viewed (via bijection $\phi$) as a~gluing point of permutohedra $\Pi_{n_v}$ and $\Pi_{n_w}$, represented by the edge that connects vertices $v$ and $w$ in our tree.

    Since all vertices of the tree are connected by exactly one path, the definition of quivers $Q_{v,\rho}$ is consistent and the graph associated to $\bigcup_{v=1}^p \{Q_{v,\rho}\}_{\rho\in S_{n_v}}$ by bijection $\phi$ is connected. This means that quivers $\bigcup_{v=1}^p \{Q_{v,\rho}\}_{\rho\in S_{n_v}}$ form a~permutohedra graph according to Definition~\ref{dfn:permutohedra graph}.
\end{proof}

\subsection{Example -- permutohedra graph containing two \texorpdfstring{$\Pi_2$}{Pi2}}\label{sec:Two Pi_2}

Let us see how the construction given in Proposition~\ref{prp:Permutohedra graph} looks in the case of the permutohedra graph containing two $\Pi_2$. 

We start from a~tree 
\[1\text{---}2\] 
composed of $p=2$ vertices connected by an edge. We choose $n_1=n_2=2$, which means that we assign permutation group $S_2$ to each vertex. We set $\pi_1 = \pi_2 = (1,2) \in S_2$, which means that we assign a~trivial permutation of two elements to each end of the edge. The path connecting vertices $1$ and $2$ contains just a~single edge, so
\begin{equation}
    \psi(1) = (\pi_2) = ((1,2)), \qquad \psi(2) = (\pi_1) = ((1,2)) \,.
\end{equation}
We take any quiver $Q$ with at least $p+n_1+n_2=6$ nodes, choose six of them, and denote as $\iota_1,j_{1,1},j_{1,2}, \iota_2,j_{2,1},j_{2,2}$. Then we use the following sequences of unlinking:
\begin{align}\label{eq: two Pi_2}
    \boldsymbol{U}(\boldsymbol{\iota_1}\boldsymbol{j_{(1,2)}}) &= U(\iota_1 j_{1,2})U(\iota_1 j_{1,1}) \ , &
    \boldsymbol{U}(\boldsymbol{\iota_2}\boldsymbol{j_{(1,2)}}) &= U(\iota_2 j_{2,2})U(\iota_2 j_{2,1}) \ , \\
    \boldsymbol{U}(\boldsymbol{\iota_1}\boldsymbol{j_{(2,1)}}) &= U(\iota_1 j_{1,1})U(\iota_1 j_{1,2}) \ , &
    \boldsymbol{U}(\boldsymbol{\iota_2}\boldsymbol{j_{(2,1)}}) &= U(\iota_2 j_{2,1})U(\iota_2 j_{2,2}) \ , \nonumber
\end{align}
to define quivers 
\begin{align}\label{eq: two Pi_2 cont'd}
    Q_{1,(1,2)} &= \boldsymbol{U}(\boldsymbol{\iota_1}\boldsymbol{j_{(1,2)}})\boldsymbol{U}(\boldsymbol{\iota_2}\boldsymbol{j_{(1,2)}})Q \ , &
    Q_{2,(1,2)} &= \boldsymbol{U}(\boldsymbol{\iota_2}\boldsymbol{j_{(1,2)}})\boldsymbol{U}(\boldsymbol{\iota_1}\boldsymbol{j_{(1,2)}})Q \ , \\
    Q_{1,(2,1)} &= \boldsymbol{U}(\boldsymbol{\iota_1}\boldsymbol{j_{(2,1)}})\boldsymbol{U}(\boldsymbol{\iota_2}\boldsymbol{j_{(1,2)}})Q \ , &
    Q_{2,(2,1)} &= \boldsymbol{U}(\boldsymbol{\iota_2}\boldsymbol{j_{(2,1)}})\boldsymbol{U}(\boldsymbol{\iota_1}\boldsymbol{j_{(1,2)}})Q \,. \nonumber
\end{align}
Note that the last (i.e. the leftmost) sequence of unlinkings acting on $Q$ follows indices of the quiver. The remaining sequence of unlinkings acting on $Q$ (the one next to $Q$) is determined by $\psi(1)=(\pi_2)=((1,2))$ in the left column and $\psi(2)=(\pi_1)=((1,2))$ in the right column. 

Now we will follow the reasoning from the proof of Proposition~\ref{prp:Permutohedra graph} to see that quivers $\{Q_{1,(1,2)}, Q_{1,(2,1)} \} \cup \{Q_{2,(1,2)}, Q_{2,(2,1)} \}$ form a~permutohedra graph composed of two $\Pi_2$ glued together. 

If we recall Sec.~\ref{sec:Creating permutohedron Pi3} and compare it with
\begin{align}\label{eq:Two Pi_2 from definition}
    Q_{1,(1,2)} &= U(\iota_1 j_{1,2})U(\iota_1 j_{1,1})Q' \ , &
    Q_{2,(1,2)} &= U(\iota_2 j_{2,2})U(\iota_2 j_{2,1})Q'' \ , \\
    Q_{1,(2,1)} &= U(\iota_1 j_{1,1})U(\iota_1 j_{1,2})Q' \ , &
    Q_{2,(2,1)} &= U(\iota_2 j_{2,1})U(\iota_2 j_{2,2})Q'' \ , \nonumber
\end{align}
where $Q'=\boldsymbol{U}(\boldsymbol{\iota_2}\boldsymbol{j_{(1,2)}})Q$ and $Q''=\boldsymbol{U}(\boldsymbol{\iota_1}\boldsymbol{j_{(1,2)}})Q$ \footnote{Quivers $Q'$ and $Q''$ are determined by $\psi(1)$ and $\psi(2)$ respectively -- see the previous paragraph and Eq. \eqref{eq: two Pi_2 cont'd}.}, then we can see that $\{Q_{1,(1,2)}, Q_{1,(2,1)} \}$ form permutohedron $\Pi_2$ which corresponds to vertex $1$ in the tree. Similarly, the bijection $\phi$ maps $\{Q_{2,(1,2)}, Q_{2,(2,1)} \}$ to $\Pi_2$ which corresponds to vertex $2$ in the tree.

In the next step we focus on the quivers $Q_{1,(1,2)}$ and $Q_{2,(1,2)}$. Since $\psi(1)$ amended by $\pi_1$ and $\psi(2)$ amended by $\pi_2$ are both equal to $((1,2),(1,2))$ and $\boldsymbol{U}(\boldsymbol{\iota_1}\boldsymbol{j_{(1,2)}})$ commutes with $\boldsymbol{U}(\boldsymbol{\iota_2}\boldsymbol{j_{(1,2)}})$, we have
\begin{equation}
    Q_{1,(1,2)} = \boldsymbol{U}(\boldsymbol{\iota_1}\boldsymbol{j_{(1,2)}})\boldsymbol{U}(\boldsymbol{\iota_2}\boldsymbol{j_{(1,2)}})Q = \boldsymbol{U}(\boldsymbol{\iota_2}\boldsymbol{j_{(1,2)}})\boldsymbol{U}(\boldsymbol{\iota_1}\boldsymbol{j_{(1,2)}})Q = Q_{2,(1,2)} \,.
\end{equation}
We can see that $Q_{1,(1,2)}$ and $ Q_{2,(1,2)}$ are really the same quiver which can be viewed (by bijection $\phi$) as a~gluing point of two $\Pi_2$, represented by the edge that connects vertices $1$ and $2$ in our tree. The whole construction is summarised in Fig.~\ref{fig:two Pi2 s}.

\begin{figure}
\[\begin{tikzcd}[column sep={1cm},row sep={1cm}]
	&& \bullet && \bullet \\
	& \textcolor{black}{Q_{1,(2,1)}} && \textcolor{black}{Q_{1,(1,2)}=Q_{2,(1,2)}} && \textcolor{black}{Q_{2,(2,1)}} \\
	\bullet && \bullet && \bullet && \bullet \\
	& \bullet && \bullet && \bullet \\
	&& \bullet && \bullet \\
	&&& Q
	\arrow["j_{22}\iota_2"{description}, from=4-6, to=3-7]
	\arrow["j_{21}\iota_2"{description}, from=4-6, to=3-5]
	\arrow["j_{12}\iota_1"{description}, from=4-2, to=3-1]
	\arrow["j_{11}\iota_1"{description}, from=4-2, to=3-3]
	\arrow["j_{12}\iota_1"{description}, from=3-3, to=2-4]
	\arrow["j_{22}\iota_2"{description}, from=3-5, to=2-4]
	\arrow["j_{22}\iota_2"{description}, from=4-4, to=3-3]
	\arrow["j_{12}\iota_1"{description}, from=4-4, to=3-5]
	\arrow["j_{11}\iota_1"{description}, from=5-3, to=4-4]
	\arrow["j_{22}\iota_2"{description}, from=5-3, to=4-2]
	\arrow["j_{21}\iota_2"{description}, from=5-5, to=4-4]
	\arrow["j_{12}\iota_1"{description}, from=5-5, to=4-6]
	\arrow["j_{11}\iota_1"{description}, from=6-4, to=5-5]
	\arrow["j_{21}\iota_2"{description}, from=6-4, to=5-3]
	\arrow[from=2-4, to=1-3]
	\arrow[from=2-4, to=1-5]
	\arrow[from=2-2, to=1-3]
	\arrow["j_{11}\iota_1"{description}, from=3-1, to=2-2]
	\arrow["j_{21}\iota_2"{description}, from=3-7, to=2-6]
	\arrow[from=2-6, to=1-5]
\end{tikzcd}\]

\caption{Construction of a~permutohedra graph which consists of two copies of $\Pi_2$ glued at a~point. The vertices of such graph are formed by quivers $Q_{1,(2,1)}$, $Q_{1,(1,2)}=Q_{2,(2,2)}$ and $Q_{2,(2,1)}$.}

\label{fig:two Pi2 s}

\end{figure}


\section{Universal quivers for permutohedra graphs constructed by unlinking}\label{sec:Universal quivers}

In the remaining part of the paper we focus on universal quivers, which we define as follows:

\begin{dfn}
     Quiver $\hat Q$ is called a~universal quiver for quivers $Q_1,Q_2,\dots,Q_m$ if there exist sequences of unlinking $\boldsymbol{U}(\boldsymbol{i_1}\boldsymbol{j_{1}}),\boldsymbol{U}(\boldsymbol{i_2}\boldsymbol{j_{2}}),\dots,\boldsymbol{U}(\boldsymbol{i_m}\boldsymbol{j_{m}})$ such that
\begin{equation}
\hat{Q}=\boldsymbol{U}(\boldsymbol{i_{k}}\boldsymbol{j_{k}})Q_{k}\qquad\forall k=1,\dots,m.
\end{equation}
\end{dfn}

In this section we show the existence of universal quivers for permutohedra constructed by unlinking, i.e. following Definition \ref{dfn:Permutohedron}, and illustrate it with examples of permutohedra $\Pi_2$ and $\Pi_3$. Similarly, we prove the existence of universal quivers for permutohedra graphs constructed by unlinking, i.e. following Definition \ref{dfn:permutohedra graph} and Proposition \ref{prp:Permutohedra graph}, and apply it to the case of permutohedra graph containing two $\Pi_2$.

\subsection{Proof of the Permutohedron Theorem}

In the first step of our analysis, we combine the Connector Theorem with the fact that quivers forming permutohedron arise from unlinking the same quiver in different order to prove Theorem~\ref{thm:Permutohedron Theorem}.

\begin{proof}[Proof of the Permutohedron Theorem]
From Definition~\ref{dfn:Permutohedron} we know that for every $\{Q_{\sigma}\}_{\sigma\in S_n}$ there exists a~sequence of unlinkings $U(\iota j_{\sigma(n)})\dots U(\iota j_{\sigma(2)})U(\iota j_{\sigma(1)})$ and a~quiver $Q$ such that:
\begin{equation}
    Q_{\sigma}=U(\iota j_{\sigma(n)})\dots U(\iota j_{\sigma(2)})U(\iota j_{\sigma(1)})Q \,.
\end{equation}

From the Connector Theorem we know that for each pair of permutations $\sigma_{1},\sigma_{2}\in S_{n}$ there exist finite sequences of unlinking $\boldsymbol{U}(\boldsymbol{k_{1}}\boldsymbol{l_{1}}),\boldsymbol{U}(\boldsymbol{k_{2}}\boldsymbol{l_{2}})$ such that 
\begin{equation}
\begin{split}\boldsymbol{U}(\boldsymbol{k_{1}}\boldsymbol{l_{1}})U(\iota j_{\sigma_1(n)})\dots U(\iota j_{\sigma_1(1)})Q &= \boldsymbol{U}(\boldsymbol{k_{2}}\boldsymbol{l_{2}})U(\iota j_{\sigma_2(n)})\dots U(\iota j_{\sigma_2(1)})Q\,,\\
\boldsymbol{U}(\boldsymbol{k_{1}}\boldsymbol{l_{1}})Q_{\sigma_{1}} &= \boldsymbol{U}(\boldsymbol{k_{2}}\boldsymbol{l_{2}})Q_{\sigma_{2}}\,.
\end{split}
\end{equation}
If the resulting quivers are different for some pairs, we can use the Connector Theorem again (this time starting with respective $\boldsymbol{U}(\boldsymbol{k}\boldsymbol{l})U(\iota j_{\sigma(n)})\dots U(\iota j_{\sigma(1)})$), and repeating this procedure we eventually obtain sequences of unlinking $\boldsymbol{U}(\boldsymbol{i_{\sigma}}\boldsymbol{j_{\sigma}})$ and a~quiver $\hat{Q}$ such that\footnote{We label unlinkings $\boldsymbol{U}(\boldsymbol{i_{\sigma}}\boldsymbol{j_{\sigma}})$ by permutations $\sigma \in S_n$. Equivalently, we could number all permutations in $S_n$ and assign these numbers to unlinkings and quivers.} 
\begin{equation}\label{eq:Universal quiver permutohedron}
\boldsymbol{U}(\boldsymbol{i_{\sigma}}\boldsymbol{j_{\sigma}})Q_{\sigma}=\hat{Q}\qquad\forall\sigma \in S_n\ ,
\end{equation}
which we wanted to show.    
\end{proof}

\subsection{Examples -- universal quivers for permutohedra \texorpdfstring{$\Pi_2$}{Pi2} and \texorpdfstring{$\Pi_3$}{Pi3}}\label{sec:Contracting permutohedron Pi3}

\subsubsection*{Universal quiver for \texorpdfstring{$\Pi_2$}{Pi2}}
Comparing the proof of the Permutohedron Theorem with Sec.~\ref{sec:Creating permutohedron Pi3}, we can see that the application of $\boldsymbol{U}(\boldsymbol{k_{1}}\boldsymbol{l_{1}})= U(j_{2}(\iota j_{1}))$ to $Q_{\sigma_{1}}=Q_{(1,2)}$ and $\boldsymbol{U}(\boldsymbol{k_{1}}\boldsymbol{l_{1}})=U((j_{2}\iota )j_{1})$ to $Q_{\sigma_{2}}=Q_{(2,1)}$ leads to the same quiver and there are no more permutations in $S_2$. This means that the universal quiver and corresponding unlinkings are given by
\begin{equation}
\begin{split}
    \hat{Q} &= \boldsymbol{U}(\boldsymbol{i_{(1,2)}}\boldsymbol{j_{(1,2)}})Q_{(1,2)} = U(j_{2}(\iota j_{1}))Q_{(1,2)} \\
    &= \boldsymbol{U}(\boldsymbol{i_{(2,1)}}\boldsymbol{j_{(2,1)}})Q_{(2,1)} = U((j_{2}\iota )j_{1})Q_{(1,2)} \,.
\end{split}
\end{equation}

\subsubsection*{Universal quiver for \texorpdfstring{$\Pi_3$}{Pi3}}

We proceed with a construction of the universal quiver for permutohedron $\Pi_3$ which we have constructed in Sec.~\ref{sec:Creating permutohedron Pi3} using permutations of unlinkings. Our starting point is the final diagram from Sec.~\ref{sec:Creating permutohedron Pi3} where yellow dots mark quivers which form permutohedron~$\Pi_3$ (according to Definition~\ref{dfn:Permutohedron}).

\begin{figure}[h!]
    \centering
    \includegraphics[scale=0.75]{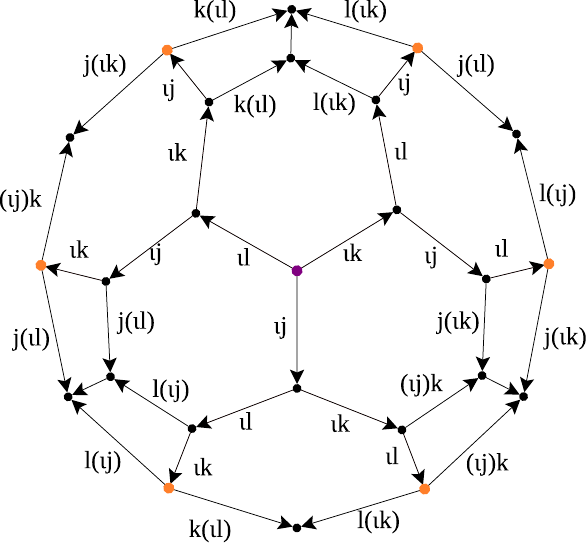}
\end{figure}

From now on, we only need to focus on those unlinkings which emanate from the yellow dots. Thus, denoting $a=(\iota j), b=(\iota k), c=(\iota l)$, we get the top-left diagram in Fig.~\ref{fig:pi3_contraction}.
The rest of our construction corresponds to moving inwards, thus obtaining a~3d graph in which the ``equator'' corresponds to permutohedron $\Pi_3$ (see Fig.~\ref{fig:Pi3_Rocket}).
\begin{figure}[h!]
    \centering
    \includegraphics[scale=0.75]{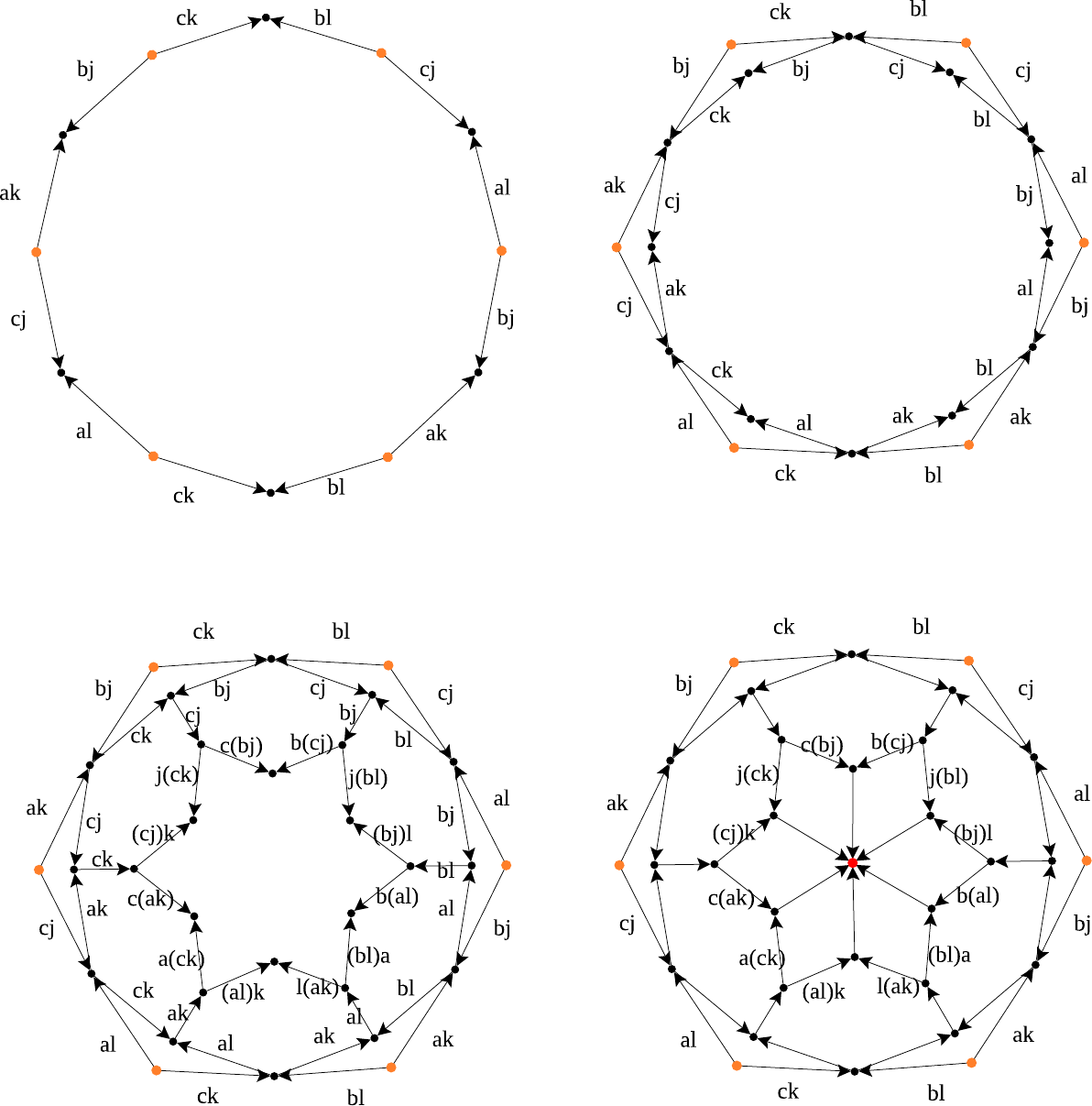}
    \caption{Construction of the universal quiver for permutohedron $\Pi_3$.}
    \label{fig:pi3_contraction}
\end{figure}

The first step is to deal with all six pairs of adjacent unlinkings:
\begin{equation}
\begin{aligned}
 \{U(ak),U(cj)\},\{U(al),& U(ck)\},\{U(bl),U(ak)\},\\ 
 \{U(bj),U(al)\},\{U(cj),& U(bl)\},\{U(ck),U(bj)\}.
\end{aligned}
\end{equation}
Since all indices in every pair are distinct, we can apply the square identity (\ref{eq:square}) to every such pair and get the top-right diagram in Fig.~\ref{fig:pi3_contraction}. A the next step we can apply the hexagon identity (\ref{eq:first hexagon}) to each of the following pairs of unlinkings
\begin{equation}
\begin{aligned}
\{U(ck),U(cj)\},\{U(bj),& U(cj)\},\{U(bl),U(bj)\}, \\
\{U(al),U(bl)\},\{U(ak),& U(al)\},\{U(ck),U(ak)\},
\end{aligned}
\end{equation}
to get the bottom-left diagram on Fig.~\ref{fig:pi3_contraction}.
Lastly, we have to apply square identity \eqref{eq:square} to connect pairs
\begin{equation}
\begin{aligned}
\{U(j(ck)),U(c(bj))\},\{U(b(cj)),& U(j(bl))\}, \{U((bj)l),U(b(al))\},  \\
\{U((bl)a),U(l(ak))\},\{U((al)k),& U(a(ck))\}, \{U(c(ak)),U((cj)k)\}.
\end{aligned}
\end{equation}
Note that the hexagon relation from the previous step leads to the following identification of nodes (see Remark~\ref{rmk:square and hexagon}):
\begin{equation}
\begin{aligned}
  (al) =((\iota j)l) = ((j\iota )l) &= (j(\iota l)) = (jc) = (cj)\,,   \\
(ak)  = ((j\iota )k) &= (j(\iota k)) = (bj)\,, \\
(ck) = ((l \iota)k) &= (l (\iota k)) = (bl) \,.
\end{aligned}
\end{equation}
This identification is crucial to close the structure -- otherwise we would have ended with an infinite ``cylinder'' of unlinkings.
As a~result, we finally close permutohedron $\Pi_3$ after five iterations, as shows the bottom-right diagram on Fig.~\ref{fig:pi3_contraction}. 
Note that all these pictures (including the ones from Sec.~\ref{sec:Creating permutohedron Pi3}) can be combined into one, shown in the introduction as Fig.~\ref{fig:Pi3_Rocket}.

\subsection{Proof of the Permutohedra Graph Theorem}

In the previous section we analysed universal quivers for permutohedra created by unlinking. Now we move to permutohedra graphs and prove Theorem~\ref{thm:Permutohedra Graph Theorem}.

\begin{proof}[Proof of the Permutohedra Graph Theorem]
From Definition~\ref{dfn:permutohedra graph} we know that the permutohedra graph is connected and made of several permutohedra. Let us take any two of them that are glued together (i.e. have nonempty intersection), denote them by $\Pi_{n_1}$ and $\Pi_{n_2}$, and denote a~quiver that is a~gluing point by  $Q_{1|2}$ (strictly speaking, we mean that $\phi(Q_{1|2})\in \Pi_{n_1} \cap \Pi_{n_2}$, where $\phi$ is a~ bijection from Proposition~\ref{prp:Permutohedron by unlinking}, and if there are more such quivers, we choose one of them). From the Permutohedron Theorem we know that for $\Pi_{n_1}$ as well as $\Pi_{n_2}$ there exists a~universal quiver -- $\hat{Q}_{1}$ and $\hat{Q}_{2}$ respectively -- and sequences of unlinking -- $\boldsymbol{U}(\boldsymbol{k_{\rho}}\boldsymbol{l_{\rho}})$ and $\boldsymbol{U}(\boldsymbol{m_{\pi}}\boldsymbol{n_{\pi}})$ -- such that
\begin{equation}
\hat{Q}_{1}=\boldsymbol{U}(\boldsymbol{k_{\rho}}\boldsymbol{l_{\rho}})Q_{\rho}\,,\qquad\hat{Q}_{2}=\boldsymbol{U}(\boldsymbol{m_{\pi}}\boldsymbol{n_{\pi}})Q_{\pi}\label{eq:closing two permutohedra}
\end{equation}
for all $\rho\in S_{n_1}, \pi\in S_{n_2}$. Since $\phi(Q_{1|2})\in \Pi_{n_1} \cap \Pi_{n_2}$, we know that there exist sequences of unlinking $\boldsymbol{U}(\boldsymbol{k_{1|2}}\boldsymbol{l_{1|2}})$ and $\boldsymbol{U}(\boldsymbol{m_{1|2}}\boldsymbol{n_{1|2}})$ such that
\begin{equation}
\label{eq:gluing}
\hat{Q}_{1}=\boldsymbol{U}(\boldsymbol{k_{1|2}}\boldsymbol{l_{1|2}})Q_{1|2}\,,\qquad\hat{Q}_{2}=\boldsymbol{U}(\boldsymbol{m_{1|2}}\boldsymbol{n_{1|2}})Q_{1|2}\,.
\end{equation}
From the Connector Theorem we know that there exist finite sequences of unlinking $\boldsymbol{U}(\boldsymbol{k'_{1|2}}\boldsymbol{l'_{1|2}})$ and $\boldsymbol{U}(\boldsymbol{m'_{1|2}}\boldsymbol{n'_{1|2}})$ such that
\begin{equation}
\boldsymbol{U}(\boldsymbol{m'_{1|2}}\boldsymbol{n'_{1|2}})\boldsymbol{U}(\boldsymbol{k_{1|2}}\boldsymbol{l_{1|2}})Q_{1|2}=\boldsymbol{U}(\boldsymbol{k'_{1|2}}\boldsymbol{l'_{1|2}})\boldsymbol{U}(\boldsymbol{m_{1|2}}\boldsymbol{n_{1|2}})Q_{1|2}\,.
\end{equation}
Combining it with (\ref{eq:closing two permutohedra}), we learn that
\begin{equation}
\boldsymbol{U}(\boldsymbol{m'_{1|2}}\boldsymbol{n'_{1|2}})\boldsymbol{U}(\boldsymbol{k_{\rho}}\boldsymbol{l_{\rho}})Q_{\rho}=\boldsymbol{U}(\boldsymbol{k'_{1|2}}\boldsymbol{l'_{1|2}})\boldsymbol{U}(\boldsymbol{m_{\pi}}\boldsymbol{n_{\pi}})Q_{\pi}=\hat{Q}_{1|2}
\end{equation}
for all $\rho\in S_{n_1}, \pi\in S_{n_2}$. We will shortly write it as
\begin{equation}
\hat{Q}_{1|2}=\boldsymbol{U}(\boldsymbol{i_{1|2,\sigma}}\boldsymbol{j_{1|2,\sigma}})Q_{\sigma}
\end{equation}
for all quivers that form permutohedra $\Pi_{n_1}$ and $\Pi_{n_2}$.

In the next step we take another quiver that is a~gluing point of $\Pi_{n_1}$ or $\Pi_{n_2}$ with some other permutohedron $\Pi_{n_3}$ and repeat the reasoning above, but this time we use universal quivers
$\hat{Q}_{1|2}$ and $\hat{Q}_{3}$ and respective sequences of unlinking. In consequence, we obtain the universal quiver $\hat{Q}_{1|2|3}$ and sequences of unlinking $\boldsymbol{U}(\boldsymbol{i_{1|2|3,\sigma}}\boldsymbol{j_{1|2|3,\sigma}})$ such that
\begin{equation}
\hat{Q}_{1|2|3}=\boldsymbol{U}(\boldsymbol{i_{1|2|3,\sigma}}\boldsymbol{j_{1|2|3,\sigma}})Q_{\sigma}
\end{equation}
for all quivers that form permutohedra $\Pi_{n_1}$, $\Pi_{n_2}$ and $\Pi_{n_3}$.

Repeating this procedure for all $p$ permutohedra composing the graph, we find that the universal quiver $\hat{Q}=\hat{Q}_{1|\dots|p}$ and sequences of unlinking $\boldsymbol{U}(\boldsymbol{i_{\sigma}}\boldsymbol{j_{\sigma}})=\boldsymbol{U}(\boldsymbol{i_{1|\dots|p,\sigma}}\boldsymbol{j_{1|\dots|p,\sigma}})$ satisfy
\begin{equation}
\hat{Q}=\boldsymbol{U}(\boldsymbol{i_{\sigma}}\boldsymbol{j_{\sigma}})Q_{\sigma}
\end{equation}
for all quivers $Q_{\sigma}$ that form the permutohedra graph, which is what we wanted to show.
\end{proof}

\subsection{Example -- universal quiver for permutohedra graph containing \texorpdfstring{\newline}{ } two \texorpdfstring{$\Pi_2$}{Pi2}}

Let us illustrate Permutohedra Graph Theorem with the example of a~permutohedra graph containing two $\Pi_2$, described in Sec.~\ref{sec:Two Pi_2}. We have quivers $\{Q_{1,(1,2)}, Q_{1,(2,1)} \}$ that form permutohedron $\Pi_{n_1}$ and quivers $\{Q_{2,(1,2)}, Q_{2,(2,1)} \}$ that form permutohedron $\Pi_{n_2}$, where $n_1=n_2=2$. Moreover, we know that gluing comes from the fact that $Q_{1,(1,2)}$ and $Q_{2,(1,2)}$ are really the same quiver which will be also denoted by $Q_{1|2}$\,.

From Sec.~\ref{sec:Contracting permutohedron Pi3} we know that for $\Pi_{n_1}$ the universal quiver is given by
\begin{equation}\label{eq:Pi_n_1}
    \hat{Q}_{1} = U((j_{12}\iota_1)j_{11})Q_{1,(2,1)} = U(j_{12}(\iota_1 j_{11}))Q_{1,(1,2)} \,.
\end{equation}
Similarly, for $\Pi_{n_2}$ the universal quiver is given by
\begin{equation}\label{eq:Pi_n_2}
    \hat{Q}_{2} = U((j_{21}\iota_2)j_{22})Q_{2,(1,2)} = U(j_{21}(\iota_2 j_{22}))Q_{2,(2,1)} \,.
\end{equation}
If we recall that $Q_{1,(1,2)}=Q_{2,(1,2)}=Q_{1|2}$ and apply Connector Theorem to $U(j_{12}(\iota_1 j_{11}))Q_{1|2}$ and $U((j_{21}\iota_2)j_{22})Q_{1|2}$, we obtain
\begin{equation}
    U((j_{21}\iota_2)j_{22})U(j_{12}(\iota_1 j_{11}))Q_{1|2} = U(j_{12}(\iota_1 j_{11}))U((j_{21}\iota_2)j_{22})Q_{1|2} \,.
\end{equation}
Combining it with Eqs. (\ref{eq:Pi_n_1}-\ref{eq:Pi_n_2}) we learn that the universal quiver for the whole permutohedra graph is given by
\begin{equation}
\begin{split}
    \hat{Q} &= U((j_{21}\iota_2)j_{22})U((j_{12}\iota_1)j_{11})Q_{1,(2,1)} = U((j_{21}\iota_2)j_{22})U(j_{12}(\iota_1 j_{11}))Q_{1,(1,2)} \\
    &= U(j_{12}(\iota_1 j_{11}))U((j_{21}\iota_2)j_{22})Q_{2,(1,2)} = U(j_{12}(\iota_1 j_{11}))U(j_{21}(\iota_2 j_{22}))Q_{2,(2,1)} \,.
\end{split}
\end{equation}
Summing up, we have obtained the following diagram

\[\begin{tikzcd}
	&& {\hat{Q}} \\
	& {\hat{Q}_1} && {\hat{Q}_2} \\
	{Q_{1,(2,1)}} && {Q_{1,(1,2)}=Q_{2,(1,2)}} && {Q_{2,(2,1)}}
	\arrow["{(j_{12}\iota_1)j_{11}}", from=3-1, to=2-2]
	\arrow["{j_{12}(\iota_1j_{11})}"', from=3-3, to=2-2]
	\arrow["{(j_{21}\iota_2)j_{22}}", from=3-3, to=2-4]
	\arrow["{(j_{21}\iota_2)j_{22}}"', from=3-5, to=2-4]
	\arrow["{j_{12}(\iota_1j_{11})}"', from=2-4, to=1-3]
	\arrow["{(j_{21}\iota_2)j_{22}}", from=2-2, to=1-3]
\end{tikzcd}\]

Note that this can be combined with Fig.~\ref{fig:two Pi2 s} into a network of unlinkings that starts from and ends on a single quiver.

\section{Universal quivers for knots}\label{sec:Universal quivers for knots}

In this section we discuss the main application of our results. In \cite{JKLNS2105} it was shown that quivers corresponding to the same knot form permutohedra graphs that arise from splitting (see App.~\ref{sec:Statements from the permutohedra paper}). Using Permutohedron Theorem and Permutohedra Graph Theorem, we can show that they can be all unlinked to the same quiver.

\subsection{Splitting by unlinking} 

In Proposition~\ref{prp:Permutohedron by unlinking} we showed that the application of a~sequence of $n$ unlinkings that share one index leads to the structure of permutohedron $\Pi_n$. On the other hand, \cite{JKLNS2105} provides an~analysis of permutohedra composed of equivalent quivers corresponding to the same knot, which can be generated by the procedure of splitting (see App.~\ref{sec:Statements from the permutohedra paper}). This brings a~natural question about the relationship between permutohedra coming from unlinkings and permutohedra coming from splitting.

\begin{prp}
\label{prp:splitting by unlinking}
Every splitting can be realised by an application of a~sequence of unlinkings to the prequiver if we allow for the presence of an extra row/column\footnote{Since we consider symmetric quivers, rows and columns are in 1-1 correspondence. When we talk about row/column of $Q$, we mean both the row/column in the adjacency matrix $C$ and the corresponding entry of the vector of generating parameters $\boldsymbol{x}$.}. 
\end{prp}
\begin{proof}
Let us consider  $(k,l)$-splitting of $n$ nodes of prequiver $\check{Q}$ with permutation $\sigma\in S_{n}$ in the presence of $s$ spectators with corresponding shifts $h_{1},\dots,h_{s}$, and multiplicative factor $\kappa$ (see Definition \ref{def:splitting}). 
Following the idea from \cite{EKL2108}, let us consider a~quiver $\check{Q'}$ which adjacency matrix is given by the one of $\check{Q}$ plus one extra row/column. We will denote its index by $\iota$ and take $x_{\iota}=q \kappa$ (we always assume that this is the last row/column; recall that we write only the upper-triangular part of the symmetric matrix):
\begin{equation}
\label{eq:checkQ' definition}
\begin{split}
    \check{C}'&=\left[\begin{array}{ccccccc}
     &  &  &  &  &  & h_{1}\\
     &  &  &  &  &  & \vdots\\
     &  &  &  &  &  & h_{s}\\
     &  &  & \check{C} &  &  & k+1\\
     &  &  &  &  &  & \vdots\\
     &  &  &  &  &  & k+1\\
    &  &  &  &  &  & l-2k-1
    \end{array}\right]    \\
    &=\left[\begin{array}{ccccccc}
    \check{C}_{11} &  &  &  & \cdots &  & h_{1}\\
     & \ddots &  &  &  &  & \vdots\\
     &  & \check{C}_{ss} &  & \cdots &  & h_{s}\\
     &  &  & \check{C}_{s+1,s+1} & \cdots &  & k+1\\
     &  &  &  & \ddots &  & \vdots\\
     &  &  &  &  & \check{C}_{s+n,s+n} & k+1\\
    &  &  &  &  &  & l-2k-1
    \end{array}\right].  
\end{split}
\end{equation}
Without loss of generality, we consider unlinkings $U(i\iota)$ and $U(j\iota)$ ($i,j\in\{s+1,\dots,s+n\}$) acting in two different orders (we allow for other unlinkings from the set $\{U(s+1,\iota),\dots,U(s+n,\iota)\}$ to act before, in between, or after them). On one hand we have
\begin{align} 
 \dots  U(j\iota) &\dots  U(i\iota)\dots\check{C'}= \nonumber \\
&=  \dots U(j\iota)\dots U(i\iota)\dots\left[\begin{array}{cccccccc}
\ddots &  &  & \vdots &  & \vdots &  & \vdots\\
 & \check{C}_{ss} & \cdots & \check{C}_{si} & \cdots & \check{C}_{sj} & \cdots & h_{s}\\
&  &  \ddots & \vdots &  & \vdots &  & \vdots\\
&  &  &  \check{C}_{ii} & \cdots & \check{C}_{ij} & \cdots & k+1\\
&  &  &  & \ddots & \vdots &  & \vdots\\
&  &  &  &  & \check{C}_{jj} & \cdots & k+1\\
&  &  &  &  &  & \ddots & \vdots\\
&  &  &  &  &  &  & l-2k-1
\end{array}\right] \label{eq:splitting from unlinking 1}\\
&=  \dots U(j\iota)\dots\left[\begin{array}{ccccccccc}
\ddots &  &  & \vdots & \vdots &  & \vdots &  & \vdots\\
 & \check{C}_{ss} & \cdots & \check{C}_{si} & \check{C}_{si}+h_{s} & \cdots & \check{C}_{sj} & \cdots & h_{s}\\
& & \ddots & \vdots &  &  & \vdots &  & \vdots\\
& & & \check{C}_{ii} & \check{C}_{ii}+k & \cdots & \check{C}_{ij} & \cdots & k\\
& & & & \check{C}_{ii}+l & \cdots & \check{C}_{ij}+k+1 & \cdots & l-k-1\\
& & & & & \ddots & \vdots &  & \vdots\\
& & & & & & \check{C}_{jj} & \cdots & k+1\\
& & & & & & & \ddots & \vdots\\
& & & & & & & & l-2k-1
\end{array}\right] \nonumber\\
&=  \left[\begin{array}{cccccccccc}
\ddots &  &  & \vdots & \vdots &  & \vdots & \vdots &  & \vdots\\
 & \check{C}_{ss} & \cdots & \check{C}_{si} & \check{C}_{si}+h_{s} & \cdots & \check{C}_{sj} & \check{C}_{sj}+h_{s} & \cdots & h_{s}\\
&  & \ddots & \vdots & \vdots &  & \vdots & \vdots &  & \vdots\\
&  &  & \check{C}_{ii} & \check{C}_{ii}+k & \cdots & \check{C}_{ij} & \check{C}_{ij}+k & \cdots & k\\
&  &  &  & \check{C}_{ii}+l & \cdots & \check{C}_{ij}+k+1 & \check{C}_{ij}+l & \cdots & l-k-1\\
&  &  &  &  & \ddots & \vdots & \vdots &  & \vdots\\
&  &  &  &  &  & \check{C}_{jj} & \check{C}_{jj}+k & \cdots & k\\
&  &  &  &  &  &  & \check{C}_{jj}+l & \cdots & l-k-1\\
 &  &  &  &  &  &  &  & \ddots & \vdots\\
 &  &  &  &  &  &  &  &  & l-2k-1
\end{array}\right]\,, \nonumber
\end{align}
on the other:
\begin{align}  
\label{eq:splitting from unlinking 2}
\dots  U(i\iota) & \dots  U(j\iota)\dots\check{C'}=\\
= & \left[\begin{array}{cccccccccc}
\ddots &  &  & \vdots & \vdots &  & \vdots & \vdots &  & \vdots\\
 & \check{C}_{ss} & \cdots & \check{C}_{si} & \check{C}_{si}+h_{s} & \cdots & \check{C}_{sj} & \check{C}_{sj}+h_{s} & \cdots & h_{s}\\
&  & \ddots & \vdots & \vdots &  & \vdots & \vdots &  & \vdots\\
&  &  & \check{C}_{ii} & \check{C}_{ii}+k & \cdots & \check{C}_{ij} & \check{C}_{ij}+k+1 & \cdots & k\\
&  &  &  & \check{C}_{ii}+l & \cdots & \check{C}_{ij}+k & \check{C}_{ij}+l & \cdots & l-k-1\\
&  &  &  &  & \ddots & \vdots & \vdots &  & \vdots\\
&  &  &  &  &  & \check{C}_{jj} & \check{C}_{jj}+k & \cdots & k\\
&  &  &  &  &  &  & \check{C}_{jj}+l & \cdots & l-k-1\\
 &  &  &  &  &  &  &  & \ddots & \vdots\\
 &  &  &  &  &  &  &  &  & l-2k-1
\end{array}\right]. \nonumber
\end{align}
In both cases the variables associated to new nodes arising from unlinkings
$U(i\iota)$ and $U(j\iota)$ are given by $x_{i}\kappa$ and $x_{j}\kappa$.

Analysing these adjacency matrices and referring to Definition \ref{def:splitting}, we can see that $U(s+\sigma(n),\iota)\dots U(s+\sigma(1),\iota)\check{Q'}$ is the matrix that arises from $(k,l)$-splitting of $n$ nodes of $\check{Q}$ with permutation $\sigma\in S_{n}$ in the presence of $s$ spectators with corresponding shifts $h_{1},\dots,h_{s}$, multiplicative factor $\kappa$, and one extra row/column 
\begin{equation}
\label{eq:extra row/column}
    (h_{1}\,,\,\dots,h_{s}\,,\,k\,,\,l-k-1\,,\,k,\,l-k-1\,,\,\dots\,,\,k\,,\,l-k-1\,,\,l-2k-1)\,,
\end{equation}
that we denoted by $\iota$.\footnote{One can compare matrices (\ref{eq:splitting from unlinking 1}) and (\ref{eq:splitting from unlinking 2}) with  Definition \ref{def:splitting} and check that the position of the extra unit follows exactly the
inverses in permutation $\sigma$.}
\end{proof}

\begin{rmk}
\label{rmk:generalisation to other permutohedra}
    Splitting requires equal number of arrows between the extra node and the nodes that are split: 
\begin{equation}
\label{eq:splitting condition}
    \check{C}_{s+1,\iota}=\check{C}_{s+2,\iota}=\dots = \check{C}_{s+n,\iota}=k+1\ ,
\end{equation}
which is not required for quivers forming permutohedron in general. 
\end{rmk}

In Sec.~\ref{sec:Universal quivers} we analysed quivers forming permutohedra and permutohedra graphs which do not necessarily satisfy condition \eqref{eq:splitting condition}. On the other hand, in this section we focus on permutohedra graphs for quivers corresponding to knots, which were found in \cite{JKLNS2105} and always come from splitting.\footnote{In case of splitting, the variable associated to the extra node is fixed to be $q^{-1}\kappa$, whereas in general it can be arbitrary.}

\subsection{Universal quivers for quivers coming from splitting}

Before applying theorems from Sec.~\ref{sec:Universal quivers}, we will show that extra rows/columns (which appeared e.g. in Proposition ~\ref{prp:splitting by unlinking}) do not alter sequences of unlinking leading to the universal quiver.

\begin{lma}
\label{lma:erasing row/column}
Assume that quivers $Q'_{1},Q'_{2},\dots Q'_{m}$ can be unlinked to the same universal quiver $\hat{Q}'$ by sequences of unlinking $\boldsymbol{U}(\boldsymbol{i_1}\boldsymbol{j_{1}}),\boldsymbol{U}(\boldsymbol{i_2}\boldsymbol{j_{2}}),\dots,\boldsymbol{U}(\boldsymbol{i_m}\boldsymbol{j_{m}})$:
\begin{equation}
\hat{Q}'=\boldsymbol{U}(\boldsymbol{i_{k}}\boldsymbol{j_{k}})Q'_{k}\qquad\forall k=1,\dots,m,
\end{equation}
and that all $Q'_{k}$ ($k=1,\dots,m$) share the same row/column that does not participate in any of the unlinkings, i.e. the corresponding node is not unlinked by any of the operators from  $\boldsymbol{U}(\boldsymbol{i_{k}}\boldsymbol{j_{k}})$ ($k=1,\dots,m$). 

If we erase this row/column in every $Q'_{k}$ ($k=1,\dots,m$) as well as (its enlarged version) in $\hat{Q}'$, and denote the resulting quivers $Q_{k}$ ($k=1,\dots,m$) and $\hat{Q}$ respectively, then they are related by the same sequences of unlinking:
\begin{equation}
\hat{Q}=\boldsymbol{U}(\boldsymbol{i_{k}}\boldsymbol{j_{k}})Q_{k}\qquad\forall k=1,\dots,m.
\end{equation}
\end{lma}

\begin{proof}
Since the node corresponding to the shared row/column is not unlinked by any of the operators from sequences  $\boldsymbol{U}(\boldsymbol{i_{k}}\boldsymbol{j_{k}})$ ($k=1,\dots,m$), they all act on it as identity. In consequence, the following diagram commutes:
\[
\begin{array}{ccccc}
 &  & \textrm{acting by }\boldsymbol{U}(\boldsymbol{i_{k}}\boldsymbol{j_{k}})\\
 & Q'_{k} & \longrightarrow & \hat{Q}'\\
\textrm{erasing the row/column} & \downarrow &  & \downarrow & \textrm{erasing the row/column}\\
 & Q_{k} & \longrightarrow & \hat{Q}\\
 &  & \textrm{acting by }\boldsymbol{U}(\boldsymbol{i_{k}}\boldsymbol{j_{k}})
\end{array}
\]
which proves the statement.    
\end{proof}

Having Lemma~\ref{lma:erasing row/column}, we are ready to prove the existence of a~universal quiver for any permutohedron coming from splitting. 

\begin{thm}\label{thm:Universal quiver for splitting}
For any quivers $\{Q_{\sigma}\}_{\sigma=1,\dots,n!}$ arising from $(k,l)$-splitting of $n$ nodes of $\check{Q}$ in the presence of $s$ spectators with corresponding shifts $h_{1},\dots,h_{s}$ and multiplicative factor~$\kappa$ there exists a~universal quiver $\hat{Q}$.
\end{thm}

\begin{proof}
From Proposition~\ref{prp:splitting by unlinking} we know that for every $Q_{\sigma}$, $\sigma\in\{1,\dots,n!\}$ there exists a~sequence of unlinkings $U(s+\sigma(n),\iota)\dots U(s+\sigma(1),\iota)$  and quivers $\check{Q'}$, $Q'_{\sigma}$ such that:
\begin{itemize}
    \item $U(s+\sigma(n),\iota)\dots U(s+\sigma(1),\iota)\check{Q'}=Q'_{\sigma}$ ,
    \item $\check{C}'$ arises from $C$ by adding an extra row/column\footnote{This is the row/column that we denoted by $\iota$ and it is crucial that
it the same for all $Q_{\sigma}$. We overload $\sigma$ a~little:
sometimes it is a~number from the set $\{1,\dots,n!\}$, and sometimes
it is a~permutation of $n$~elements. However, we hope that every
use is clear from the context and one can define a~bijection between
these sets if necessary.}: 
\begin{equation}
\label{eq:extra row/column before}
    (h_{1}\,,\,\dots,h_{s}\,,\,k+1,\,k+1,\,\dots\,,\,k+1\,,\,l-2k-1)\,,
\end{equation}
\item $C'_{\sigma}$ arises from $C_{\sigma}$ by adding an extra row/column: 
\begin{equation}
\label{eq:extra row/column after}
    (h_{1}\,,\,\dots,h_{s}\,,\,k\,,\,l-k-1\,,\,k,\,l-k-1\,,\,\dots\,,\,k\,,\,l-k-1\,,\,l-2k-1)\,.
\end{equation}
\end{itemize}

From the Permutohedron Theorem we know that there exist sequences of unlinking $\boldsymbol{U}(\boldsymbol{i_{\sigma}}\boldsymbol{j_{\sigma}})$ and the quiver $\hat{Q}'$ such that 
\begin{equation}
\boldsymbol{U}(\boldsymbol{i_{\sigma}}\boldsymbol{j_{\sigma}})Q'_{\sigma}=\hat{Q}'\qquad\forall\sigma=1,\dots,n!\,.
\end{equation}

From the Connector Algorithm we know that the extra row/column does not participate and remains the same for all quivers. Using Lemma~\ref{lma:erasing row/column}, we can erase the extra row/column from $Q'_{\sigma}$ and $\hat{Q}'$ and obtain
\begin{equation}
\boldsymbol{U}(\boldsymbol{i_{\sigma}}\boldsymbol{j_{\sigma}})Q_{\sigma}=\hat{Q}\qquad\forall\sigma=1,\dots,n!\ ,
\end{equation}
which we wanted to show.    
\end{proof}

\subsection{Example -- universal quiver for permutohedron \texorpdfstring{$\Pi_2$}{Pi2} for knot \texorpdfstring{$4_1$}{41}}

Let us show how the construction from Proposition~\ref{prp:splitting by unlinking}, Lemma~\ref{lma:erasing row/column}, and Theorem~\ref{thm:Universal quiver for splitting} works in the case of permutohedron $\Pi_2$ formed by the quivers corresponding to the knot~$4_1$. We have $n=n!=2$:
\begin{equation}
    \{Q_{\sigma}\}_{\sigma=1,\dots,n!}=\{Q_1,Q_2\}=\{(C_1,\boldsymbol{x_1}),(C_2,\boldsymbol{x_2})\}
\end{equation}
where the adjacency matrices and generating parameters are given by \cite{JKLNS2105,EKL1910}:

\begin{equation}\label{eq:figure-eight quivers}
\begin{split}
    C_1 & =\left[\begin{array}{ccccc}
    0  &- 1 & -1 &  0 &  0\\
    -1 & -2 & -2 & -1 &  \mathbf{-1}\\
    -1 & -2 & -1 &   \mathbf{0} &  0\\
    0  & -1 &   0 &  1 &  1\\
    0  &  -1 &  0 &  1 &  2
    \end{array}\right]\,, \qquad  \qquad 
     C_2=\left[\begin{array}{ccccc}
    0  &- 1 & -1 &  0 &  0\\
    -1 & -2 & -2 & -1 &  \mathbf{0}\\
    -1 & -2 & -1 &  \mathbf{-1} & 0\\
    0  & -1 &  -1 &  1 & 1\\
    0  &   0 &  0 &  1 & 2
    \end{array}\right] \,,\\
    \boldsymbol{x_1} & = \left[ 1, a^{-2} q^2 (-t)^{-2}, q^{-1} (-t)^{-1}, q (-t), a^2 q^{-2} (-t)^2 \right]
     = \boldsymbol{x_2}  \,.
\end{split}
\end{equation}

These quivers are related by a transposition $C_{25}\leftrightarrow C_{34}$ and satisfy conditions (\ref{eq:center of mass contition}-\ref{eq:theorem second case}). They arise from the~prequiver $\check{Q}=(\check{C},\boldsymbol{\check{x}})$ given by \cite{JKLNS2105}
\begin{equation}
\begin{split}
    \label{eq:figure-eight prequiver}
    \check{C} & =\left[\begin{array}{ccc}
    0  & -1 & 0 \\
    -1  & -2 & -1 \\
    0  & -1 & 1
    \end{array}\right], \\
    \boldsymbol{\check{x}} & =
   [ 1 , a^{-2} q^2 (-t)^{-2}, q (-t) ] \,.
\end{split}
\end{equation}
as a~result of $(0,1)$-splitting of nodes number 2 and 3, with $h_1=0$ and $\kappa=-a^2q^{-3}t$ (see Definition \ref{def:splitting}).

Following the procedure from the proof of the Proposition~\ref{prp:splitting by unlinking}, we consider a~quiver $\check{Q'}=(\check{C}',\boldsymbol{\check{x}'})$ which comes from adding one extra row/column given by \eqref{eq:checkQ' definition} to the prequiver:
\begin{equation}\label{eq:Enlarged prequiver 4_1}
\begin{split}
    \check{C}'&=\left[\begin{array}{ccc:c}
     &  &  & h_{1}\\
     & \check{C} &  & k+1\\
     &  &  & k+1\\
     \hdashline
    h_{1} & k+1 & k+1 & l-2k-1
    \end{array}\right]=\left[\begin{array}{ccc:c}
    0 & -1 & 0 & 0\\
    -1 & -2 & -1 & 1\\
    0 & -1 & 1 & 1\\
    \hdashline
    0 & 1 & 1 & 0
    \end{array}\right]\ ,\\
    \boldsymbol{\check{x}'} & = [\boldsymbol{\check{x}}, q\kappa] = 
   [ 1 , a^{-2} q^2 (-t)^{-2}, q (-t) \dashline a^2q^{-2}(-t)] \,.
\end{split}
\end{equation}
We apply unlinkings $U(i\iota)=U(24)$ and $U(j\iota)=U(34)$ in two different orders and obtain
\begin{equation}
U(34)U(24)\check{C}'=
\left[\begin{array}{ccccc:c}
    0  &- 1 & -1 &  0 &  0 &  0\\
    -1 & -2 & -2 & -1 &  -1 &  0\\
    -1 & -2 & -1 &   0 &  0 &  0\\
    0  & -1 &   0 &  1 &  1 &  0\\
    0  &  -1 &  0 &  1 &  2 &  0\\
    \hdashline
    0  & 0 & 0 &  0 &  0 &  0
    \end{array}\right]
    =C'_{1}
\end{equation}
and
\begin{equation}
U(24)U(34)\check{C}'=
\left[\begin{array}{ccccc:c}
    0  &- 1 & -1 &  0 &  0 &  0\\
    -1 & -2 & -2 & -1 &  0 &  0\\
    -1 & -2 & -1 &  -1 &  0 &  0\\
    0  & -1 &  -1 &  1 &  1 &  0\\
    0  &  0 &  0 &  1 &  2 &  0\\
    \hdashline
    0  & 0 & 0 &  0 &  0 &  0
    \end{array}\right]
    =C'_{2}\,,
\end{equation}
where $C'_{1}$ and $C'_{2}$ arise from $C_{1}$ and $C_{2}$ by adding an extra row/column given by (\ref{eq:extra row/column}).  In both cases the generating parameters are given by\footnote{Note that we adjusted the ordering to be able to compare two cases easily. 
}
\begin{equation}
    \boldsymbol{x'}  = \left[ 1, a^{-2} q^2 (-t)^{-2}, q^{-1} (-t)^{-1}, q (-t), a^2 q^{-2} (-t)^2 \dashline a^2q^{-2}(-t) \right] \,.
\end{equation}
From the Connector Theorem we know that there exist finite sequences of unlinking $\boldsymbol{U}(\boldsymbol{k_{1}}\boldsymbol{l_{1}}),\boldsymbol{U}(\boldsymbol{k_{2}}\boldsymbol{l_{2}})$ such that 
\begin{equation}
\begin{split}\boldsymbol{U}(\boldsymbol{k_{1}}\boldsymbol{l_{1}})U(24)U(34)\check{Q'} & = \boldsymbol{U}(\boldsymbol{k_{2}}\boldsymbol{l_{2}})U(34)U(24)\check{Q'}\,,\\
\boldsymbol{U}(\boldsymbol{k_{1}}\boldsymbol{l_{1}})Q'_{1} & = \boldsymbol{U}(\boldsymbol{k_{2}}\boldsymbol{l_{2}})Q'_{2}\,.
\end{split}
\end{equation}
In this case they are given by\footnote{We have to be careful about the numbering of nodes. Since we put new nodes in between the old, the numbering changes and node $(24)$ becomes $3$, $3$ becomes $4$, and $(34)$ becomes $5$.} $\boldsymbol{U}(\boldsymbol{k_{1}}\boldsymbol{l_{1}})=U(2(34))$ and $\boldsymbol{U}(\boldsymbol{k_{2}}\boldsymbol{l_{2}})=U(3(24))$ and they lead to
\begin{equation}
\label{eq:41 universal}
\begin{split}
    \hat{Q}' & = U(2(34))U(24)U(34)\check{Q'}=  U(3(24))U(34)U(24)\check{Q'}\\
    & = U(2(34))Q'_{1} = U(3(24))Q'_{2} \ ,\\
    \hat{C}' & = \left[\begin{array}{cccccc:c}
        0  &- 1 & -1 &  0 &  0 & -1 & 0\\
        -1 & -2 & -2 & -1 &  -1 & -3 & 0\\
        -1 & -2 & -1 &  -1 &  0 & -2 & 0\\
        0  & -1 &  -1 &  1 &  1 & 0 & 0\\
        0  &  -1 &  0 &  1 &  2 & 1 & 0\\
        -1 & -3 & -2 & 0 & 1 & -1 & 0 \\
        \hdashline
        0  & 0 & 0 &  0 &  0 &  0 & 0
        \end{array}\right]\ , \\
    \boldsymbol{\hat{x}'} & = \left[ 1, a^{-2} q^2 (-t)^{-2}, q^{-1} (-t)^{-1}, q (-t), a^2 q^{-2} (-t)^2, 1 \dashline a^2q^{-3}(-t) \right] \,.  
\end{split}
\end{equation}

Using Lemma~\ref{lma:erasing row/column} we can erase the extra row/column and obtain the universal quiver $\hat{Q}=(\hat{C},\boldsymbol{\hat{x}})$:\footnote{Here we have the same subtlety with indices.}
\begin{equation}
\begin{split}
    \hat{C} & = \left[\begin{array}{cccccc}
        0  &- 1 & -1 &  0 &  0 & -1 \\
        -1 & -2 & -2 & -1 &  -1 & -3 \\
        -1 & -2 & -1 &  -1 &  0 & -2 \\
        0  & -1 &  -1 &  1 &  1 & 0 \\
        0  &  -1 &  0 &  1 &  2 & 1 \\
        -1 & -3 & -2 & 0 & 1 & -1 
        \end{array}\right] \ , \\
    \boldsymbol{\hat{x}} & = \left[ 1, a^{-2} q^2 (-t)^{-2}, q^{-1} (-t)^{-1}, q (-t), a^2 q^{-2} (-t)^2, 1 \right] \,.
\end{split}
\end{equation}
We can cross-check that it is related to the initial quivers forming permutohedron $\Pi_2$  by unlinkings  $\boldsymbol{U}(\boldsymbol{k_{1}}\boldsymbol{l_{1}})$ and $\boldsymbol{U}(\boldsymbol{k_{2}}\boldsymbol{l_{2}})$:\footnote{Comparison with \eqref{eq:41 universal} requires adaptation of the numbering of nodes described in the previous footnote.}
\begin{equation}
\begin{split}
    \hat{Q}&=\boldsymbol{U}(\boldsymbol{k_{1}}\boldsymbol{l_{1}})Q_{1} = U(25)Q_{1}\\
    &= \boldsymbol{U}(\boldsymbol{k_{2}}\boldsymbol{l_{2}})Q_{2} = U(34)Q_{2} \,.
\end{split}
\end{equation}

Note that permutohedron $\Pi_2$ discussed in this section can be viewed as a~special case of permutohedron  $\Pi_2$  from Sec.~\ref{sec:Example - Pi2}. Quiver $Q$, which is arbitrary there, here is given by  $\check{Q'}=(\check{C}',\boldsymbol{\check{x}'})$ from \eqref{eq:Enlarged prequiver 4_1}.  Three vertices that we chose are $\iota =4$, $j_1=2$, $j_2=3$, so $Q_{(1,2)}=U(\iota j_2)U(\iota j_1)Q$ is given by $Q'_{2}=U(34)U(24)\check{Q'}$ here, and analogously $Q_{(2,1)}=U(\iota j_1)U(\iota j_2)Q$ is given by $Q'_{1}=U(24)U(34)\check{Q'}$. 

Consequently, the universal quiver with extra row/column (and corresponding unlinkings) discussed here can be viewed as a~special case of the universal quiver for permutohedron $\Pi_2$ from Sec.~\ref{sec:Contracting permutohedron Pi3}. More precisely, $\hat{Q} =U(j_{2}(\iota j_{1}))Q_{(1,2)} = U((j_{2}\iota )j_{1})Q_{(2,1)}$ is given by 
$\hat{Q}' = U(3(24))Q'_{2} = U(2(34))Q'_{1}$ here.

The relation between quivers from this section and those from  Sections~\ref{sec:Example - Pi2} and ~\ref{sec:Contracting permutohedron Pi3} is a~reflection of Remark~\ref{rmk:generalisation to other permutohedra}.
Namely, we could have chosen different entries in the extra row/column (specifically, we could have $\check{C}'_{24}\neq \check{C}'_{34}$) and still find the universal quiver, as guaranteed by the Permutohedron Theorem. However, in that case the quivers would no longer come from splitting and correspond to the $4_1$ knot.

\subsection{Universal quivers for permutohedra graphs associated to knots}

In this section we will combine the results of \cite{JKLNS2105}, Theorem~\ref{thm:Universal quiver for splitting}, and the Permutohedra Graph Theorem to show the following statement:

\begin{thm}\label{thm:Universal quiver for knot}
Consider quiver $Q_K$ which corresponds to the knot $K$.  If  quivers  $\{Q_{\sigma}\}_{\sigma=1,\dots,m}$ are related to $Q_K$ by a~sequence of disjoint transpositions (\ref{eq:transposition}) and satisfy conditions (\ref{eq:center of mass contition}-\ref{eq:theorem second case}), then there exists a~universal quiver $\hat{Q}$ for $\{Q_K, Q_{1}, Q_{2},\dots,Q_{m}\}$.
\end{thm}

\begin{proof}
From \cite{JKLNS2105} we know that $\{Q_{\sigma}\}_{\sigma=1,\dots,m}$ correspond to the same knot $K$. Without loss of generality, we assume that the set $\{Q_K, Q_{1}, Q_{2},\dots,Q_{m}\}$ contains all quivers which correspond to the same knot and satisfy conditions (\ref{eq:center of mass contition}-\ref{eq:theorem second case}). If this is not the case, we can apply the procedure described in \cite{JKLNS2105}, find all of them, and add to the set of quivers. When we show the existence of the universal quiver for this enlarged set, it will automatically prove the statement for the initial set.

From \cite{JKLNS2105} we know that $\{Q_K, Q_{1}, Q_{2},\dots,Q_{m}\}$ can be bijectively mapped to a~graph composed of several permutohedra glued together and each permutohedron is an image of quivers arising from $(k,l)$-splitting of $n$ nodes of $\check{Q}$ in the presence of $s$ spectators with corresponding shifts $h_{1},\dots,h_{s}$ and multiplicative factor $\kappa$ .  Theorem~\ref{thm:Universal quiver for splitting} assures that for these quivers (that come from one kind of splitting and are mapped to one permutohedron) there exists a~universal quiver.

The last remaining step is the argument from the proof of the Permutohedra Graph Theorem. Successively taking the quivers which are gluing points of different permutohedra and applying the Connector Theorem we eventually find the universal quiver $\hat{Q}$ and sequences of unlinking $\boldsymbol{U}(\boldsymbol{i_{\sigma}}\boldsymbol{j_{\sigma}})$ such that 
\begin{equation}
\hat{Q}=\boldsymbol{U}(\boldsymbol{i_{\sigma}}\boldsymbol{j_{\sigma}})Q_{\sigma}\qquad\forall\sigma=K,1,\dots,m\ ,
\end{equation}
which we wanted to show.    
\end{proof}

\begin{rmk}\label{rmk:non-uniqueness of universal quiver}
We remark that in general the universal quiver is not unique. For example, in the context of the knots-quivers correspondence, it is unique for $4_1$, $5_1$, but not for $7_1$.
There are two main sources of this non-uniqueness: the first one is the choice of sequences of unlinking from the gluing point to the universal quivers for each permutohedron (e.g. the choice of $\boldsymbol{U}(\boldsymbol{k_{1|2}}\boldsymbol{l_{1|2}})$ and $\boldsymbol{U}(\boldsymbol{m_{1|2}}\boldsymbol{n_{1|2}})$ in \eqref{eq:gluing}).
Different choices of such connecting sequences may lead to different resulting quivers. On the other hand, after a~pair of sequences is chosen, we need to connect them via the Connector Algorithm from Sec.~\ref{sec:connector}. While doing so, every time we encounter a~repeated index in $\{U(ij),U(jk)\}$, we must apply either hexagon (\ref{eq:first hexagon}) or hexagon (\ref{eq:second hexagon}), and either choice is valid at each step of calculation. This leads to a~multitude of universal quivers already for knot $7_1$.\footnote{Note the complexity of the permutohedra graph for $7_1$ in comparison with relatively simple graphs for $4_1$ and $5_1$, where computation of the universal quiver does not require application of the hexagon identity.} It is an~ interesting problem to determine the structure of all universal quivers for a given set of quivers (in particular, determine which of them are minimal), and we leave it for future investigation.
\end{rmk}

\subsection{Example -- universal quiver for knot \texorpdfstring{$5_1$}{51}}

From \cite{JKLNS2105} we know that the following three quivers form
a permutohedra graph made of two $\Pi_{2}$ glued together:
\begin{align}\label{eq:5_1 three quivers}
    C_{1} &=\left[\begin{array}{ccccc}
    0 & 1 & 1 & 3 & 2\\
    1 & 2 & 3 & 3 & 3\\
    1 & 3 & 3 & 4 & 4\\
    3 & 3 & 4 & 4 & 4\\
    2 & 3 & 4 & 4 & 5
    \end{array}\right],\qquad C_{1|2}=\left[\begin{array}{ccccc}
    0 & 1 & 1 & 3 & 3\\
    1 & 2 & 2 & 3 & 3\\
    1 & 2 & 3 & 4 & 4\\
    3 & 3 & 4 & 4 & 4\\
    3 & 3 & 4 & 4 & 5
    \end{array}\right],\qquad C_{2}=\left[\begin{array}{ccccc}
    0 & 1 & 1 & 3 & 3\\
    1 & 2 & 2 & 3 & 4\\
    1 & 2 & 3 & 3 & 4\\
    3 & 3 & 3 & 4 & 4\\
    3 & 4 & 4 & 4 & 5
    \end{array}\right]\,,\nonumber\\
        \boldsymbol{x_1} & = \boldsymbol{x_{1|2}} = \boldsymbol{x_2} = \left[ x a^4 q^{−4},
    x a^4 q^{−2} (−t)^2,
    x a^6 q^{−5} (−t)^3,
    x a^4 (−t)^{4},
    x a^6 q^{−3} (−t)^5 \right] \,.
\end{align}
Our goal is to give a~detailed analysis of this permutohedra graph, using all the tools we have obtained so far.
We know that quivers  $Q_1=(C_1,\boldsymbol{x_1})$ and $Q_{1|2}=(C_{1|2},\boldsymbol{x_{1|2}})$  arise from  $(1,3)$-splitting of nodes number 1 and 2 in the prequiver
\begin{equation}
\check{C}_{1}=\left[\begin{array}{ccc}
0 & 1 & 3 \\
1 & 2 & 3 \\
3 & 3 & 4 \\
\end{array}\right] \,, \qquad
\boldsymbol{\check{x}_1} = \left[ x a^4 q^{−4},
    x a^4 q^{−2} (−t)^2,
    x a^4 (−t)^{4}\right] \,,
\end{equation}
where node number 3 is a~spectator with $h_3 = 1$ and the multiplicative factor is given by $\kappa = −a^2 q^{−1} t^3$ \cite{JKLNS2105}.  Similarly,  quivers  $Q_{1|2}=(C_{1|2},\boldsymbol{x_{1|2}})$   and   $Q_2=(C_2,\boldsymbol{x_2})$ arise from  $(0,1)$-splitting of nodes number 2 and 3 in the prequiver
\begin{equation}
\check{C}_{2}=\left[\begin{array}{ccc}
0 & 1 & 3 \\
1 & 2 & 3 \\
3 & 3 & 4 \\
\end{array}\right] \,, \qquad
\boldsymbol{\check{x}_2} = \left[ x a^4 q^{−4},
    x a^4 q^{−2} (−t)^2,
    x a^4 (−t)^{4}\right] \,,
\end{equation}
where node number 1 is a~spectator with $h_1 = 0$ and the multiplicative factor is given by $\kappa = −a^2 q^{−3} t$ \cite{JKLNS2105}.

Following the proof of Proposition~\ref{prp:splitting by unlinking}, we can realise above splittings by unlinking of the following enlarged prequivers:
\begin{align}
    \check{C}'_{1} &=\left[\begin{array}{ccc:c}
    0 & 1 & 3 & 2\\
    1 & 2 & 3 & 2\\
    3 & 3 & 4 & 1\\
    \hdashline
    2 & 2 & 1 & 0
    \end{array}\right], & 
    \check{C}'_{2} &=\left[\begin{array}{ccc:c}
    0 & 1 & 3 & 0\\
    1 & 2 & 3 & 1\\
    3 & 3 & 4 & 1\\
    \hdashline
    0 & 1 & 1 & 0
    \end{array}\right] \,, \\
    \boldsymbol{\check{x}'_1} &= \left[ x a^4 q^{−4},
    x a^4 q^{−2} t^2,
    x a^4 t^{4} \dashline −a^2 t^3\right] \,, & \boldsymbol{\check{x}'_2} &= \left[ x a^4 q^{−4},
    x a^4 q^{−2} t^2,
    x a^4 t^{4} \dashline −a^2 q^{−2} t\right]\,. \nonumber
\end{align}

In order to avoid confusion, we will denote the extra row/column of $\check{C}'_{1}$ by $\iota_1$ and the extra row/column of $\check{C}'_{2}$ by $\iota_2$ (in both cases it refers to the node number 4). Using these variables, we can write

\begin{align}\label{eq:5_1 quivers with prime}
    U(\iota_1 2)U(1 \iota_1)\check{C}'_{1} = &\ 
    \left[
\begin{array}{cccccc}
 0 & 1 & 3 & 1 & 1 & 2 \\
 1 & 2 & 3 & 1 & 3 & 3 \\
 3 & 3 & 4 & 1 & 4 & 4 \\
 1 & 1 & 1 & 0 & 1 & 1 \\
 1 & 3 & 4 & 1 & 3 & 4 \\
 2 & 3 & 4 & 1 & 4 & 5 \\
\end{array}
\right],
\quad
U(1 \iota_1)U(\iota_1 2)\check{C}'_{1} =
\left[
\begin{array}{cccccc}
 0 & 1 & 3 & 1 & 3 & 1 \\
 1 & 2 & 3 & 1 & 3 & 2 \\
 3 & 3 & 4 & 1 & 4 & 4 \\
 1 & 1 & 1 & 0 & 1 & 1 \\
 3 & 3 & 4 & 1 & 5 & 4 \\
 1 & 2 & 4 & 1 & 4 & 3 \\
\end{array}
\right]
\\ \label{eq:5_1 quivers with prime 2}
U(\iota_2 3)U(2 \iota_2)\check{C}'_{2} = &\
\left[
\begin{array}{cccccc}
 0 & 1 & 3 & 0 & 1 & 3 \\
 1 & 2 & 3 & 0 & 2 & 3 \\
 3 & 3 & 4 & 0 & 4 & 4 \\
 0 & 0 & 0 & 0 & 0 & 0 \\
 1 & 2 & 4 & 0 & 3 & 4 \\
 3 & 3 & 4 & 0 & 4 & 5 \\
\end{array}\right],
\quad
U(2 \iota_2)U(\iota_2 3)\check{C}'_{2} =
\left[
\begin{array}{cccccc}
 0 & 1 & 3 & 0 & 3 & 1 \\
 1 & 2 & 3 & 0 & 4 & 2 \\
 3 & 3 & 4 & 0 & 4 & 3 \\
 0 & 0 & 0 & 0 & 0 & 0 \\
 3 & 4 & 4 & 0 & 5 & 4 \\
 1 & 2 & 3 & 0 & 4 & 3 \\
\end{array}
\right].
\end{align}
In order to obtain original quivers for knot $5_1$, given in (\ref{eq:5_1 three quivers}), from (\ref{eq:5_1 quivers with prime}) and (\ref{eq:5_1 quivers with prime 2}), we need to remove the fourth row and column in each augmented quiver. However, with this extra row and column it is easy to see why each permutohedron $\Pi_2$ has a~universal quiver. Let us switch to the diagrammatic view of (\ref{eq:5_1 quivers with prime}) and (\ref{eq:5_1 quivers with prime 2}):
\[\begin{tikzcd}
	& \bullet && {\bullet\ \bullet} && \bullet \\
	\bullet && \bullet && \bullet && \bullet \\
	& {\check{C}_1'} &&&& {\check{C}'_2}
	\arrow["{2\iota_2}"', from=2-7, to=1-6]
	\arrow["{\iota_12}", from=3-2, to=2-3]
	\arrow["{1\iota_1}"', from=3-2, to=2-1]
	\arrow["{2\iota_2}"', from=3-6, to=2-5]
	\arrow["{\iota_23}", from=3-6, to=2-7]
	\arrow["{1\iota_1}", from=2-3, to=1-4]
	\arrow["{\iota_23}"', from=2-5, to=1-4]
	\arrow["{\iota_12}", from=2-1, to=1-2]
\end{tikzcd}\]
Evidently, this picture can be closed-up by applying hexagon relation (\ref{eq:first hexagon}) for the left and right
permutohedron independently:

\begin{equation}\label{eq:5_1 diagram with primes}
\begin{tikzcd}
	&& {\hat{C}_1'} && {\hat{C}_2'} \\
	& \bullet && {\bullet\ \bullet} && \bullet \\
	\bullet && \bullet && \bullet && \bullet \\
	& {\check{C}_1'} &&&& {\check{C}'_2}
	\arrow["{2\iota_2}"', from=3-7, to=2-6]
	\arrow["{\iota_12}", from=4-2, to=3-3]
	\arrow["{1\iota_1}"', from=4-2, to=3-1]
	\arrow["{2\iota_2}"', from=4-6, to=3-5]
	\arrow["{\iota_23}", from=4-6, to=3-7]
	\arrow["{1\iota_1}", from=3-3, to=2-4]
	\arrow["{\iota_23}"', from=3-5, to=2-4]
	\arrow["{\iota_12}", from=3-1, to=2-2]
	\arrow["{1(\iota_12)}", from=2-4, to=1-3]
	\arrow["{(1\iota_1)2}", from=2-2, to=1-3]
	\arrow["{(2\iota_2)3}"', from=2-4, to=1-5]
	\arrow["{2(\iota_23)}"', from=2-6, to=1-5]
\end{tikzcd}\end{equation}

\noindent where 

\begin{equation}
    \hat{C}_1' = \left[\begin{array}{ccccccc}
 0 & 1 & 3 & 1 & 1 & 2 & 2 \\
 1 & 2 & 3 & 1 & 2 & 3 & 4 \\
 3 & 3 & 4 & 1 & 4 & 4 & 7 \\
 1 & 1 & 1 & 0 & 1 & 1 & 2 \\
 1 & 2 & 4 & 1 & 3 & 4 & 5 \\
 2 & 3 & 4 & 1 & 4 & 5 & 7 \\
 2 & 4 & 7 & 2 & 5 & 7 & 10 \\
\end{array}\right]\,,
\qquad
\hat{C}_2' = \left[\begin{array}{ccccccc}
 0 & 1 & 3 & 0 & 3 & 1 & 4 \\
 1 & 2 & 3 & 0 & 3 & 2 & 5 \\
 3 & 3 & 4 & 0 & 4 & 3 & 7 \\
 0 & 0 & 0 & 0 & 0 & 0 & 0 \\
 3 & 3 & 4 & 0 & 5 & 4 & 8 \\
 1 & 2 & 3 & 0 & 4 & 3 & 6 \\
 4 & 5 & 7 & 0 & 8 & 6 & 14 \\
\end{array}\right]\,.
\end{equation}

In other words, it amounts to application of Permutohedron Theorem~\ref{thm:Permutohedron Theorem} two times, one for each permutohedron $\Pi_2$.
The two quivers on top ($\hat{C}_1'$ and $\hat{C}_2'$) are universal
for each permutohedron $\Pi_2$, and are obtained from unlinking of $\check{C}_1'$ and $\check{C}_2'$, respectively.
(The ordering of nodes in $\hat{C}_1'$ and $\hat{C}_2'$ given above, corresponds to the leftmost and rightmost paths in diagram
(\ref{eq:5_1 diagram with primes}).)
However, due to the presence of the extra row and column we cannot immediately apply the universal quiver theorem for the full permutohedra graph. Instead, we first truncate all quivers by removing $\iota_1$ and $\iota_2$, or equivalently the fourth row and column in (\ref{eq:5_1 quivers with prime}) and (\ref{eq:5_1 quivers with prime 2}), and apply Lemma 
\ref{lma:erasing row/column}. As a~result, the two middle quivers merge into one, and the remaining part of the diagram is

\[
\begin{tikzcd}
	& {\hat{C}_1} && {\hat{C}_2} \\
	{C_1} && {C_{1|2}} && {C_2}
	\arrow["{(1\iota_1)2}", from=2-1, to=1-2]
	\arrow["{1(\iota_12)}", from=2-3, to=1-2]
	\arrow["{(2\iota_2)3}"', from=2-3, to=1-4]
	\arrow["{2(\iota_23)}"', from=2-5, to=1-4]
\end{tikzcd}
\]
\begin{equation}
\hat{C}_{1}=\left[\begin{array}{cccccc}
0 & 1 & 1 & 3 & 2 & 2\\
1 & 2 & 2 & 3 & 3 & 4\\
1 & 2 & 3 & 3 & 4 & 5\\
3 & 3 & 3 & 4 & 4 & 7\\
2 & 3 & 4 & 4 & 5 & 7\\
2 & 4 & 5 & 7 & 7 & 10
\end{array}\right],
\quad
\hat{C}_{2}=\left[\begin{array}{cccccc}
0 & 1 & 1 & 3 & 2 & 4\\
1 & 2 & 2 & 3 & 3 & 5\\
1 & 2 & 3 & 3 & 4 & 6\\
3 & 3 & 3 & 4 & 4 & 7\\
2 & 3 & 4 & 4 & 5 & 8\\
4 & 5 & 6 & 7 & 8 & 14
\end{array}\right].
\end{equation}
(In order to have a~visual comparison with (\ref{eq:5_1 three quivers}), we re-ordered the nodes according to the number of loops.) 

Importantly, we have ``lost'' the bottom part -- this is due to removal the quiver node which has been subject to unlinking in the previous graph. However, the upper part with two universal quivers survives, because no unlinking on the upper part take as an argument either $\iota_1$ or $\iota_2$. Therefore, having now this diagram and truncated quivers, we can easily connect the two quivers $\hat{C}_1$ and $\hat{C}_2$ (by applying Theorem~\ref{thm:Universal quiver for knot}) by choosing two paths from the gluing quiver $C_{1|2}$ and connecting them.
This choice is unique in this case and, after applying square identity (\ref{eq:square}), we get:

\begin{equation}\label{eq:diagram of 5_1 permutohedron and universal quiver}
\begin{tikzcd}
	&& {\hat{C}} \\
	& {\hat{C}_1} && {\hat{C}_2} \\
	{C_1} && {C_{1|2}} && {C_2}
	\arrow["{(1\iota_1)2}", from=3-1, to=2-2]
	\arrow["{1(\iota_12)}", from=3-3, to=2-2]
	\arrow["{(2\iota_2)3}"', from=3-3, to=2-4]
	\arrow["{2(\iota_23)}"', from=3-5, to=2-4]
	\arrow["{1(\iota_12)}"', from=2-4, to=1-3]
	\arrow["{(2\iota_2)3}", from=2-2, to=1-3]
\end{tikzcd}
\end{equation}
We thus learn that the universal quiver for (\ref{eq:5_1 three quivers}) is given by
\begin{equation}
\hat{C}=  \left[\begin{array}{ccccccc}
0 & 1 & 1 & 3 & 2 & 2 & 4\\
1 & 2 & 2 & 3 & 3 & 4 & 5\\
1 & 2 & 3 & 3 & 4 & 5 & 6\\
3 & 3 & 3 & 4 & 4 & 7 & 7\\
2 & 3 & 4 & 4 & 5 & 7 & 8\\
2 & 4 & 5 & 7 & 7 & 10 & 12\\
4 & 5 & 6 & 7 & 8 & 12 & 14
\end{array}\right]\,.
\end{equation}

We may ask if this permutohedra graph comes from a~construction similar to that of Proposition~\ref{prp:Permutohedra graph}, where we start from a~single quiver and apply unlinking to generate the whole permutohedra graph. Let us return to the diagram (\ref{eq:5_1 diagram with primes}) with two distinct quivers $\check{C}_1$ and $\check{C}_2$ which generate each permutohedron $\Pi_2$. It is interesting to see how this picture can be completed
so that the whole graph is obtained from a~single quiver at the bottom.

\begin{figure}[h!]
\centering
\begin{tikzcd}
	&&& \hat{C}'' \\
	&& \bullet && \bullet \\
	& C_1'' && C_{1|2}'' && C_2'' \\           
	\bullet && \bullet && \bullet && \bullet \\
	& {\check{C}_1''} && \bullet && {\check{C}_2''} \\
	&& \bullet && \bullet \\
	&& \bullet && \bullet \\
	&&& C''
	\arrow["{(1\iota_1)2}", from=3-2, to=2-3]
	\arrow["", from=3-4, to=2-3]
	\arrow["", from=3-4, to=2-5]
	\arrow["{2(\iota_2 3)}"', from=3-6, to=2-5]
	\arrow["{1(\iota_12)}"', from=2-5, to=1-4]
	\arrow["{(2\iota_2)3}", from=2-3, to=1-4]
	\arrow["{1\iota_1}"', from=5-2, to=4-1]
	\arrow["{\iota_12}", from=4-1, to=3-2]
	\arrow["{\iota_12}", from=5-2, to=4-3]
	\arrow["{1\iota_1}"', from=4-3, to=3-4]
	\arrow["{2\iota_2}"', from=5-6, to=4-5]
	\arrow["{\iota_23}", from=4-5, to=3-4]
	\arrow["{\iota_23}", from=5-6, to=4-7]
	\arrow["{2\iota_2}"', from=4-7, to=3-6]
	\arrow["{1\iota_1}", from=5-4, to=4-5]
	\arrow["{\iota_23}"', from=5-4, to=4-3]
	\arrow["{\iota_12}", from=6-3, to=5-4]
	\arrow["{\iota_23}"', from=6-3, to=5-2]
	\arrow["{1\iota_1}", from=6-5, to=5-6]
	\arrow["{2\iota_2}"', from=6-5, to=5-4]
	\arrow["{2\iota_2}", from=8-4, to=7-3]
	\arrow["{\iota_1(2\iota_2)}", from=7-3, to=6-3]
	\arrow["{\iota_12}"', from=8-4, to=7-5]
	\arrow["{(\iota_12)\iota_2}"', from=7-5, to=6-5]
\end{tikzcd}
\caption{The complete unlinking structure for permutohedra graph for knot $5_1$.}
\label{fig:5_1 permutohedra full diagram}
\end{figure}

It turns out to be possible and the solution is presented 
in Fig.~\ref{fig:5_1 permutohedra full diagram}, where we start from a single quiver $Q''=(C'',\boldsymbol{x''})$  and apply unlinking to generate permutohedra graph consisting two $\Pi_2$ and the universal quiver. The structure of unlinkings presented in this diagram can be applied to any quiver $Q''$ containing nodes $\{1,2,3,\iota_1,\iota_2\}$ but for 
\begin{equation}
\begin{split}
    C'' &= \left[
    \begin{array}{ccccc}
     0 & 1 & 3 & 0 & 2 \\
     1 & 2 & 3 & 1 & 2 \\
     3 & 3 & 4 & 1 & 1 \\
     0 & 1 & 1 & 0 & \alpha \\
     2 & 2 & 1 & \alpha & 0 \\
    \end{array}
    \right] \,,   \\
    \boldsymbol{x''} &= [x_1,x_2,x_3,x_{\iota_1},x_{\iota_2}] \,,
\end{split}
\end{equation}
we can reproduce quivers corresponding to the knot $5_1$ for arbitrary $\alpha$. More precisely, after removing nodes $\iota_1,\iota_2,(2\iota_2),(\iota_1(2\iota_2)),(\iota_23)$ from $\check{C}_1''$, and $\iota_1,\iota_2,(\iota_12),((\iota_12)\iota_2),(1\iota_1)$ from $\check{C}_2''$, quivers from Fig.~\ref{fig:5_1 permutohedra full diagram} reproduce quivers from (\ref{eq:diagram of 5_1 permutohedron and universal quiver}).\footnote{As a~cross-check, note that $C''$ must have size at least 5. Then, unlinking it three times yields $\check{C}_1''$ and $\check{C}_2''$, which now have size 8. Therefore, the above mentioned truncation corresponds to the map $\check{C}_i''\mapsto \check{C}_i$ and gives the initial permutohedra graph (\ref{eq:5_1 three quivers}).} We emphasise that truncation procedure does not preserve the motivic generating series, but it does preserve the structure of unlinkings, as prescribed by Lemma~\ref{lma:erasing row/column}.

On one hand there are many similarities between Figs.~\ref{fig:two Pi2 s} and~\ref{fig:5_1 permutohedra full diagram}.
However, the shape of the diagram is slightly different because of the repeated index in the initial unlinkings $U(\iota_1 2)$ and $U(2\iota_2)$ applied to $C''$. Therefore, this case does not precisely fall into the construction of Proposition~\ref{prp:Permutohedra graph}, but nevertheless it is quite close to that. We thus believe that a~larger class of permutohedra graphs (in particular those coming from knots) arise from a~single quiver and multiple unlinking operations applied to this quiver.
We leave the investigation of such more general graphs, as well as interpretation of this ``minimal'' quiver, for a~future work.

\subsection{Example -- universal quiver for knot \texorpdfstring{$7_1$}{71}}

We finish this section with the analysis of the permutohedra graph for knot $7_1$. Although the computations are lengthier, the method still follows the proof of Theorem~\ref{thm:Universal quiver for knot}:
\begin{enumerate}
    \item We take the permutohedra graph for knot $7_1$ which was found in \cite{JKLNS2105}.
    \item Each permutohedron in this graph comes from splitting a prequiver, so first we use Theorem~\ref{thm:Universal quiver for splitting} to find a universal quiver for every permutohedron.
    \item We take the gluing points of different permutohedra  and sequences of unlinking leading to respective universal quivers and apply the Connector Algorithm.
    \item We repeat this procedure for every pair of permutohedra until we find a universal quiver of the knot $7_1$.
\end{enumerate}

The permutohedra graph is depicted on Fig.~\ref{fig:permutohedron 7_1} and consists of 13 quivers.
\begin{figure}[h!]
\centering
\input{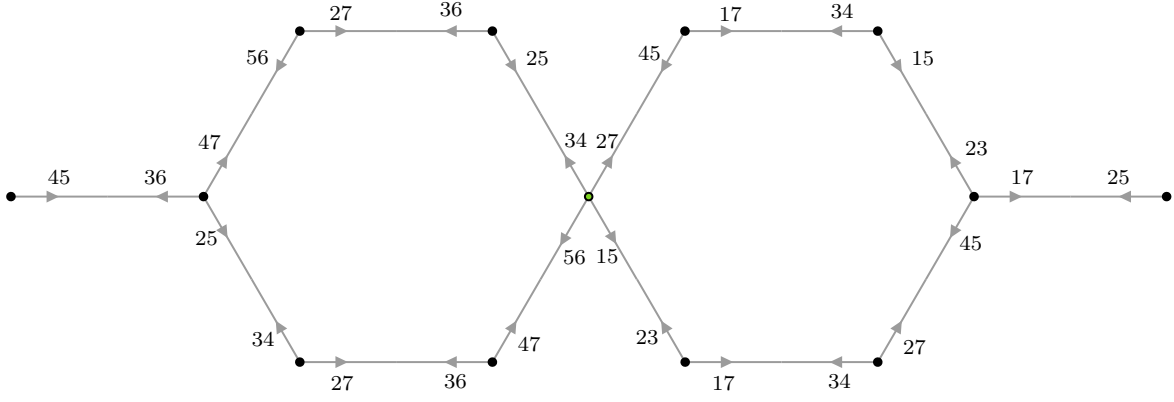}
\caption{Permutohedra graph for knot $7_1$, formed by two $\Pi_3$ glued at the center, and two outer $\Pi_2$ glued to each larger permutohedron.
We also show unlinking operations, which connect every two quivers
sharing the same edge of a permutohedron.}
\label{fig:permutohedron 7_1}
\end{figure}
The quiver 
\begin{equation}\label{eq:7_1Quiver}
\begin{split}
     C &=
    \left[
    \begin{array}{ccccccc}
         0 & 1 & 1 & 3 & 3 & 5 & 5 \\
         1 & 2 & 2 & 3 & 3 & 5 & 5 \\
         1 & 2 & 3 & 4 & 4 & 6 & 6 \\
         3 & 3 & 4 & 4 & 4 & 5 & 5 \\
         3 & 3 & 4 & 4 & 5 & 6 & 6 \\
         5 & 5 & 6 & 5 & 6 & 6 & 6 \\
         5 & 5 & 6 & 5 & 6 & 6 & 7 \\
    \end{array}
    \right], \\
    \boldsymbol{x} &= \left[a^6 q^{-6}, a^6 t^2 q^{-2}, a^8 t^3 q^{-4}, a^6 q^2 t^4, a^8 t^5, a^6 q^6 t^6, a^8 q^4 t^7\right],
\end{split}
\end{equation}
is the gluing node of two permutohedra $\Pi_3$, denoted as the central (green) node in Fig.~\ref{fig:permutohedron 7_1}. The prequivers 
\begin{align}
    \check{C}_1 &=
    \left[
    \begin{array}{cccc}
         0 & 1 & 3 & 5 \\
         1 & 2 & 3 & 5 \\
         3 & 3 & 4 & 5 \\
         5 & 5 & 5 & 6 \\
    \end{array}
    \right], &
    \check{C}_2 &=
    \left[
    \begin{array}{cccc}
         6 & 5 & 5 & 5 \\
         5 & 0 & 1 & 3 \\
         5 & 1 & 2 & 3 \\
         5 & 3 & 3 & 4 \\
    \end{array}
    \right], \\
    \boldsymbol{\check{x}}_1  & = \left[a^6q^{-6}, a^6 t^2q^{-4}, a^6 t^4q^{-2},a^6 t^6\right], &
    \boldsymbol{\check{x}}_2  & = \left[a^6 t^6, a^6q^{-6}, a^6 t^2 q^{-4}, a^6 t^4q^{-2}\right] \nonumber
\end{align}
correspond to the left and right $\Pi_3$ in Fig.~\ref{fig:permutohedron 7_1} respectively. To obtain quivers in the left $\Pi_3$, we perform a~$(0, 1)$-splitting of the three nodes with 2, 4 and 6 loops in $\check{C}_1$. Analogously, we get quivers in the right $\Pi_3$ by applying a~$(1, 3)$-splitting to the three nodes with 0, 2 and 4 loops in $\check{C}_2$. Additionally, there are two $\Pi_2$ glued to each of the $\Pi_3$ which also come from splitting, but we omit their prequivers for the sake of brevity. 

In the next step, we derive the universal quivers for each $\Pi_3$ following the steps from the proof of Theorem~\ref{thm:Universal quiver for splitting}. The resulting quivers are
\begin{equation}\label{eq:7_1 Left universal quiver}
    \begin{split}
        \hat{C}_1 = &
        \left[
        \begin{array}{ccccccc:ccccc}
             0 & 1 & 1 & 3 & 3 & 5 & 5 & 8 & 4 & 6 & 9 & 9 \\
             1 & 2 & 2 & 3 & 3 & 5 & 5 & 8 & 5 & 7 & 10 & 10 \\
             1 & 2 & 3 & 3 & 4 & 5 & 6 & 9 & 6 & 8 & 11 & 12 \\
             3 & 3 & 3 & 4 & 4 & 5 & 5 & 9 & 7 & 8 & 12 & 12 \\
             3 & 3 & 4 & 4 & 5 & 5 & 6 & 10 & 8 & 9 & 13 & 14 \\
             5 & 5 & 5 & 5 & 5 & 6 & 6 & 11 & 10 & 11 & 16 & 16 \\
             5 & 5 & 6 & 5 & 6 & 6 & 7 & 12 & 11 & 12 & 17 & 18 \\
             \hdashline
             8 & 8 & 9 & 9 & 10 & 11 & 12 & 22 & 19 & 21 & 30 & 31 \\
             4 & 5 & 6 & 7 & 8 & 10 & 11 & 19 & 14 & 17 & 24 & 25 \\
             6 & 7 & 8 & 8 & 9 & 11 & 12 & 21 & 17 & 20 & 28 & 29 \\
             9 & 10 & 11 & 12 & 13 & 16 & 17 & 30 & 24 & 28 & 41 & 41 \\
             9 & 10 & 12 & 12 & 14 & 16 & 18 & 31 & 25 & 29 & 41 & 44 \\
        \end{array}
        \right], \\
        \boldsymbol{\hat{x}}_1 = & \left[a^6q^{-6}, a^6 t^2q^{-2}, a^8 t^3q^{-4}, a^6 q^2 t^4, a^8 t^5, a^6 q^6 t^6, a^8 q^4 t^7 \dashline \right.\\
         & \left. a^{14} q^5 t^{11}, a^{14} t^7q^{-3}, a^{14} q t^9, a^{20} q^2 t^{13}, a^{22} t^{14}\right],
    \end{split}
\end{equation}
 for the left $\Pi_3$ and 
\begin{equation}\label{eq:7_1 Right universal quiver}
    \begin{split}
        \hat{C}_2 = &
        \left[
        \begin{array}{ccccccc:ccccc}
             0 & 1 & 1 & 3 & 2 & 5 & 4 & 2 & 5 & 4 & 5 & 6 \\
             1 & 2 & 2 & 3 & 3 & 5 & 4 & 4 & 6 & 5 & 7 & 8 \\
             1 & 2 & 3 & 4 & 4 & 6 & 6 & 5 & 8 & 7 & 9 & 11 \\
             3 & 3 & 4 & 4 & 4 & 5 & 5 & 7 & 8 & 8 & 11 & 12 \\
             2 & 3 & 4 & 4 & 5 & 6 & 6 & 7 & 9 & 8 & 11 & 13 \\
             5 & 5 & 6 & 5 & 6 & 6 & 6 & 11 & 11 & 11 & 16 & 17 \\
             4 & 4 & 6 & 5 & 6 & 6 & 7 & 10 & 11 & 11 & 15 & 17 \\
             \hdashline
             2 & 4 & 5 & 7 & 7 & 11 & 10 & 10 & 15 & 13 & 17 & 20 \\
             5 & 6 & 8 & 8 & 9 & 11 & 11 & 15 & 18 & 17 & 23 & 26 \\
             4 & 5 & 7 & 8 & 8 & 11 & 11 & 13 & 17 & 16 & 21 & 24 \\
             5 & 7 & 9 & 11 & 11 & 16 & 15 & 17 & 23 & 21 & 29 & 32 \\
             6 & 8 & 11 & 12 & 13 & 17 & 17 & 20 & 26 & 24 & 32 & 38 \\
        \end{array}
        \right], \\
        \boldsymbol{\hat{x}}_2 = & \left[a^6 q^{-6}, a^6 t^2 q^{-2}, a^8 t^3 q^{-4}, a^6 q^2 t^4, a^8 t^5, a^6 q^6 t^6, a^8 q^4 t^7 \dashline \right. \\
         & \left.a^{14} t^5 q^{-7}, a^{14} q t^9, a^{14} t^7 q^{-3}, a^{20} t^9 q^{-6}, a^{22} t^{12} q^{-4}\right],
    \end{split}
\end{equation}
for the right $\Pi_3$. 
The above universal quivers are represented as the yellow nodes centered on each $\Pi_3$ in Fig.~\ref{fig:Web unlinkings 7_1}. 

Moving to the next step, we select any two sequences of unlinking which transform $C$ into $\hat{C}_1$ and $\hat{C}_2$ respectively. As written in Remark~\ref{rmk:non-uniqueness of universal quiver}, different choices of such paths lead to different, but equivalent universal quivers. Our choice is presented in Fig.~\ref{fig:Web unlinkings 7_1} as two paths which start from the central gluing point and end at the center of the left and right permutohedron. Each path has length 5, so we need to construct the $(5,5)$-connector using the Connector Algorithm.
The resulting structure is presented in Fig.~\ref{fig:Sanctuary diagram 7_1}. The red node at the top of the structure represents the universal quiver for both permutohedra $\Pi_3$.
\newpage
\begin{figure}[h!]
\centering
\input{tikz/Pi_knot71.tikz}
\caption{Unlinkings of the permutohedra graph for $7_1$ with highlighted sequences that start at the center quiver (in green) and end at the universal quivers of each $\Pi_3$ (in yellow).}
\label{fig:Web unlinkings 7_1}
\end{figure}
\begin{figure}[h!]
\centering
\input{tikz/two_pi3.tikz}
\caption{Connecting the two universal quivers ($\hat{C}_1$, $\hat{C}_2$) for each permutohedron $\Pi_3$.}
\label{fig:Sanctuary diagram 7_1}
\end{figure}
\newpage
Finally, in order to obtain the universal quiver for the whole permutohedra graph, we need to take into account the two external $\Pi_2$. This is achieved by unlinking the top quiver in Fig.~\ref{fig:Sanctuary diagram 7_1} twice more 
--
once for the left $\Pi_2$ by applying $U(36)$ and once for the right by applying $U(17)$. The reason for this is the following.
Let us focus on the left half of the permutohedra graph, shown in Fig.~\ref{fig:knot71_pi_half}. We first have to connect the universal quiver for the left $\Pi_2$, which is given by applying $\boldsymbol{U_l}=U(36)$ to one of the quivers in~$\Pi_3$ (the red path).
On the other hand, starting from the same quiver and following the purple path $\boldsymbol{U_r}=U(2(47))U((25)7)U(27)U(25)U(47)$, we obtain universal quiver $\hat{C}_1$ for~$\Pi_3$. Thus, we first have to connect $\boldsymbol{U_r}$ and $\boldsymbol{U_l}$. Note that $U(36)$ commutes with every unlinking in $\boldsymbol{U_r}$, so the
$(1,5)$-connector is a simple application of square identities, and the resulting quiver is obtained from $\hat{C}_1$ by applying $U(36)$.
Analogously, on the right half of the permutohedra graph, we apply $U(17)$ to $\hat{C}_2$. Finally, notice that the two paths
in Fig.~\ref{fig:Sanctuary diagram 7_1} which start from $\hat{C}_1$, $\hat{C}_2$ and lead to the universal quiver for both $\Pi_3$, also do not contain unlinking of the nodes 1, 3, 6, 7. We thus conclude that the application of $U(36)U(17)$ (or $U(17)U(36)$, since they commute) connects the external legs and we obtain the universal quiver for the whole permutohedra graph for knot $7_1$:

\bigskip

\begin{equation*}\label{eq:7_1 Universal quiver}
\begin{aligned}
 \hat{C} &=  \left[\scalebox{0.8}{$
\begin{array}{ccccccccccccccccccccccc}
 0 & 1 & 1 & 3 & 2 & 5 & 3 & 2 & 5 & 4 & 5 & 6 & 7 & 7 & 11 & 4 & 6 & 9 & 8 & 8 & 12 & 6 & 3 \\
 1 & 2 & 2 & 3 & 3 & 5 & 4 & 4 & 6 & 5 & 7 & 8 & 8 & 9 & 13 & 5 & 7 & 10 & 10 & 11 & 15 & 7 & 5 \\
 1 & 2 & 3 & 3 & 4 & 4 & 6 & 5 & 8 & 7 & 9 & 11 & 9 & 10 & 16 & 6 & 8 & 11 & 12 & 13 & 19 & 7 & 7 \\
 3 & 3 & 3 & 4 & 4 & 5 & 5 & 7 & 8 & 8 & 11 & 12 & 9 & 12 & 17 & 7 & 8 & 12 & 12 & 15 & 20 & 8 & 8 \\
 2 & 3 & 4 & 4 & 5 & 5 & 6 & 7 & 9 & 8 & 11 & 13 & 10 & 12 & 18 & 8 & 9 & 13 & 14 & 16 & 22 & 9 & 8 \\
 5 & 5 & 4 & 5 & 5 & 6 & 6 & 10 & 11 & 11 & 16 & 16 & 11 & 16 & 22 & 10 & 11 & 16 & 16 & 21 & 27 & 10 & 11 \\
 3 & 4 & 6 & 5 & 6 & 6 & 7 & 10 & 11 & 11 & 15 & 17 & 12 & 16 & 23 & 11 & 12 & 17 & 18 & 22 & 29 & 12 & 10 \\
 2 & 4 & 5 & 7 & 7 & 10 & 10 & 10 & 15 & 13 & 17 & 20 & 18 & 20 & 30 & 12 & 15 & 22 & 23 & 25 & 35 & 15 & 12 \\
 5 & 6 & 8 & 8 & 9 & 11 & 11 & 15 & 18 & 17 & 23 & 26 & 20 & 26 & 37 & 16 & 19 & 27 & 28 & 34 & 45 & 19 & 16 \\
 4 & 5 & 7 & 8 & 8 & 11 & 11 & 13 & 17 & 16 & 21 & 24 & 19 & 24 & 35 & 15 & 18 & 26 & 26 & 31 & 42 & 18 & 15 \\
 5 & 7 & 9 & 11 & 11 & 16 & 15 & 17 & 23 & 21 & 29 & 32 & 27 & 33 & 48 & 20 & 25 & 36 & 36 & 42 & 57 & 25 & 20 \\
 6 & 8 & 11 & 12 & 13 & 16 & 17 & 20 & 26 & 24 & 32 & 38 & 30 & 37 & 54 & 23 & 27 & 39 & 41 & 48 & 65 & 27 & 23 \\
 7 & 8 & 9 & 9 & 10 & 11 & 12 & 18 & 20 & 19 & 27 & 30 & 22 & 29 & 41 & 19 & 21 & 30 & 31 & 38 & 50 & 20 & 19 \\
 7 & 9 & 10 & 12 & 12 & 16 & 16 & 20 & 26 & 24 & 33 & 37 & 29 & 37 & 53 & 23 & 27 & 39 & 40 & 47 & 63 & 26 & 23 \\
 11 & 13 & 16 & 17 & 18 & 22 & 23 & 30 & 37 & 35 & 48 & 54 & 41 & 53 & 77 & 34 & 39 & 56 & 58 & 70 & 93 & 38 & 34 \\
 4 & 5 & 6 & 7 & 8 & 10 & 11 & 12 & 16 & 15 & 20 & 23 & 19 & 23 & 34 & 14 & 17 & 24 & 25 & 29 & 40 & 16 & 15 \\
 6 & 7 & 8 & 8 & 9 & 11 & 12 & 15 & 19 & 18 & 25 & 27 & 21 & 27 & 39 & 17 & 20 & 28 & 29 & 35 & 47 & 19 & 18 \\
 9 & 10 & 11 & 12 & 13 & 16 & 17 & 22 & 27 & 26 & 36 & 39 & 30 & 39 & 56 & 24 & 28 & 41 & 41 & 50 & 67 & 27 & 26 \\
 8 & 10 & 12 & 12 & 14 & 16 & 18 & 23 & 28 & 26 & 36 & 41 & 31 & 40 & 58 & 25 & 29 & 41 & 44 & 52 & 70 & 28 & 26 \\
 8 & 11 & 13 & 15 & 16 & 21 & 22 & 25 & 34 & 31 & 42 & 48 & 38 & 47 & 70 & 29 & 35 & 50 & 52 & 61 & 83 & 34 & 30 \\
 12 & 15 & 19 & 20 & 22 & 27 & 29 & 35 & 45 & 42 & 57 & 65 & 50 & 63 & 93 & 40 & 47 & 67 & 70 & 83 & 113 & 46 & 41 \\
 6 & 7 & 7 & 8 & 9 & 10 & 12 & 15 & 19 & 18 & 25 & 27 & 20 & 26 & 38 & 16 & 19 & 27 & 28 & 34 & 46 & 18 & 18 \\
 3 & 5 & 7 & 8 & 8 & 11 & 10 & 12 & 16 & 15 & 20 & 23 & 19 & 23 & 34 & 15 & 18 & 26 & 26 & 30 & 41 & 18 & 14 \\
\end{array} $}
\right], \\
\\
\boldsymbol{\hat{x}} &=  \left[
 \scalebox{0.8}{$ a^6q^{-6},a^6  q^{-2}t^2 ,a^8  q^{-4} t^3,a^6 q^2 t^4,a^8 t^5,a^6 q^6 t^6,a^8 q^4 t^7, a^{14} q^{-7} t^5, a^{14} q t^9, a^{14} q^{-3} t^7,a^{20} q^{-6} t^9,a^{22} q^{-4}t^{12}, $} \right. \\ 
& \qquad \left.  \scalebox{0.8}{$ a^{14} q^5 t^{11}, a^{20} q^{-2}t^{11},a^{28} q t^{18},a^{14}, q^{-3}t^7, a^{14} q t^9,a^{20} q^2 t^{13},a^{22} t^{14},a^{28} q^{-7}t^{14},a^{36} q^{-4}t^{21},a^{14} q t^9,a^{14} q^{-3}t^7 $} \right].
    \end{aligned}
\end{equation*}

\begin{figure}[h!]
    \centering
    \input{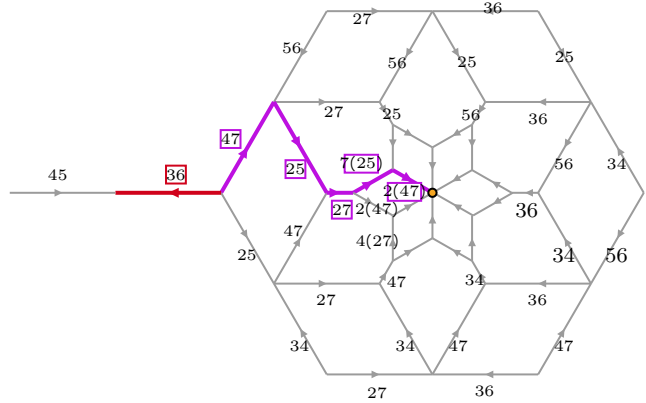}
    \caption{The left half of the permutoherdra graph for knot $7_1$ with the two marked paths.}
    \label{fig:knot71_pi_half}
\end{figure}

\newpage
\section{Future directions} \label{sec:future}
Below we provide a brief outline of interesting directions for future research.

\begin{itemize}
    \item The obvious generalisation of the results of this paper is to combine operations of unlinking and linking.
    However, it is not immediate to prove the analogous theorems (most importantly, the Connector Theorem) for such a mixed case. 
    Nevertheless, we still believe that our findings for unlinking also hold in this scenario, as we point out in Sec.~\ref{sec:linking}.
    Showing that mixed case works is very useful for a handful of reasons: for example, many quivers for knots or their complements have both positive and negative entries, and if one wants to prove the uniqueness of the diagonalisation procedure of \cite{Jankowski:2022qdp} in full generality, one can do it using the stronger version of Connector Theorem which includes linking. In turn, proving this mixed case will show that BPS spectrum for 3-manifold invariants which admit a quiver presentation is well defined from an infinite diagonal quiver, which would be one step towards understanding their categorification.
    \item As it turns out, the graphs of equivalent quivers coming from unlinking have quite an intricate structure, and one may wonder
    whether there are periodic patterns or other interesting formations one can find inside such graphs. For example, we expect the periodicity property
    for every permutohedra graph, but at the present stage we cannot even say exactly what is the structure of all universal quivers for a~given permutohedron. Understanding these patterns, combined with the connection to splitting, may shed some light on spectral
    sequences hidden behind quivers and prove a stronger version of the correspondence between a set of equivalent quivers and link homology.
    \item It is very interesting to understand the linking and unlinking operations in the context of vertex operator algebras, where the quiver matrix represents scattering of quasi-particles in an underlying CFT. Besides, linking operation is somewhat similar to the discrete analogue of RG flow for a family of rational CFTs \cite{Berkovich:1997ht}. Thus, the question is whether this is a coincidence, or can that be extended to other cases -- most prominently, logarithmic CFTs which are expected to be related to low-dimensional topology.
    \item Another important problem is to extend our case studies in the present paper from knot conormal to knot complement quivers
    \cite{EKL2108}. Study of unlinking for quivers and permutohedra graphs related to knot complements will improve the general understanding of their $F_K$ invariants and, saying it again, will help to define their homological version. We may also ask if the cohomological Hall algebra of a~symmetric quiver plays a role in this potential construction.
    \item Motivated by the relations to vertex operator algebras and modular forms, we may ask if quiver generating series can form a~ring structure -- i.e. can linear combinations of quiver generating series be rewritten as a~quiver generating series again? Note that the product of two quiver series for matrices $C_{(\boldsymbol{ij})}$ and $C_{(\boldsymbol{kl})}$ can be also written as a~single quiver in a~block diagonal form, so 
    it is closed under multiplication, but not manifestly under addition. As a prototype example, there are two completely different 
    fermionic formulas for the triplet log-VOA \cite{Cheng:2022rqr}, but one of them comes from a surgery on trefoil knot and is given
    by a linear combination of quiver generating series for the same quiver but in a different framing, while the other comes from an intrinsic quasi-particle description for this log-VOA.
\end{itemize}

\section*{Acknowledgements}

We thank Jakub Jankowski for collaboration in initial stages of this project. The work of P.S. has been supported by the OPUS grant no. 2022/47/B/ST2/03313 \emph{``Quantum geometry and BPS states''} funded by the National Science Centre, Poland. The work of P.K. has been supported by the Sonata grant no. 2022/47/D/ST2/02058 funded by the National Science Centre, Poland.

\newpage

\appendix

\section{Equivalent quivers and splitting for knots}
\label{sec:Statements from the permutohedra paper}

In this section we summarise the construction of the permutohedra graphs introduced in~\cite{JKLNS2105}.
The starting point was the following theorem:
\begin{thm}\label{thm:Local eqiuvalence}
Consider a~quiver $Q$ corresponding to the~knot $K$ and another quiver $Q'$ with the same number of vertices and the~same change of variables associated to the knots-quivers correspondence.
If $Q$ and $Q'$ are related by a~sequence of disjoint transpositions, each exchanging non-diagonal elements 
    \begin{equation}\label{eq:transposition}
         C_{ab}\leftrightarrow C_{cd}, \qquad C_{ba}\leftrightarrow C_{dc},
    \end{equation}
for some pairwise different $a,b,c,d,\in  Q_{0}$, such that
    \begin{equation}\label{eq:center of mass contition}
         x_{a} x_{b} = x_{c} x_{d}
    \end{equation}
and 
    \begin{equation}\label{eq:theorem first case}
        C_{ab} = C_{cd}-1,\qquad\quad
        C_{ai}+C_{bi}=C_{ci}+C_{di}-\delta_{ci}-\delta_{di},\quad \forall i\in Q_{0},
    \end{equation}
    or
    \begin{equation}\label{eq:theorem second case}
        C_{cd} = C_{ab}-1,\qquad\quad
        C_{ci}+C_{di}=C_{ai}+C_{bi}-\delta_{ai}-\delta_{bi},\quad \forall i\in Q_{0},
    \end{equation}
then the generating series of $Q$ and $Q'$ -- after the application of the knots-quivers change of variables -- are equal: $P_Q=P_{Q'}$.
\end{thm}

Using this statement, one can take any quiver corresponding to a given knot, check conditions (\ref{eq:center of mass contition}-\ref{eq:theorem second case}), and -- if they are satisfied -- generate different quivers corresponding to the same knot. 
In \cite{JKLNS2105} it was proven that the resulting quivers and transpositions that connect them can be represented by vertices and edges of a graph made of several permutohedra glued together. It was done using the procedure of splitting:

\begin{dfn}\label{def:splitting}
    A~$(k,l)$-splitting of $n$ nodes with permutation $\sigma\in S_n$ in the~presence of $m-2n$ spectators (with corresponding integer shifts $h_s$) and with a~multiplicative factor $\kappa$ is defined as the~following transformation of a~quiver $(\check{C},\check{\boldsymbol{x}})$ with
    \begin{equation}
        \check{C}=\left[\begin{array}{c:c:c:c:c}
\check{C}_{ss} & \cdots  & \check{C}_{si} & \cdots  & \check{C}_{sj}\\ \hdashline
\vdots  & \ddots  & \vdots  &  & \vdots \\ \hdashline
\check{C}_{is} & \cdots  & \check{C}_{ii} & \cdots  & \check{C}_{ij}\\ \hdashline
\vdots  &  & \vdots  & \ddots  & \vdots \\ \hdashline
\check{C}_{js} & \cdots  & \check{C}_{ji} & \cdots  & \check{C}_{jj}
\end{array}\right]\,.
    \end{equation}
    For any two split nodes $i$ and $j$, $i<j$, and any spectator $s$, we transform the~matrix $\check{C}$ in the~following way (depending on the~presence of inversion in permutation $\sigma$): either into
    \begin{equation}
            \left[\begin{array}{ c : c : c : c : c : c : c }
\check{C}_{ss} & \cdots  & \check{C}_{si} & \check{C}_{si} +h_{s} & \cdots  & \check{C}_{sj} & \check{C}_{sj} +h_{s}\\ \hdashline
\vdots  & \ddots  & \vdots  & \vdots  &  & \vdots  & \vdots \\ \hdashline
\check{C}_{is} & \cdots  & \check{C}_{ii} & \check{C}_{ii} +k & \cdots  & \check{C}_{ij} & \textcolor[rgb]{0.96,0.65,0.14}{\check{C}_{ij} +k}\\ \hdashline
\check{C}_{is} +h_{s} & \cdots  & \check{C}_{ii} +k & \check{C}_{ii} +l & \cdots  & \textcolor[rgb]{0.96,0.65,0.14}{\check{C}_{ij} +k+1} & \check{C}_{ij} +l\\ \hdashline
\vdots  &  & \vdots  & \vdots  & \ddots  & \vdots  & \vdots \\ \hdashline
\check{C}_{js} & \cdots  & \check{C}_{ji} & \textcolor[rgb]{0.96,0.65,0.14}{\check{C}_{ji} +k+1} & \cdots  & \check{C}_{jj} & \check{C}_{jj} +k\\ \hdashline
\check{C}_{js} +h_{s} & \cdots  & \textcolor[rgb]{0.96,0.65,0.14}{\check{C}_{ji} +k} & \check{C}_{ji} +l & \cdots  & \check{C}_{jj} +k & \check{C}_{jj} +l
\end{array}\right]
    \end{equation}
    for $\sigma(i)<\sigma(j)$, or
    \begin{equation}
      \left[\begin{array}{c:c:c:c:c:c:c}
\check{C}_{ss} & \cdots  & \check{C}_{si} & \check{C}_{si} +h_{s} & \cdots  & \check{C}_{sj} & \check{C}_{sj} +h_{s}\\ \hdashline
\vdots  & \ddots  & \vdots  & \vdots  &  & \vdots  & \vdots \\ \hdashline
\check{C}_{is} & \cdots  & \check{C}_{ii} & \check{C}_{ii} +k & \cdots  & \check{C}_{ij} & \textcolor[rgb]{0.96,0.65,0.14}{\check{C}_{ij} +k+1}\\ \hdashline
\check{C}_{is} +h_{s} & \cdots  & \check{C}_{ii} +k & \check{C}_{ii} +l & \cdots  & \textcolor[rgb]{0.96,0.65,0.14}{\check{C}_{ij} +k} & \check{C}_{ij} +l\\ \hdashline
\vdots  &  & \vdots  & \vdots  & \ddots  & \vdots  & \vdots \\ \hdashline
\check{C}_{js} & \cdots  & \check{C}_{ji} & \textcolor[rgb]{0.96,0.65,0.14}{\check{C}_{ji} +k} & \cdots  & \check{C}_{jj} & \check{C}_{jj} +k\\ \hdashline
\check{C}_{js} +h_{s} & \cdots  & \textcolor[rgb]{0.96,0.65,0.14}{\check{\textcolor[rgb]{0.96,0.65,0.14}{C}}_{ji} +k+1} & \check{C}_{ji} +l & \cdots  & \check{C}_{jj} +k & \check{C}_{jj} +l
\end{array}\right]  
    \end{equation}
    for $\sigma(i)>\sigma(j)$,
whereas for any permutation the~generating parameters are transformed as follows:
\begin{equation*}
    \left[\begin{array}{c}
    \check{x}_s\\
    \vdots\\
    \check{x}_i\\
    \vdots\\
    \check{x}_j
    \end{array}\right]
    \longrightarrow
    \left[\begin{array}{c}
    \check{x}_s\\
    \vdots\\
    \check{x}_i\\
    \check{x}_i\kappa\\
    \vdots\\
    \check{x}_j\\
    \check{x}_j\kappa
    \end{array}\right]\,.
\end{equation*}
\end{dfn}

\begin{dfn}\label{def:prequiver}
    If the~inverse of splitting -- for any parameters from Definition~\ref{def:splitting} --  can be applied to a~given quiver $(C,\boldsymbol{x})$, we call the~target of this operation a~prequiver $(\check{C},\check{ \boldsymbol{x}})$. Conversely, splitting the~nodes of a~prequiver produces the~quiver:
\begin{equation}
    (\check{C},\boldsymbol{\check{x}})
    \;\longrightarrow\; (C,\boldsymbol{x})\,.
    \end{equation}
    \end{dfn}

In \cite{JKLNS2105} it was shown that in the permutohedra graph composed of quivers satisfying conditions (\ref{eq:center of mass contition}-\ref{eq:theorem second case}) each permutohedron arises from splitting of some prequiver.


\section{Unlinking as a multi-cover skein relation}
\label{sec:Statements from the multi-cover skein paper}

As discussed in \cite{EKL1910}, the relation between quivers (and their motivic Donaldson-Thomas invariants) and open topological strings (with respective open Gromov-Witten invariants) allows for interpretation of unlinking (as well as linking) operation as a multi-cover skein relation. 

Taking into account the duality between topological strings and Chern-Simons theory, the following setup defines the knots-quivers correspondence in the language of geometry~\cite{EKL1811}.
Recall that the theory associated to the knot $K$ arises from the M-theory on the resolved conifold~$X$ with a single M5-brane wrapping the conormal Lagrangian of the knot $L_K$: 
\begin{equation}\label{eq:M-theory-setup}
\begin{split}
	\text{space-time}: \quad& \mathbb{R}^4 \times S^1 \times X \\
			& \cup \phantom{ \ \times S^1 \times \ \ } \cup\\
	\text{M5}: \quad & \mathbb{R}^2\times S^1 \times L_K.
\end{split}
\end{equation}
In this setup, motivic generating series (\ref{Z}) encodes the spectrum of BPS particles originating from M2-branes ending on the M5 brane. From the symplectic geometric point of view, BPS states correspond to generalized holomorphic curves with boundary on the Lagrangian submanifold $L_K$ (let us stress that the M-theory setup \eqref{eq:M-theory-setup} with more general Lagrangian submanifolds can be applied to quivers that do not correspond to knots).
Furthermore, it is assumed that all holomorphic curves are obtained from combinations of branched covers of finitely many basic discs, which are in one-to-one correspondence
with the quiver nodes. 
This allows to view unlinking operation as an operation on basic disks, which skein the pair of disks through each other, creating a~new disk and changing their linking numbers. For example, the simplest case of unlinking
\begin{equation}
\label{eq:link-removal}
	C = \left[
	\begin{array}{cc}
		0 & 1 \\
		1 & 0
	\end{array}
	\right] 
	\quad
	\rightsquigarrow \quad
	U(12)C = \left[
	\begin{array}{ccc}
		0 & 0 & 0\\
		0 & 0 & 0\\
		0 & 0 & 1
	\end{array}
	\right]\,
\end{equation}
can be understood as a skein relation shown in Fig.~\ref{fig:multi-cover skein relation} \cite{EKL1910}:

\begin{figure}[h!]
    \centering
    \includegraphics[width=0.7\textwidth]{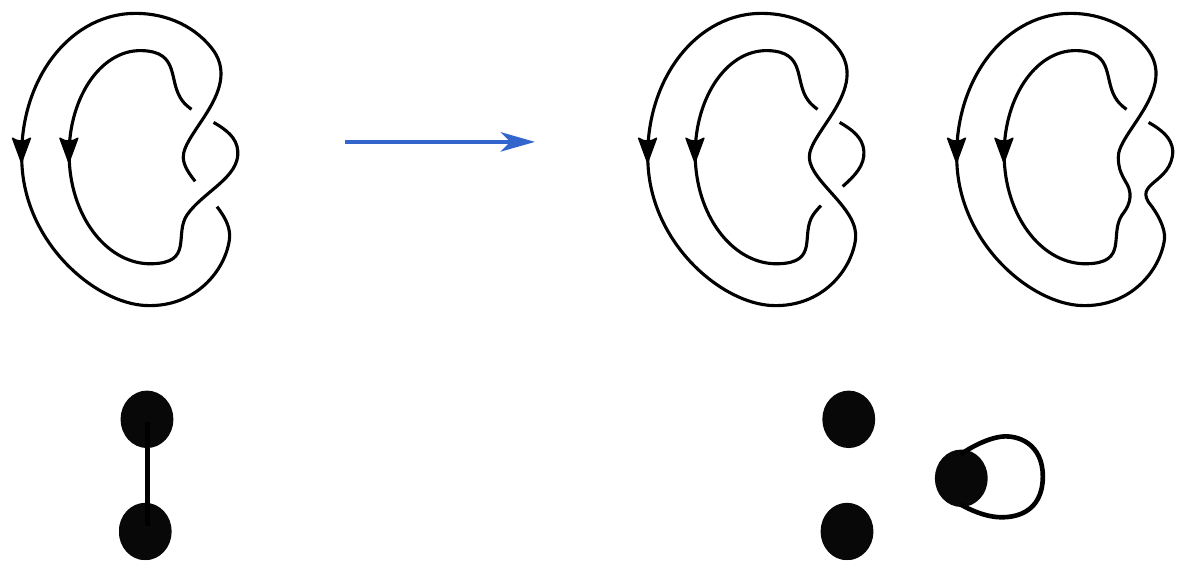}
    \caption{Multi-cover skein relations on linking disks and dual quiver description.}
    \label{fig:multi-cover skein relation}
\end{figure}

From the physical perspective, it corresponds to expressing the same BPS spectrum in two different ways -- on the left we have two basic states and one boundstate, whereas on the right we have have three basic states and no interaction. The equality between corresponding generating series can be interpreted as duality between two 3d $\mathcal{N}=2$ theories \cite{EKL1910}. 

Geometric and physical interpretation of linking is more involved than the case of unlinking. It involves the introduction of the redundant pair of nodes, for which the contributions from each disk cancel each other. The crucial property of linking is that its composition with unlinking leads to a redundant pair of nodes. It also shows that linking is not an inverse of unlinking, since it also increases the number of nodes, but it plays a very similar role. More detailed explanations can be found in \cite{EKL1910}.

Having said the above, it would be very interesting to understand a multi-cover skein interpretation of the relations (\ref{eq:commutators}-\ref{eq:commutators3}), as well
as other results found in this paper, in the language of curve counts.

\section{Proof of unlinking identities}\label{app:hexagon}

Below we provide a~detailed proof of identities (\ref{eq:square})-(\ref{eq:second hexagon}). Without loss of generality, consider quiver with nodes $i,j,k,l,s$, where $s$ plays a~role of the spectator, i.e. it does not participate in any unlinking. We have $Q=(C,\boldsymbol{x})$, where\footnote{In this appendix it is convenient to use a column vectors instead of row ones.}

\begin{equation}
C = \left[ \begin {array}{ccccc}
C_{{ii}}&C_{{ij}}&C_{{ik}}&C_{{il}}&C_{is}
\\ \noalign{\medskip}&C_{{jj}}&C_{{jk}}&C_{{jl}}&C_{js}
\\ \noalign{\medskip}&&C_{{kk}}&C_{{kl}}&C_{ks}
\\ \noalign{\medskip}&&&C_{{ll}}&C_{ls}
\\ \noalign{\medskip}&&&&C_{ss}
\end {array}
 \right],
 \quad 
 \boldsymbol{x} = \left[ \begin {array}{c} x_{{i}}\\ \noalign{\medskip}x_{{j}}
\\ \noalign{\medskip}x_{{k}} \\ \noalign{\medskip}x_{{l}}
\\ \noalign{\medskip}x_{{s}}\end {array}
 \right].
\end{equation}
Note that the square relation (\ref{eq:square}) transforms the vectors of variables (assuming that the variable for newly created node is appended at the bottom) as
\begin{equation*}\label{eq:second hexagon lhs transforming vector2}
\left[ \begin {array}{c} x_{{i}}\\ \noalign{\medskip}x_{{j}}
\\ \noalign{\medskip}x_{{k}}\\ \noalign{\medskip}x_{{l}} \\ \noalign{\medskip}x_{{s}}\end {array}
 \right] \ \xrightarrow{U(ij)} \  \left[ \begin {array}{c} x_{{i}}\\ \noalign{\medskip}x_{{j}
}\\ \noalign{\medskip}x_{{k}}\\ \noalign{\medskip}x_{{l}} \\ \noalign{\medskip}x_{{s}}
\\ \noalign{\medskip}{\frac {x_{{i}}x_{{j}}}{q}}\end {array} \right] \ \xrightarrow{U(kl)} \ 
 \left[ \begin {array}{c} x_{{i}}\\ \noalign{\medskip}x_{{j}}
\\ \noalign{\medskip}x_{{k}}\\ \noalign{\medskip}x_{{l}} \\ \noalign{\medskip}x_{{s}}
\\ \noalign{\medskip}{\frac{x_ix_j}{q}}\\ \noalign{\medskip}{
\frac {x_{{k}}x_{{l}}}{{q}}}\end {array} \right]\,,
\qquad
\left[ \begin {array}{c} x_{{i}}\\ \noalign{\medskip}x_{{j}}
\\ \noalign{\medskip}x_{{k}}\\ \noalign{\medskip}x_{{l}} \\ \noalign{\medskip}x_{{s}}\end {array}
 \right] \ \xrightarrow{U(kl)} \  \left[ \begin {array}{c} x_{{i}}\\ \noalign{\medskip}x_{{j}
}\\ \noalign{\medskip}x_{{k}}\\ \noalign{\medskip}x_{{l}} \\ \noalign{\medskip}x_{{s}}
\\ \noalign{\medskip}{\frac {x_{{k}}x_{{l}}}{q}}\end {array} \right] \ \xrightarrow{U(ij)} \ 
 \left[ \begin {array}{c} x_{{i}}\\ \noalign{\medskip}x_{{j}}
\\ \noalign{\medskip}x_{{k}}\\ \noalign{\medskip}x_{{l}} \\ \noalign{\medskip}x_{{s}}
\\ \noalign{\medskip}{\frac{x_kx_l}{q}}\\ \noalign{\medskip}{
\frac {x_{{i}}x_{{j}}}{{q}}}\end {array} \right].
\end{equation*}
Moreover, one can see that the left and right hand sides of (\ref{eq:square}) give equivalent matrices which differ by permutation of rows and columns:
\begin{equation*}
\left[
\begin{array}{ccccccc}
C_{ii} &C_{ij}-1 &C_{ik} &C_{il} &C_{is} &C_{ii}+C_{ij}-1 &C_{ik}+C_{il} \\
&C_{jj} &C_{jk} &C_{jl} &C_{js} &C_{ij}+C_{jj}-1 &C_{jk}+C_{jl} \\
& &C_{kk} &C_{kl}-1 &C_{ks} &C_{ik}+C_{jk} &C_{kk}+C_{kl}-1 \\
& & &C_{ll} &C_{ls} &C_{il}+C_{jl} &C_{kl}+C_{ll}-1 \\
& & & &C_{ss} &C_{is}+C_{js} &C_{ks}+C_{ls} \\
& & & & &C_{ii}+2C_{ij}+C_{jj}-1 &C_{ik}+C_{il}+C_{jk}+C_{jl} \\
& & & & & &C_{kk}+2C_{kl}+C_{ll}-1 \\
\end{array}
\right]\,,
\end{equation*}
\begin{equation*}
    \left[
\begin{array}{ccccccc}
 C_{ii} & C_{ij}-1 & C_{ik} & C_{il} & C_{is} & C_{ik}+C_{il} & C_{ii}+C_{ij}-1 \\
   & C_{jj} & C_{jk} & C_{jl} & C_{js} & C_{jk}+C_{jl} & C_{ij}+C_{jj}-1 \\
   &   & C_{kk} & C_{kl}-1 & C_{ks} & C_{kk}+C_{kl}-1 & C_{ik}+C_{jk} \\
   &   &   & C_{ll} & C_{ls} & C_{kl}+C_{ll}-1 & C_{il}+C_{jl} \\
   &   &   &   & C_{ss} & C_{ks}+C_{ls} & C_{is}+C_{js} \\
   &   &   &   &   & C_{kk}+2 C_{kl}+C_{ll}-1 & C_{ik}+C_{il}+C_{jk}+C_{jl} \\
   &   &   &   &   &   & C_{ii}+2 C_{ij}+C_{jj}-1 \\
\end{array}
\right].
\end{equation*}
In turn, this confirms square relation (\ref{eq:square}).

For hexagon relations (\ref{eq:first hexagon})-(\ref{eq:second hexagon}) it is sufficient to consider quiver with nodes $i,j,k,s$:
\begin{equation}
C = \left[ \begin {array}{cccc}
                     C_{{ii}}&C_{{ij}}&C_{{ik}}&C_{{is}}
\\ \noalign{\medskip}&C_{{jj}}&C_{{jk}}&C_{{js}}
\\ \noalign{\medskip}&&C_{{kk}}&C_{{ks}}
\\ \noalign{\medskip}&&&C_{{ss}}
\end {array}
 \right]\,,
 \quad 
 \boldsymbol{x} = \left[ \begin {array}{c} x_{{i}}\\ \noalign{\medskip}x_{{j}}
\\ \noalign{\medskip}x_{{k}}\\ \noalign{\medskip}x_{{s}}\end {array}
 \right].
\end{equation}
Note that the left hand side of hexagon (\ref{eq:second hexagon}) transforms the vector of variables as
\begin{equation}\label{eq:second hexagon lhs transforming vector}
\left[ \begin {array}{c} x_{{i}}\\ \noalign{\medskip}x_{{j}}
\\ \noalign{\medskip}x_{{k}}\\ \noalign{\medskip}x_{{s}}\end {array}
 \right] \ \xrightarrow{U(ij)} \  \left[ \begin {array}{c} x_{{i}}\\ \noalign{\medskip}x_{{j}
}\\ \noalign{\medskip}x_{{k}}\\ \noalign{\medskip}x_{{s}}
\\ \noalign{\medskip}{\frac {x_{{i}}x_{{j}}}{q}}\end {array} \right] \ \xrightarrow{U((ij)k)} \ 
 \left[ \begin {array}{c} x_{{i}}\\ \noalign{\medskip}x_{{j}}
\\ \noalign{\medskip}x_{{k}}\\ \noalign{\medskip}x_{{s}}
\\ \noalign{\medskip}{\frac {x_{{i}}x_{{j}}}{q}}\\ \noalign{\medskip}{
\frac {x_{{i}}x_{{j}}x_{{k}}}{{q}^{2}}}\end {array} \right] \ \xrightarrow{U(jk)} \  \left[ 
\begin {array}{c} x_{{i}}\\ \noalign{\medskip}x_{{j}}
\\ \noalign{\medskip}x_{{k}}\\ \noalign{\medskip}x_{{s}}
\\ \noalign{\medskip}{\frac {x_{{i}}x_{{j}}}{q}}\\ \noalign{\medskip}{
\frac {x_{{i}}x_{{j}}x_{{k}}}{{q}^{2}}}\\ \noalign{\medskip}{\frac {x_
{{j}}x_{{k}}}{q}}\end {array} \right]\,,
\end{equation}
whereas the right hand side leads to
\begin{equation}\label{eq:second hexagon rhs transforming vector}
\left[ \begin {array}{c} x_{{i}}\\ \noalign{\medskip}x_{{j}}
\\ \noalign{\medskip}x_{{k}}\\ \noalign{\medskip}x_{{s}}\end {array}
 \right] \ \xrightarrow{U(jk)} \  \left[ \begin {array}{c} x_{{i}}\\ \noalign{\medskip}x_{{j}
}\\ \noalign{\medskip}x_{{k}}\\ \noalign{\medskip}x_{{s}}
\\ \noalign{\medskip}{\frac {x_{{j}}x_{{k}}}{q}}\end {array} \right] \ \xrightarrow{U(i(jk))} \ 
 \left[ \begin {array}{c} x_{{i}}\\ \noalign{\medskip}x_{{j}}
\\ \noalign{\medskip}x_{{k}}\\ \noalign{\medskip}x_{{s}}
\\ \noalign{\medskip}{\frac {x_{{j}}x_{{k}}}{q}}\\ \noalign{\medskip}{
\frac {x_{{i}}x_{{j}}x_{{k}}}{{q}^{2}}}\end {array} \right] \ \xrightarrow{U(ij)} \  \left[ 
\begin {array}{c} x_{{i}}\\ \noalign{\medskip}x_{{j}}
\\ \noalign{\medskip}x_{{k}}\\ \noalign{\medskip}x_{{s}}
\\ \noalign{\medskip}{\frac {x_{{j}}x_{{k}}}{q}}\\ \noalign{\medskip}{
\frac {x_{{i}}x_{{j}}x_{{k}}}{{q}^{2}}}\\ \noalign{\medskip}{\frac {x_
{{i}}x_{{j}}}{q}}\end {array} \right]\,. 
\end{equation}
This means that the only possibility for (\ref{eq:second hexagon}) to hold is when the two resulting quiver matrices differ by a~permutation
$(ij)\leftrightarrow (jk)$, which means (in particular) that node $(i(jk))$ must be identical to node $((ij)k)$. We will see that this is indeed the case by a~direct computation.

Application of the left hand side of hexagon (\ref{eq:second hexagon}) gives
\begin{equation}\label{eq:hexagon matrix left}
\left[ \begin {array}{ccccccc}
C_{{ii}}&C_{{ij}}-1&C_{{ik}}&C_{{is}}&C_{i,(ij)}&C_{i,((ij)k)}&C_{i,(jk)}\\
\noalign{\medskip}&C_{{jj}}&C_{{jk}}-1&C_{{js}}&C_{j,(ij)}&C_{j,((ij)k)}&C_{j,(jk)}\\
\noalign{\medskip}&&C_{{kk}}&C_{{ks}}&C_{k,(ij)}&C_{k,((ij)k)}&C_{k,(jk)}\\
\noalign{\medskip}&&&C_{{ss}}&C_{{is}}+C_{{js}}&C_{{is}}+C_{{js}}+C_{ks}&C_{{js}}+C_{{ks}}\\
\noalign{\medskip}&&&&C_{(ij),(ij)}&C_{(ij),((ij)k)}&C_{(ij),(jk)}\\
\noalign{\medskip}&&&&&C_{((ij)k),((ij)k)}&C_{((ij)k),(jk)}\\ 
\noalign{\medskip}&&&&&&C_{(jk),(jk)} \end {array} \right],
\end{equation}
whereas the right hand side of (\ref{eq:second hexagon}) leads to 
\begin{equation}\label{eq:hexagon matrix right}
\left[ \begin {array}{ccccccc} C_{{ii}}&C_{{ij}}-1&C_{{ik}}&C_{{is}}&{\it C}_{{i,(jk)}}&{\it C}_{{i,(i(jk))}}&{\it C}_{{i,(ij)}}\\
\noalign{\medskip}&C_{{jj}}&C_{{jk}}-1&C_{{js}}&{\it C}_{{j,(jk)}}&{\it C}_{{j,(i(jk))}}&{\it C}_{{j,(ij)}} \\
\noalign{\medskip}&&C_{{kk}}&C_{{ks}}&{\it C}_{{k,(jk)}}&{\it C}_{{k,(i(jk))}}&{\it C}_{{k,(ij)}}\\
\noalign{\medskip}&&&C_{{ss}}&C_{{js}}+C_{{ks}}&C_{{ks}}+C_{{is}}+C_{{js}}&C_{{is}}+C_{{js}}\\
\noalign{\medskip}&&&&{\it C}_{{(jk),(jk)}}&{\it C}_{{(i(jk)),(jk)}}&{\it C}_{{(ij),(jk)}}\\
\noalign{\medskip}&&&&&{\it C}_{{(i(jk)),(i(jk))}}&{\it 
C}_{{(ij),(i(jk))}}\\
\noalign{\medskip}&&&&&&{\it C}_{{(ij),(ij)}}\end {array} \right]\,. 
\end{equation}
For both matrices, we have
\begin{equation}\label{eq:hexagon matrix coefficients}
    \begin{aligned}
        C_{(\alpha\beta),(\alpha\beta)} &= C_{\alpha\alpha}+C_{\beta\beta}+2C_{\alpha\beta}-1 \quad \text{for $\alpha\beta=ij$ and $\alpha\beta=jk$}, \\
        C_{((ij)k),((ij)k)} &= C_{(i(jk)),(i(jk))} = C_{ii}+C_{jj}+C_{kk}+2(C_{ij}+C_{ik}+C_{jk}-1), \\
        C_{(ij),((ij)k)} &= C_{(ij),(i(jk))} = C_{ii}+C_{jj}+C_{jk}+2(C_{ij}-1), \\
        C_{(ij),(jk)} &= C_{jj}+C_{ij}+C_{ik}+C_{jk}-2, \\
        C_{((ij)k),(jk)} &= C_{(i(jk)),(jk)}  = C_{jj}+C_{kk}+C_{ij}+C_{ik}+2(C_{jk}-1), \\
        C_{\alpha,(\beta\gamma)} &= C_{\beta\alpha}+C_{\gamma\alpha}-1 \quad \text{for $\alpha\in\{i,j,k\}$, $\beta\gamma=ij$ and $\beta\gamma=jk$}, \\
        C_{\alpha,((ij)k)} &= C_{\alpha,(i(jk))} =  C_{\alpha i}+C_{\alpha j}+C_{\alpha k}-1 \quad \text{for $\alpha \in \{i,j,k\}$ }.
    \end{aligned}
\end{equation}
As a~result, matrices (\ref{eq:hexagon matrix left}) and (\ref{eq:hexagon matrix right}) differ by a~single transposition of the row and column corresponding to $(ij)$ and $(jk)$, which in turn confirms relation (\ref{eq:second hexagon}).

Likewise, for another hexagon (\ref{eq:first hexagon}) we obtain the vectors of variables as transpositions of the rightmost vectors in (\ref{eq:second hexagon lhs transforming vector}) and (\ref{eq:second hexagon rhs transforming vector}). The left hand side of (\ref{eq:first hexagon}) gives
\begin{equation}
\left[
\begin{array}{ccccccc}
 C_{ii} & C_{ij}-1 & C_{ik} & C_{is} & C_{i,(ij)} & C_{i,(jk)} & C_{i,((ij)k)} \\
  & C_{jj} & C_{jk}-1 & C_{js} & C_{j,(ij)} & C_{j,(jk)} & C_{j,((ij)k)} \\
  &  & C_{kk} & C_{ks}-1 & C_{k,(ij)} & C_{k,(jk)} & C_{k,((ij)k)} \\
  &  &  & C_{ss} & C_{is}+C_{js} & C_{js}+C_{ks} & C_{is}+C_{js}+C_{ks} \\
  &  &  &  & C_{(ij),(ij)} & C_{(ij),(jk)} & C_{(ij),((ij)k)} \\
  &  &  &  &  & C_{(jk),(jk)} & C_{(jk),((ij)k)} \\
  &  &  &  &  &  & C_{((ij)k),((ij)k)} \\
\end{array}
\right] \,,
\end{equation}
whereas the right hand side leads to
\begin{equation}
\left[
\begin{array}{ccccccc}
 C_{ii} & C_{ij}-1 & C_{ik} & C_{is} & C_{i,(jk)} & C_{i,(ij)} & C_{i,(i(jk))} \\
  & C_{jj} & C_{jk}-1 & C_{js} & C_{j,(jk)} & C_{j,(ij)} & C_{j,(i(jk))} \\
  &  & C_{kk} & C_{ks}-1 & C_{k,(jk)} & C_{k,(ij)} & C_{k,(i(jk))} \\
  &  &  & C_{ss} & C_{js}+C_{ks} & C_{is}+C_{js} & C_{is}+C_{js}+C_{ks} \\
  &  &  &  & C_{(jk),(jk)} & C_{(ij),(jk)} & C_{(jk),(i(jk))} \\
  &  &  &  &  & C_{(ij),(ij)} & C_{(ij),(i(jk))} \\
  &  &  &  &  &  & C_{(i(jk)),(i(jk))} \\
\end{array}
\right]\,.
\end{equation}
For both matrices, we have
\begin{equation}\label{eq:hexagon matrix coefficients2}
    \begin{aligned}
        C_{(\alpha\beta),(\alpha\beta)} &= C_{\alpha\alpha}+C_{\beta\beta}+2C_{\alpha\beta}-1 \quad \text{for $\alpha\beta=ij$ and $\alpha\beta=jk$}, \\
        C_{((ij)k),((ij)k)} &= C_{(i(jk)),(i(jk))} = C_{ii}+C_{jj}+C_{kk}+2(C_{ij}+C_{ik}+C_{jk}-1), \\
        C_{(ij),((ij)k)} &= C_{(jk),(i(jk))} = C_{ii}+C_{jj}+C_{jk}+2(C_{ij}-1), \\
        C_{(ij),(jk)} &= C_{jj}+C_{ij}+C_{ik}+C_{jk}-1, \\
        C_{((ij)k),(jk)} &= C_{(i(jk)),(jk)}  = C_{jj}+C_{kk}+C_{ij}+C_{ik}+2(C_{jk}-1), \\
        C_{\alpha,(\beta\gamma)} &= C_{\beta\alpha}+C_{\gamma\alpha}-1 \quad \text{for $\alpha\in\{i,j,k\}$, $\beta\gamma=ij$ and $\beta\gamma=jk$}, \\
        C_{\alpha,((ij)k)} &= C_{\alpha,(i(jk))} =  C_{\alpha i}+C_{\alpha j}+C_{\alpha k}-1 \quad \text{for $\alpha \in \{i,k\}$ }, \\
        C_{j,((ij)k)} &= C_{j,(i(jk))} =  C_{\alpha i}+C_{\alpha j}+C_{\alpha k}-2.
    \end{aligned}
\end{equation}
As the result, the two matrices (\ref{eq:hexagon matrix left}) and (\ref{eq:hexagon matrix right}) differ by a~single transposition of the row and column corresponding to $(ij)$ and $(jk)$, which in turn confirms relation (\ref{eq:first hexagon}).

Note, however, that the resulting quivers corresponding to (\ref{eq:first hexagon}) and (\ref{eq:second hexagon}) are not quite the same:
they differ by a~transposition of a~pair of arrows between $\{j,((ij)k)\}$ and $\{(ij),(jk)\}$ (which can be eliminated by applying unlinking once again).

\bibliography{refs}
\bibliographystyle{JHEP}

\end{document}